\documentclass[twoside,12pt]{article} 
\usepackage[margin=1.0in]{geometry}

\usepackage{amssymb,amsmath,amsthm}
\usepackage{mathtools}
\usepackage{bm,lscape}
\usepackage{bbm}
\usepackage{algpseudocode}
\usepackage[ruled,vlined]{algorithm2e}
\usepackage{mathrsfs,dsfont}
\usepackage{wasysym}
\RequirePackage[
    colorlinks=true,
    linkcolor=blue,
    urlcolor=blue,
    citecolor=blue]{hyperref}
\RequirePackage{natbib}
\usepackage[utf8]{inputenc}
\usepackage[english]{babel}

\usepackage{graphicx}
\usepackage{color}
\usepackage{enumitem}
\usepackage{array}
\usepackage{rotating}
\usepackage{authblk}
\usepackage{appendix}

\allowdisplaybreaks

\def\ab{\mathbf{a}}

\def\xb{\mathbf{x}}

\def\vb{\mathbf{v}}

\def\ub{\mathbf{u}}

\def\thetab{\boldsymbol{\theta}}

\def\simind{\stackrel{\mbox{\scriptsize{ind}}}{\sim}}
\def\simiid{\stackrel{\mbox{\scriptsize{iid}}}{\sim}}

\def\thetab{{\boldsymbol \theta}}

\def\vb{\mathbf{v}}
\def\xb{\mathbf{x}}
\def\kb{\mathbf{k}}

\def\gb{\mathbf{g}}
\def\hb{\mathbf{h}}
\def\rb{\mathbf{r}}

\def\Mb{\mathbf{M}}
\def\Hb{\mathbf{H}}
\def\epsilonb{{\boldsymbol\epsilon}}

\def\ab{\mathbf{a}}

\def\vb{\mathbf{v}}

\def\ub{\mathbf{u}}

\def\thetab{\boldsymbol{\theta}}

\allowdisplaybreaks

\def\Xb{\mathbf{X}}
\def\Sb{\mathbf{S}}

\newcommand{\Pn}{\mathbb{P}^{(n)}}

\newcommand{\Z}{\mathbb{Z}}

\newcommand{\R}{\mathbb{R}}
\newcommand{\E}{\mathds{E}}

\newcommand{\N}{\mathbb{N}}
\renewcommand{\P}{\mathds{P}}

\newcommand{\PP}{\mathds{P}}
\newcommand{\ddr}{\mathrm{d}}
\newcommand{\EE}{\mathds{E}}

\newcommand{\indicator}{\mathrm{I}}

\DeclareMathOperator*{\argmax}{arg\,max}

\newtheorem{thm}{Theorem}[section]
\newtheorem{lem}[thm]{Lemma}

\newtheorem{prp}[thm]{Proposition}

\newtheorem{remark}[thm]{Remark}

\providecommand{\keywords}[1]
{
  \small	
  \textbf{\textit{Keywords:}} #1
}

\begin{document}

\title{Quasi-Bayes empirical Bayes estimation of sums of random variables}


\author[1]{Stefano Favaro\thanks{stefano.favaro@unito.it}}
\author[2]{Sandra Fortini\thanks{sandra.fortini@unibocconi.it}}
\affil[1]{\small{Department of Economics and Statistics, University of Torino and Collegio Carlo Alberto, Italy}}
\affil[2]{\small{Department of Decision Sciences, Bocconi University, Italy}}

\maketitle

\begin{abstract}
The estimation of sums of functions of observable and unobservable variables is a long-standing problem in statistics with applications across many domains. Empirical Bayes methods provide a natural framework for this task under mixture models, but existing approaches often rely on restrictive parametric assumptions or apply only to limited classes of functionals in nonparametric settings. We propose a nonparametric methodology, referred to as quasi-Bayes empirical Bayes, that addresses these limitations through a recursive estimation of the mixing distribution based on Newton’s algorithm. The resulting plug-in estimate of the target sum is computationally efficient, scalable, and applicable to a broad class of utility functions, while enabling uncertainty quantification via asymptotic credible intervals derived from a Gaussian central limit theorem. We establish large sample asymptotic theoretical guarantees by proving a merging between the quasi-Bayes and Bayes estimates and by showing consistency under a correctly specified frequentist model. Synthetic-data and real-data analyses demonstrate the practical accuracy and stability of the method, with performance comparable to, and in some cases better than, existing empirical Bayes procedures.
\end{abstract}



\section{Introduction}\label{sec1}
\subsection{Problem setting}

Consider $n\geq1$ i.i.d. real-valued random vectors $(X_{1},\theta_{1}),\ldots,(X_{n},\theta_{n})$ with unknown distribution, where the $X_{i}$'s are observable and the $\theta_{i}$'s are either unobservable or constants. Given a utility function $u:\mathbb{R}\times\mathbb{R}\rightarrow\mathbb{R}$, a broad class of statistical problems concerns the estimation of sums
\begin{equation}\label{eq:sum}
S_{n}=\sum_{i=1}^{n}u(X_{i},\theta_{i});
\end{equation}
see, for instance, \citet{Rob(88),Rob(00),Zha(05)}. As an illustration we recall the motorist example of \citet{Rob(77)}. For each individual $i$ in a pool of $n$ motorists, let $X_{i}$ denote the number of accidents this year, $Y_{i}$ the (unobservable) number of accidents the next year, and $\theta_{i}$ the expected number of accidents or intensity, $i=1,\ldots,n$. Depending on the choice $u$, a variety of quantities can be expressed in the form \eqref{eq:sum}. For example, if $u(x,\theta)=\theta I(x=\kappa)$ with $\kappa\in\mathbb{N}$, then $S_{n}$ is the total intensity of motorists with $\kappa$ accidents this year; if $u(x,\theta)=I(x>\theta)$ then $S_{n}$ is the number of motorists with a number of accidents this year above their intensity; if $u(x,y)=yI(x\leq\kappa)$ with $\kappa\in\mathbb{N}$, then $S_{n}$ is the total number of accidents next year for motorists with no more than $\kappa$ accidents this year. 

Early studies on the estimation of sums of random variables $S_{n}$ focused primarily on treatment effects in clinical trials  \citep{Robz(88),Rob(89),Rob(91)}. A much broader range of applications is reviewed in \citet{Zha(05)}, highlighting both the generality and practical relevance of this class of problems. Examples include species sampling problems in biological and physical sciences, e.g., the estimation of the number of species in a population, the missing mass, and the number of unseen species \citep{Goo(53),Goo(56),Efr(76),Bun(93),Mao(02)}; data confidentiality, e.g., the estimation of measures of disclosure risks   \citep{Bet(90),Ski(02),Roc(19)}; and network data analysis, e.g., the estimation of node degrees based on source-destination data \citep{Var(96),Teb(98),Zha(05)}.

\subsection{Background on empirical Bayes}
\citet{Rob(00)} and \cite{Zha(05)} developed a general theory for the efficient estimation of $S_{n}$, based on a mixture model for the unknown distribution of the $(X_{i},\theta_{i})$'s. Specifically, let
\begin{align}\label{eq:model}
X_i\mid\theta_{i} & \quad\simind\quad k(\cdot\mid\theta_{i})\qquad i=1,\ldots,n\\[-0.2cm]
\notag\theta_{i} & \quad\simiid\quad G,
\end{align}
where $k(\cdot\mid\theta)$ is a known kernel density with respect to a dominant measure $\lambda$, and $G$ is the unknown mixing distribution on $\Theta\subseteq\mathbb{R}$. Within the model \eqref{eq:model}, empirical Bayes (EB) provides a principled approach to estimate $S_{n}$ \citep{Rob(56)}. Specifically, interpreting $G$ as a prior distribution on the $\theta_{i}$'s, the optimal estimate of $S_{n}$ under the squared-error loss is the posterior mean
\begin{equation}\label{eq:post_exp}
\hat{S}_{n}(G)=\E_{G}[S_{n}\mid X_{1:n}]=\sum_{i=1}^{n}\E_{G}[u(X_{i},\theta_{i})\mid X_{i}],
\end{equation}
i.e. the Bayes estimate. The EB approach then consists of estimating $G$ in $\eqref{eq:post_exp}$ from the observable data $X_{i}$'s, a strategy commonly referred to as $g$-modeling; see, for instance, \citet{Efr(14),Efr(19)}.

Parametric $g$-modeling assumes that $G$ belongs to the parametric family $\{G_{\tau}:\tau\in\mathcal{T}\}$, where $\mathcal{T}$ is a Euclidean space. An estimate of $G$ is then obtained by replacing $\tau$ with an estimate $\hat{\tau}_{n}$, yielding $G_{\hat{\tau}_{n}}$. Substituting $G_{\hat{\tau}_{n}}$ into \eqref{eq:post_exp} produces an EB estimate of $S_{n}$, which is asymptotically efficient provided that $\hat{\tau}_{n}$ is itself asymptotically efficient for $\tau$ \citep[Theorem 2.2]{Zha(05)}. While typically computationally tractable, this approach relies on strong (parametric) modeling assumptions on the form of $G_{\tau}$, making it vulnerable to model misspecification. Nonparametric $g$-modeling instead replaces $G$ in \eqref{eq:post_exp} with a nonparametric estimate $\hat{G}_{n}$, for instance based on maximum likelihood, minimum distance or Bayesian methods. While this is a natural extension of parametric $g$-modeling, EB estimates of the form $\hat{S}_{n}(\hat{G}_{n})$ have not been investigated in the literature; their theoretical properties are unknown, and the approach is expected to be computationally demanding due to the infinite-dimensional nature of estimating $G$. A more tractable nonparametric alternative is the ``$u,v$'' method of \citet{Rob(88)}, which is closely related to the so-called $f$-modeling strategy \citep{Efr(14)}. Although \citet[Theorem 2.5]{Zha(05)} shows that ``$u,v$'' estimates of $S_{n}$ are asymptotically efficient, their construction applies only to a restricted class of utility functions $u$.

\subsection{Preview of our contributions}

We propose a nonparametric $g$-modeling approach for estimating sums of random variables $S_{n}$ that combines modeling flexibility with computational scalability and provable theoretical guarantees, thereby addressing key limitations of existing parametric and nonparametric methods. 

The proposed approach relies on a recursive estimation of the mixing distribution $G$. Specifically, starting from an initial guess $G_0$ for $G$, for $n\geq1$ we iteratively update $G_{n-1}$ upon observing $X_{n}$ according to the procedure of \citet{Smi(78)}, referred to as Newton's algorithm \citep{New(98),Mar(08)}.  This recursive scheme can be interpreted as a quasi-Bayesian learning model and yields a predictive construction that is asymptotically equivalent, as $n\rightarrow+\infty$, to a Bayesian model \citep{For(20),For(25)}. A quasi-Bayes EB estimate of $S_{n}$ is then obtained from \eqref{eq:post_exp} by replacing $G$ with $G_{n}$. The resulting estimate $\hat{S}_{n}(G_{n})$ is straightforward to evaluate and computationally efficient, as the update of $G_{n}$ incurs a constant per-observation computational cost as data accrue. Viewing Newton’s algorithm as a learning model allows uncertainty quantification for $\hat{S}_{n}(G_{n})$; in particular, we construct asymptotic credible intervals based on a central limit theorem.

We establish theoretical guarantees for the quasi-Bayes EB estimate $\hat{S}_{n}(G_{n})$. First, as $n\rightarrow+\infty$, we prove a merging of the distributions induced by the quasi-Bayesian learning process and the corresponding Bayesian learning process, namely the posterior distribution. This result thus rigorously justifies the interpretation of $\hat{S}_{n}(G_{n})$ as an approximate Bayesian estimate of $S_{n}$ in large samples. Second, under the assumption that the observations $X_{i}$'s are i.i.d. from the model \eqref{eq:model} with a ``true'' mixing distribution $G^{\ast}$, we prove frequentist guarantees in terms of consistency of $\hat{S}_{n}(G_{n})$ relative to the oracle estimate $\hat{S}_{n}(G^{\ast})$, as $n\rightarrow+\infty$. In particular, we provide sufficient conditions under which $n^{-1}|\hat{S}_{n}(G_{n})-\hat{S}_{n}(G^{\ast})| \to 0$ almost surely as $n\rightarrow+\infty$, and we then derive a convergence rate under additional assumptions.

We empirically validate the proposed methodology on both synthetic and real datasets, focusing on Poisson and Gaussian mixture models. To this end, we consider examples of $S_{n}$ drawn from \citet[Sections 2.3 and 2.4]{Zha(05)}. The quasi-Bayes EB estimate $\hat{S}_{n}(G_{n})$ is compared with parametric EB estimates and, whenever available, with nonparametric ``$u,v$'' estimates. The results indicate that $\hat{S}_{n}(G_{n})$ attains accuracy that is generally comparable to, and in some cases improves upon, these alternatives, while remaining computationally efficient and accommodating a broader class of sums through the choice of the utility function $u$.

\subsection{Related work}

The term quasi-Bayes EB was introduced by \citet{Fav(24)}, who proposed a quasi-Bayesian $g$-modeling approach to mean estimation in Poisson mixture models \citep{Bro(13),Jan(24),She(24),Can(26)}. As in \citet{Fav(24)}, the methodology proposed here relies on Newton's algorithm, but it is conceptually and technically distinct, as it addresses the more challenging problem of estimating sums of random variables. Both \citet{Fav(24)} and the present paper contribute to bridging the EB literature with a growing body of work on predictive or recursive inference, often referred to as the “post-Bayes” framework. This line of research includes developments on predictive distributions and martingale posteriors \citep{Han(18), Fon(23), For(20), For(25), Bat(25)}, as well as approaches based on generalized Bayesian updates \citep{Bis(16), Kno(22)}.

\subsection{Organization of the paper}

Section~\ref{sec2} introduces the quasi-Bayes EB approach for estimating $S_n$, while Section~\ref{sec3} establishes its asymptotic frequentist guarantees. Section~\ref{sec4} presents examples of $S_n$ in the Poisson mixture model, together with numerical experiments based on synthetic data. For these examples, Section~\ref{sec5} provides a real-data illustration using data on the total number of goals scored in the National Hockey League during the 2017--2018 and 2018--2019 seasons, obtained from \url{https://www.hockey-reference.com/}. Section~\ref{sec6} concludes with a discussion of multidimensional extensions of the proposed quasi-Bayes EB approach, alternative nonparametric EB strategies, connections with species sampling problems, and directions for future research. Proofs, additional synthetic-data illustrations for both Poisson and Gaussian mixture models, and further real-data analyses are collected in the Supplementary Material.


\section{Quasi-Bayes EB for sums of random variables}\label{sec2}

\subsection{Preliminaries and modeling assumptions}

For $n\geq1$, let $(X_{1},\theta_{1}),\ldots,(X_{n},\theta_{n})$ be defined on a probability space $(\Omega,\mathcal{F},\mathbb{P})$ and distributed according to the mixture model \eqref{eq:model}, with unknown mixing distribution $G$ supported on $\Theta\subseteq\mathbb{R}$. We consider the problem of estimating sums of random variables $S_{n}$, as defined in \eqref{eq:sum}.

The Bayesian $g$-modeling approach to estimate $S_{n}$, here referred to as Bayes EB following the terminology of \citet{Dee(81)}, consists of placing a prior distribution on $G$. Equivalently, it is assumed the existence of a random probability distribution $\tilde{G}$ on $\Theta$, such that for $i=1,\ldots,n$
\begin{equation}\label{eq:model1bis}
X_i \mid \tilde{G}\simind f_{\tilde{G}}(x) = \int_{\Theta} k(\cdot \mid \theta) \tilde{G}(\ddr\theta)
\end{equation}
and 
\begin{equation}\label{eq:model1bis2}
\theta_{i}\mid X_{1:n},\, \tilde{G}\simind\frac{k(X_i \mid \theta)\tilde{G}(\ddr\theta)}{\int_{\Theta} k(X_i \mid \theta) \tilde{G}(\ddr\theta)}.
\end{equation}
The mixture model \eqref{eq:model1bis}-\eqref{eq:model1bis2} offers an equivalent formulation of \eqref{eq:model}, under which the Bayes estimate \eqref{eq:post_exp} becomes
\begin{equation}\label{eq:post_exp_new}
\hat{S}_{n}(\tilde{G})=\E_{\tilde{G}}[S_{n}\mid X_{1:n}]=\sum_{i=1}^{n}\E_{\tilde{G}}[u(X_{i},\theta_{i})\mid X_{i}],
\end{equation}
and $\tilde{G}$ is estimated through the posterior distribution of $\tilde{G}$ given $X_{1:n}$, e.g. the posterior mean. 

\begin{remark}
The term Bayes EB was introduced by \citet{Dee(81)}, who proposed a Bayesian $g$-modeling approach to the compound decision problems of \cite{Rob(51)}; see also \citet{Efr(14),Efr(19)}. Their treatment is largely parametric, although they briefly discuss the possibility of placing a nonparametric prior on $G$, through a Dirichlet process \citep{Fer(73)}. Recently, \citet{Ign(26)} pursued this direction by employing a Dirichlet process prior within the Bayes EB approach to mean estimation in Gaussian mixture models.
\end{remark}

Here, we develop a quasi-Bayesian $g$-modeling approach to estimate $S_{n}$. Like the Bayesian framework described above, the quasi-Bayesian framework treats the model  as a learning model rather than a data-generating process \citep{For(20),For(25)}. However, a fundamental difference lies in how the learning process for the $\theta_{i}$'s is specified: while the Bayesian framework makes use of the posterior distribution of $\tilde{G}$ given $X_{1:n}$,  yielding 
\begin{equation}\label{eq:model1bis22}
\theta_{i}\mid X_{1:n},\simind\EE\left[\frac{k(X_i \mid \theta)\tilde{G}(\ddr\theta)}{\int_{\Theta} k(X_i \mid \theta) \tilde{G}(\ddr\theta) }\mid X_{1:  n}\right],
\end{equation}
the quasi-Bayesian framework employs a recursive estimation of the mixing distribution $\tilde{G}$ through Newton's algorithm \citep{Smi(78),New(98)}, as we now describe.

The quasi-Bayes framework assumes the existence of an unknown probability distribution $ \tilde{G} $ that governs the distribution of $ X_{1:n}$, as specified in \eqref{eq:model1bis}. In addition, it specifies how to learn the $\theta_{i}$'s based on $ X_{1:n} $, by replacing the posterior-based learning process \eqref{eq:model1bis22} with  
\begin{equation}\label{cond_est}
 \theta_i \mid X_{1:n} \simind G_n(\theta \mid X_i) = \frac{k(X_i \mid \theta) G_n(\ddr\theta)}{\int_{\Theta} k(X_i \mid \theta) G_n(\ddr \theta)},
\end{equation}
where $G_n $ is a recursively updated estimate of $\tilde G$, computed via Newton’s algorithm, i.e. for $n\geq 0$
\begin{equation}\label{eq:newton}
G_{n+1}(\ddr\theta) = (1 - \alpha_{n+1}) G_n(\ddr\theta) + \alpha_{n+1} \frac{k(X_{n+1} \mid \theta) G_n(\ddr\theta)}{\int_{\Theta} k(X_{n+1} \mid \theta) G_n(\ddr \theta)},
\end{equation}
starting from an initial guess $ G_0 $, and using a sequence $ (\alpha_n)_{n \geq 1} $ of real numbers in $(0,1)$ satisfying $ \sum_{n \geq 1} \alpha_n = +\infty $ and $ \sum_{n \geq 1} \alpha_n^2 < +\infty $. From \eqref{eq:newton}, $G_n$ is updated as the observation $X_{n+1}$ becomes available, by taking a weighted average, with respect to $\alpha_{n}$, of  $G_n$ itself and of its posterior distribution based on $X_{n+1}$. The sequence $(\alpha_{n})_{n\geq1}$ is the learning rate \citep{For(20)}, for which a standard choice is $\alpha_{n}=(\alpha+n)^{-\gamma}$ for $\alpha>0$ and $\gamma\in (1/2,1]$.

The quasi-Bayesian learning process \eqref{cond_est}-\eqref{eq:newton} differs conceptually and technically from that developed in \citet{Fav(24)} for mean estimation in Poisson mixture models, though both rely on Newton's algorithm. The distinction lies not only in the target of estimation but also in the learning mechanism. In \citet{Fav(24)}, Newton's algorithm defines a sequential predictive model for the observations, as the recursive estimate of the mixing distribution determines the predictive distribution of future data. In this sense the learning proceeds in a ``forward direction", from the data toward prediction of future observations. In the present work, instead, Newton's recursion is used to approximate Bayesian learning of the latent $\theta_{i}$'s by replacing the unknown mixing distribution $\tilde{G}$ with its recursive estimate $G_{n}$ in posterior expectations, i.e. \eqref{cond_est}-\eqref{eq:newton}. The learning therefore proceeds in a ``backward direction", from the data toward inference on the latent parameters.

Under the mixture model \eqref{eq:model1bis}, let $\mu$ be a $\sigma$-finite measure on $\Theta$, and let $\mathbb G$ denote a class of probability measures on $\Theta$ that are absolutely continuous with respect to $\mu$. Let $\overline{\mathbb G}$ be the closure of $\mathbb G$ (with respect to the weak topology). With $\tilde{G}$ learned as in \eqref{cond_est}-\eqref{eq:newton}, we assume throughout:
\begin{itemize}
    \item[$(\mathcal A1)$] $\Theta$ is compact.
    \item[$(\mathcal A2)$] $G_0$ is absolutely continuous with respect to 
    $\mu$, with strictly positive density on $\Theta$.
      \item[$(\mathcal A3)$] The kernel $k(x\mid\theta)$ is strictly positive, bounded, continuous in $\theta$ for $\lambda$-a.e. $x$, and dominated by a function $h(x)$ in  $L^1(\lambda)$.
          \item[$(\mathcal A4)$] The mixture model $\{\int_\Theta k(x\mid\theta)G(\ddr\theta): G\in \overline{\mathbb G}\}$ is identifiable.
        \item[$(\mathcal A5)$]    \(
\displaystyle{\sup_{\theta,\theta',\theta''\in\Theta}\int_{\mathbb{R}}\frac{k^{2}(x\mid\theta)}{k^{2}(x\mid\theta')}k(x\mid\theta'')\lambda(\ddr x)<+\infty.}\)
          \end{itemize}
   
\subsection{Quasi-Bayes properties}

We denote by $\P^{(n)}$ and $\P$ the probability measures that specify the joint distribution of the $(X_{i},\theta_{i})$'s in the quasi-Bayesian learning process, i.e. \eqref{eq:model1bis} with  \eqref{cond_est}-\eqref{eq:newton}, and in the Bayesian learning process, i.e. \eqref{eq:model1bis} with \eqref{eq:model1bis2}. While the marginal distribution of the $X_{i}$'s coincides under $\P$ and $\P^{(n)}$, the conditional distribution of $(\theta_{1},\ldots,\theta_{n})$ given $X_{1:n}$ differs between $\P^{(n)}$ and $\P$. Although not necessary for the estimation of $S_n$, the distribution of the sequence $((X_i,\theta_i))_{i\geq1}$ under $\P^{(n)}$ can be specified consistently with the law of $((X_1,\theta_1),\dots, (X_n,\theta_n))$ by setting 
\begin{displaymath}
\theta_{n+k+1}\mid (X_1,\theta_1),\dots,(X_{n+k},\theta_{n+k})\sim G_{n+k}(\theta)\quad \mbox{ for } k\geq 0,
\end{displaymath}
where $G_{n+k}(\theta)$ is computed from Newton’s algorithm \eqref{eq:newton}, and the conditional distribution of $X_{n+k+1}$ given $((X_1,\theta_1),\dots,(X_{n+k},\theta_{n+k}),\theta_{n+k+1})$ is determined by the kernel $k(\cdot\mid\theta_{n+k+1})$ \citep{For(20)}. The next theorem shows that $\Pn$ directly connects to $\P$, as $n\rightarrow+\infty$.

\begin{thm}\label{th:convg}
Let $ G_n $ and $ G_n(\cdot \mid x) $ be defined as in  \eqref{eq:newton} and \eqref{cond_est}, respectively, and let $ f_{G_n} $ denote the marginal density function in \eqref{eq:model1bis} with $\tilde G $ being replaced by $ G_n $. Under the assumptions $(\mathcal A1)$-$(\mathcal A5)$ and $\tilde G\in\overline{\mathbb G}$,  it holds $\PP$-a.s. that, as $ n \rightarrow +\infty $, 
\begin{itemize}
\item[i)]
$
G_{n}\stackrel{w}{\longrightarrow}\tilde{G};
$
\item[ii)]
$
f_{G_{n}}\rightarrow f_{\tilde{G}}\qquad \lambda-a.e. \text{ and  in } L^1(\lambda);
$
\item[iii)]
$ 
G_{n}(\cdot \mid x)\stackrel{w}{\longrightarrow} \tilde G(\cdot \mid x)\qquad\mbox{for } \lambda\mbox{-a.e. } x,
$
\end{itemize}
where $\stackrel{w}{\longrightarrow}$ denotes the weak convergence.
\end{thm}

See \ref{appB} for the proof of Theorem \ref{th:convg}. The next theorem, which is a consequence of \eqref{cond_est} and Theorem \ref{th:convg}, establishes the merging between the probability measures $\P^{(n)}$ and $\P$ as $n\rightarrow+\infty$. We denote by $\mathcal{L}_{\P^{(n)}}(Y\mid X)$ and $\mathcal{L}_{\P}(Y\mid X)$ the conditional distribution of $Y$ given $X$ under $\P^{(n)}(\cdot)$ and $\P(\cdot)$, respectively, where $X$ and $Y$ are random variables defined on $(\Omega,\mathcal{F})$. 

\begin{thm}\label{th:merg}
Under assumptions $(\mathcal A1)$-$(\mathcal A5)$ and $\tilde G\in\overline{\mathbb G}$, for every $m\geq1$ it holds $\PP$-a.s. that, as $n\rightarrow+\infty$
\begin{equation}\label{eq:approx}
  \rho(\mathcal L_{\Pn}((\theta_1,\dots,\theta_m)\mid X_{1:n}),\mathcal L_{\P}((\theta_1,\dots,\theta_m)\mid X_{1:n}))\rightarrow 0,
\end{equation}
where $\rho$ denotes the Prohorov distance.
\end{thm}

By Theorem~\ref{th:merg}, and using the fact that the marginal distribution of $(X_1,\ldots,X_m)$ is the same under the probability measures $\PP$ and $\PP^{(n)}$ for all $n\geq1$,
we obtain that, for any $m\geq1$,
\begin{displaymath}
\P^{(n)}\bigl(((X_1,\theta_1),\ldots,(X_m,\theta_m))\in \cdot \bigr)
\stackrel{w}{\longrightarrow}
\int_{\cdot}\prod_{i=1}^m k(x_i\mid s_i)\,\tilde G(\ddr s_i)\,\lambda(\ddr x_i),
\qquad n\to+\infty .
\end{displaymath}
In other words, $\P^{(n)}$ takes on the interpretation of a large $n$ approximation of the probability measure $\P$.

\subsection{Estimation and credible intervals}
We estimate $S_{n}$ through its conditional expectation given $X_{1:n}$ under the quasi-Bayesian learning process $\Pn$. This expectation, denoted by $\E^{(n)}$, is obtained by replacing $\tilde{G}$ in \eqref{eq:post_exp_new} with $G_{n}$, i.e.
\begin{equation}\label{eq:est}
\hat{S}^{\text{\tiny{[Q-B]}}}_{n}:=\E^{(n)}[S_{n}\mid X_{1:n}]=\sum_{i=1}^{n}\int_{{\Theta}}u(X_{i},\theta)G_{n}(\ddr\theta\mid X_{i}).
\end{equation}
We refer to $\hat{S}^{\text{\tiny{[Q-B]}}}_{n}$ as the quasi-Bayes EB estimate of $S_{n}$. In particular, it is useful to observe that
\begin{displaymath}
S_{n}-\hat{S}^{\text{\tiny{[Q-B]}}}_{n}=\sum_{i=1}^{n}\left(u(X_{i},\theta_{i})-\int_{{\Theta}}u(X_{i},\theta)G_{n}(\ddr\theta\mid X_{i})\right),
\end{displaymath}
provides a sum of  independent  (and centered at zero) random variables under $\mathbb{P}^{(n)}(\cdot\mid X_{1:n})$. The next theorem follows from a conditional version of Lindeberg-Feller central limit theorem.
  
\begin{thm}\label{th:central}
Let $\hat{S}^{\text{\tiny[Q-B]}}_{n}$ be the quasi-Bayes EB estimate of $S_n$.
Assume that $(\mathcal A1)$--$(\mathcal A5)$ hold, $\tilde G\in\overline{\mathbb G}$, and that the utility function $u$ is bounded and satisfies:
\begin{enumerate}[label=(\roman*)]
\item for $\lambda$-a.e.\ $x$, $\theta \mapsto u(x,\theta)$ is $\tilde G$-a.s. continuous  $(\PP$-a.s.$)$;
\item the families $\{\int_\Theta u(x,\theta)\,G(\ddr\theta\mid x) : G \in \mathbb G\}$ and $\{\int_\Theta u^2(x,\theta)\,G(\ddr\theta\mid x) : G \in \mathbb G \}$ are locally equicontinuous on the support of $\lambda$;
\item
\(
\iint_\Theta(u(x,\theta)-\tilde u(x))^2k(x\mid\theta)\tilde G(\ddr\theta)\lambda(\ddr x)>0\)
$(\PP$-a.s.$)$,
where $\tilde u(x)=\int_\Theta u(x,\theta)\tilde G(\ddr\theta\mid x)$.
\end{enumerate}
Then, as $n \to \infty$,
\[
\P^{(n)}\!\left(
\frac{\hat{S}^{\text{\tiny[Q-B]}}_{n}-S_n}{\sqrt{B_n}}
\le z\,|\, X_{1:n}\right)
\longrightarrow \Phi(z),
\qquad \PP\text{-a.s.},\ \forall z\in\R,
\]
where $\Phi$ denotes the standard Gaussian (cumulative) distribution function, and $B_n = \sum_{i=1}^n \sigma_{n,i}^2$ with 
\begin{displaymath}
\sigma_{n,i}^2
= \int_\Theta u(X_i,\theta)^2\,G_n(\ddr\theta\mid X_i)
- \left(\int_\Theta u(X_i,\theta)\,G_n(\ddr\theta\mid X_i)\right)^2.
\end{displaymath}
\end{thm}

See \ref{app:central} for the proof of Theorem \ref{th:central}. Assumption (ii) is trivially satisfied when the kernel is discrete, since in this case equicontinuity is with respect to the discrete topology. From Theorem \ref{th:central}, denoting by $z_{1-\beta/2}$ the $(1-\beta/2)$-quantile of the standard Gaussian distribution, as $n\rightarrow\infty$
\begin{equation}\label{cred_interval}
\P^{(n)}(\hat{S}^{\text{\tiny{[Q-B]}}}_{n}-z_{1-\beta/2}\sqrt{B_n}\leq S_n\leq \hat{S}^{\text{\tiny{[Q-B]}}}_{n}+z_{1-\beta/2}\sqrt{B_n}\mid X_{1:n})\rightarrow 1-\beta,\qquad\PP\text{-a.s.}
\end{equation}
Thus, with respect to $\P^{(n)}$, an asymptotic credible interval for $S_n$ is given by $\hat{S}^{\text{\tiny{[Q-B]}}}_{n}\pm z_{1-\beta/2}\sqrt{B_n}$. In particular, the proof of Theorem~\ref{th:central} suggests that the interval length is of order $n^{1/2}$.


\section{Frequentist guarantees}\label{sec3}
\subsection{Consistency}
Consider the mixture model \eqref{eq:model} with the mixing distribution $G$ replaced by a ``true" distribution, say the oracle prior $G^*$ on $\Theta\subset\mathbb{R}$, such that the observations $X_{i}$'s are i.i.d. with density function 
\begin{equation}\label{eq:fg*}
f_{G^*}(x)=\int_{\Theta} k(x\mid\theta)G^*(\ddr\theta)\qquad x\in\mathbb{R}
\end{equation}
with respect to $\lambda$. If the distribution $G^*$ were known, one could estimate $S_n$ through the oracle estimate
\begin{displaymath}
\hat S_{n}(G^{\ast})=\sum_{i=1}^n \int_\Theta u(X_i,\theta)\, G^*(\ddr\theta\mid X_i).
\end{displaymath}
The next theorem establishes frequentist guarantees in terms of consistency of the quasi-Bayes EB estimate $\hat{S}^{\text{\tiny{[Q-B]}}}_{n}$ relative to $\hat S_{n}(G^{\ast})$, as $n\rightarrow+\infty$, under $G^{\ast}$. We denote by $\P^{\ast}$ the probability measure under which the $X_{i}$'s are i.i.d. with density function $f_{G^*}$ as in \eqref{eq:fg*} with respect to $\lambda$.

\begin{thm}\label{th:consist}
Under assumptions $(\mathcal A1)$-$(\mathcal A5)$, if $G^*\in{\mathbb G}$, the utility function $u(x,\theta)$ is $G$-a.e. continuous in $\theta$ for every $G\in\mathbb G$ and for $\lambda$-almost every $x$, and the following condition holds:
 there exist measurable non negative functions  
$a$ and $b$ such that
\begin{equation}\label{eq:cond1}
|u(x,\theta)| \le a(\theta)\,b(x)
\;\mbox{ for all }(x,\theta),\quad\quad\sup_{\theta\in\Theta}a(\theta)<\infty,\quad\int b(x)\, f_{G^*}(x)\lambda(\ddr x)<\infty.
\end{equation}
Then, as $n\rightarrow\infty$
    \begin{equation}\label{eq_consistency}
      \frac 1 n  |\hat{S}^{\text{\tiny{[Q-B]}}}_n-\hat S_{n}(G^{\ast})|\rightarrow 0\qquad \P^*\mbox{-a.s.}
    \end{equation}
\end{thm}

See \ref{app:proofconsistency} for the proof of Theorem \ref{th:consist}. 

\subsection{Rate of convergence}

Under the additional assumption that $G_{n}$ in Newton's algorithm \eqref{eq:newton} is supported on a finite grid $\Theta_\diamond\subset\Theta$, we establish a rate of convergence for $n^{-1}  |\hat{S}^{\text{\tiny{[Q-B]}}}_n-\hat S_{n}(G^*_\diamond)|$, where $G^*_\diamond$ is defined as
    \begin{equation}\label{eq:discr}
G^*_\diamond = \arg\min_{G \in \mathcal{S}_\diamond} \mathrm{KL}(f_{G^*} \,\|\, f_{G}),
\end{equation}
with $\mathcal{S}_\diamond$  denoting the class of probability distributions supported on $\Theta_\diamond$ and $\mathrm{KL}(\cdot ||\cdot)$ being the Kullback-Leibler divergence. Such a discretization of $\Theta$ is known to be a standard assumption in the analysis of convergence rates for Newton's algorithm, as it enables a quantification of asymptotic behavior without loss of generality \citep{Mar(08),Mar(12)}. 

For $n\geq1$, let $G_{n}$ in Newton's algorithm \eqref{eq:newton} be defined from an initial guess $G_{0}$ on $\Theta_\diamond$. Under $(\mathcal A4)$ and
\begin{equation}\label{eq:finitesecder}
\int_{\R}\frac{k(x\mid\theta)^2}{k(x\mid\theta')^2}f_{G^*}(x)\lambda(\ddr x)<+\infty\qquad\forall \theta,\theta'\in\Theta_\diamond,
\end{equation}
then, as $n\rightarrow+\infty$, $G_n$ converges $\P^*$-a.s. to $G^*_\diamond$ \citep[Corollary 4.7]{MarTok(09)}. Then, we denote by $\hat S_n(G^*_\diamond)$ the oracle estimate of $S_{n}$ based on the KL discretization $G_\diamond^*$ of the oracle prior $G^*$, i.e., 
\begin{equation}\label{eq:estg*}
    \hat S_n(G^*_\diamond)=\sum_{i=1}^n \int_{\Theta_\diamond} u(X_i,\theta)\, G_\diamond^*(\ddr\theta\mid X_i).
\end{equation}

\begin{thm}\label{th:regret}
Let $\hat{S}^{\text{\tiny{[Q-B]}}}_{n}$ be the quasi-Bayes EB estimate of $S_{n}$, with an initial guess $G_0$ supported on a grid $\Theta_\diamond$, with $\alpha_{n} = (\alpha + n)^{-\gamma}$ for some $\gamma\in(1/2,1)$. If the utility function $u$ is bounded and 
$(\mathcal A4)$ and \eqref{eq:finitesecder} hold, 
then
\begin{equation}\label{eq:ratedelta}
\bigl|\hat{S}^{\text{\tiny{[Q-B]}}}_{n}-\hat S_n(G^*_\diamond)\bigr|=o(n^{\delta})\qquad \P^{\ast}\text{-a.s.}
\end{equation}
for every $\delta>\tfrac{1}{2\gamma}$, where $\P^{\ast}$ is the probability measure under which the $X_{i}$’s are i.i.d. with density function $f_{G^*}$ as in \eqref{eq:fg*}.
\end{thm}

See \ref{app:proofregret} for the proof of Theorem \ref{th:regret}. The next proposition constructs a finite grid for the Poisson and Gaussian kernels ensuring a KL approximation error of order at most $\eta$, together with asymptotics.

\begin{prp}\label{prp_grid}
Let  $k(\cdot\mid\theta)$ be either the Poisson kernel with mean (intensity) $\theta$ or the Gaussian kernel with mean $\theta$ and known variance $\sigma^2$. Furthermore, assume that $\int_{\mathbb{R}} |x|^k f_{G^*}(x)\,\lambda(dx) \le m_k < +\infty$ for some $k\ge2$ (Poisson kernel) or $k>2$ (Gaussian kernel). For $\eta>0$, let $\mathcal S_\diamond$ denote the class of probability distributions supported on a grid $\Theta_\diamond$, where,
\begin{itemize}
\item in the Poisson case
\begin{displaymath}
\Theta_\diamond=\left\{ \frac{i\eta}{2}\text{ : } i=1,\dots,d \right\}
\end{displaymath} 
with 
\begin{displaymath}
d=\inf\left\{j\in\mathbb N\text{ : }\dfrac{j\eta}{2}\ge e,\dfrac{\log(j\eta/2)\,m_k}{j^{k-1}} \le \dfrac{\eta^k}{2^k}\right\};
\end{displaymath}
\item in the Gaussian case
\begin{displaymath}
\Theta_\diamond=\{ i\sigma\sqrt{\eta}\text{ : } i=-d,\dots,d \}
\end{displaymath}
with
\begin{displaymath}
d=\inf\left\{j\in\mathbb N\text{ : }\dfrac{m_k}{j^{k-2}\sigma^k} \le \eta^{k/2}\right\}.
\end{displaymath}
\end{itemize}
Then the distribution $G^*_\diamond =\arg\min_{G\in\mathcal{S}_\diamond}\mathrm{KL}(f_{G^*}||f_G)$ exists, it is unique, and it satisfies $\mathrm{KL}(f_{G^*}||f_{G^*_\diamond })<\eta$. Moreover, if $\sup_{x,\theta}|u(x,\theta)|<+\infty$, $\alpha_n=(\alpha+n)^{-\gamma}$ with $\gamma\in(1/2,1)$, and 
\begin{equation}\label{eq:momentgf}
\int_{\mathbb{R}} h(x)f_{G^*}(x)\lambda(\ddr x)<+\infty 
\end{equation}
with $h(x)=d^{2x} $ and $h(x)=\exp\left\{\frac{4d\sqrt{\eta} |x|}{\sigma}\right\}$ in the Poisson and Gaussian kernel, respectively, then \eqref{eq:ratedelta} holds.
\end{prp}

See \ref{sec_ex_grid} for the proof of Proposition \ref{prp_grid}.  Explicit constructions of a grid that guarantees a KL approximation error of at most order $\eta$  can be given for other kernels, such that the exponential family, though the Poisson and Gaussian kernels are the most common \citep[Section 2]{Zha(05)}. A construction for a general kernel 
$k$, based on its structural properties, is presented in \ref{thm_grid}.


\section{Examples with Poisson kernel}\label{sec4}

For the Poisson mixture model, we present examples of $S_{n}$ drawn from \citet[Section 2.3]{Zha(05)}. Let $n\geq1$ and consider the real-valued random vectors $(X_{1},Y_{1},\theta_{1}),\ldots,(X_{n},Y_{n},\theta_{n})$ distributed as follows:
\begin{align}\label{model_distr_poi}
Y_i\mid X_{i},\,\theta_{i} & \quad\simind\quad \text{Poisson}(\cdot\mid\theta_{i})\qquad i=1,\ldots,n\\[-0.2cm]
\notag X_i\mid\theta_{i} & \quad\simind\quad \text{Poisson}(\cdot\mid\theta_{i})\\[-0.2cm]
\notag\theta_{i}& \quad\simiid\quad{G},
\end{align}
where $\text{Poisson}(\cdot\mid\theta)$ is the Poisson kernel of mean (intensity) $\theta\in\Theta$, and ${G}$ is the unknown mixing distribution on $\Theta\subset\mathbb{R}^{+}$. Given the observed $X_{i}$'s, we address the estimation of $S_{1,n}=\sum_{1\leq i\leq n}\theta_{i}I(X_{i}\leq\kappa)$, $S_{2,n}=\sum_{1\leq i\leq n}Y_{i}I(X_{i}\leq\kappa)$ and $S_{3,n}=\sum_{1\leq i\leq n}I(X_{i}>\theta_{i})$, where $\kappa\in\mathbb{N}$ is a given threshold \citep{Rob(77),Rob(88),Robz(88),Rob(89),Rob(91),Rob(00)}. Note that, under the model \eqref{model_distr_poi}, $\E_G[Y_i\mid X_i]=\E_G[\theta_i\mid X_i]$, for $i=1,\ldots,n$; accordingly the estimates of $S_{2,n}$ coincides with those of $S_{1,n}$ after replacing $G$ by its estimate.

\begin{remark}
In the celebrated motorist example \citep{Rob(77)}, $S_{1,n}$ denotes the total intensity of motorists with at most $\kappa$ accidents this year, $S_{2,n}$ denotes the total number of accidents next year among motorists with at most $\kappa$ accidents this year (for $\kappa\in\mathbb{N}$), and $S_{3,n}$ denotes the number of motorists whose number of accidents this year exceeds their intensity.
\end{remark}

\subsection{Parametric $g$-modeling and ``$u,v$" method}

Let ${G}$ in \eqref{model_distr_poi} be an Exponential distribution with rate (inverse scale) $\tau$. Under this assumption, we apply parametric $g$-modeling to obtain EB and Bayes EB estimates of $S_{r,n}$, $r=1,2,3$.

\begin{prp}\label{prop:poisson}
Let $n\geq1$ and let $(X_{1},Y_{1},\theta_{1}),\ldots,(X_{n},Y_{n},\theta_{n})$ be random vectors as in \eqref{model_distr_poi}, ${G}$ be an Exponential distribution with rate (inverse scale) $\tau$, here denoted by ${G}_{\tau}$. Then
\begin{equation}\label{s1_poisson}
\hat{S}_{r,n}({G}_{\tau})=\E_{\tilde{G}_{\tau}}[S_{r,n}\mid X_{1:n}]=\frac{1}{n}\sum_{i=1}^{n}\frac{1+X_{i}}{1+\tau}I(X_{i}\leq\kappa)\qquad r=1,2
\end{equation}
and
\begin{equation}\label{s3_poisson}
\hat{S}_{3,n}({G}_{\tau})=\E_{\tilde{G}_{\tau}}[S_{3,n}\mid X_{1:n}]=n-\sum_{i=1}^{n}\frac{\Gamma(1+X_{i},(1+\tau)X_{i})}{\Gamma(1+X_{i})},
\end{equation}
where $\Gamma(x)=\int_{\mathbb{R}^{+}}y^{x-1}\text{e}^{-y}\ddr y$ and $\Gamma(x,z)=\int_{(z,+\infty)}y^{x-1}\text{e}^{-y}\ddr y$. EB estimates of $S_{1,n}$, $S_{2,n}$ and $S_{3,n}$ are obtained as plug-in estimates by replacing $\tau$ in \eqref{s1_poisson} and \eqref{s3_poisson} with the maximum likelihood estimate
\begin{displaymath}
\hat{\tau}^{\text{\tiny{[ML]}}}_{n}=\frac{n}{\sum_{i=1}^{n}X_{i}}.
\end{displaymath}
By placing on $\tau/(1+\tau)$ a (standard) Beta prior with (hyper)-parameter $(\alpha,\beta)$, the Bayes EB estimate of $S_{1,n}$, $S_{2,n}$ and $S_{3,n}$ are obtained as plug-in estimates by replacing $\tau$ in \eqref{s1_poisson} and \eqref{s3_poisson} with the Bayesian estimate
\begin{displaymath}
\hat{\tau}^{\text{\tiny{[B]}}}_{n}=\frac{n+\alpha}{\beta+\sum_{i=1}^{n}X_{i}}.
\end{displaymath}
The resulting EB and Bayes EB estimates of $S_{r,n}$, for $r=1,2,3$, here denoted by $\hat{S}^{\text{\tiny{[ML]}}}_{r,n}$ and $\hat{S}^{\text{\tiny{[B]}}}_{r,n}$ respectively, are asymptotically efficient \citep[Theorem 2.2]{Zha(05)}.
\end{prp}

Now, we consider ${G}$ in \eqref{model_distr_poi} completely unknown, namely we do not specify any parametric assumption. In this nonparametric setting, the ``$u,v$" method of \citet{Rob(88)} estimates $S_{n}$ in \eqref{eq:sum} with $\hat{S}^{\text{\tiny[``u,v"]}}_{n}=\sum_{1\leq i\leq n}v(X_{i})$ if there exists a function $v$ that solve the ``$u,v$" equation
\begin{equation}\label{int_eq1}
\sum_{x\geq0}(v(x)-u(x,\theta))\frac{\text{e}^{-\theta}\theta^{x}}{x!}=0\qquad\forall\, \theta\in\Theta.
\end{equation}
For the quantities $S_{1,n}$ and $S_{2,n}$, Equation \eqref{int_eq1} admits the unique solution $v(x)=xI(x-1\leq\kappa)$ \citep{Rob(88),Rob(00)}. Hence, the ``$u,v$" estimates of $S_{1,n}$ and $S_{2,n}$ are
\begin{equation}\label{uv_poiss}
\hat{S}^{\text{\tiny[``u,v"]}}_{1,n}=\hat{S}^{\text{\tiny[``u,v"]}}_{2,n}=\frac{1}{n}\sum_{i=1}^{n}X_{i}I(X_{i}-1\leq\kappa).
\end{equation}
The estimates $\hat{S}^{\text{\tiny[``u,v"]}}_{1,n}$ and $\hat{S}^{\text{\tiny[``u,v"]}}_{2,n}$ are asymptotically efficient, as $n\rightarrow+\infty$ \citep[Theorem 2.5]{Zha(05)}. With regards to $S_{3,n}$, the ``$u,v$" equation \eqref{int_eq1} does not admit an explicit solution $v$, which makes not possible the derivation of a corresponding ``$u,v$" estimate \citep[Section 2.3]{Zha(05)}.


\subsection{Synthetic-data analysis}

For $i=1,\ldots,100$, let $\mathbf{X}_{i}=X_{1:100 i}$ denote a dataset of size $n=100i$ generated from the Poisson mixture model \eqref{model_distr_poi}, with Weibull prior  $G$ of scale $5$ and shape $3$. The $\mathbf{X}_{i}$'s are nested, so that the sample size increases progressively: at stage $i$, we have $n=100i$ data, obtained by adding $100$ new data at each step, starting from $n=100$. For each dataset $\mathbf{X}_{i}$ of size $n=100i$, we apply the quasi-Bayes EB approach to estimate $S_{1,n}$ and $S_{2,n}$, with $\kappa=2$, and $S_{3,n}$.

The implementation of Newton's algorithm \eqref{eq:newton} requires the numerical evaluation of an integral, i.e. the marginal likelihood, which we approximate via the trapezoidal rule. To this end, the density function of $G_{n}$ is represented through its values on a fixed grid of $d$ quadrature points over $\Theta$, with $d$ controlling the numerical resolution of the integration. This representation is used solely for numerical evaluation and does not impose any modeling restriction on $\Theta$. For a dataset $X_{1:n}$ of size $n$, letting $U_\Theta=\max\{\max\{X_{1:n}\},\lceil Q_{n,0.99}+4\sqrt{\max\{Q_{n,0.99},1\}}\rceil\}$, with $Q_{n,0.99}=\text{Quantile}(X_{1:n};0.99)$, we employ a uniform grid with $d=1{,}000$ over $\Theta=(0,U_\Theta)$, set $G_0$ to be Uniform on $\Theta$, and set $\alpha_n=(1+n)^{-0.99}$. Under this setting for Newton's algorithm, we obtain an estimate $G_n$ of the mixing distribution $G$, which, when substituted into \eqref{eq:est}, gives the quasi-Bayes EB estimates of $S_{r,n}$, $r=1,2,3$, i.e. 
\begin{displaymath}
\hat{S}^{\text{\tiny{[Q-B]}}}_{1,n}=\hat{S}^{\text{\tiny{[Q-B]}}}_{2,n}=\sum_{i=1}^n\E_{G_n}[\theta_i\mid X_i] I(X_i\leq \kappa)=\sum_{i=1}^{n}I(X_i\leq \kappa)
\int_{\Theta}\theta \,G_n(d\theta\mid X_i),
\end{displaymath}
and
\begin{displaymath}
\hat{S}^{\text{\tiny{[Q-B]}}}_{3,n}=\sum_{i=1}^n\E_{G_n}\!\left[I(X_i>\theta_i)\mid X_i\right]=\sum_{i=1}^n\int_{\Theta}I(X_i>\theta)\,G_n(d\theta\mid X_i).
\end{displaymath}
Note that $\hat{S}^{\text{\tiny{[Q-B]}}}_{1,n}=\hat{S}^{\text{\tiny{[Q-B]}}}_{2,n}$ because, under the model \eqref{model_distr_poi}, $\E_G[Y_i\mid X_i]=\E_G[\theta_i\mid X_i]$, for $i=1,\ldots,n$.

The estimates $\hat{S}^{\text{\tiny{[Q-B]}}}_{r,n}$, normalized by the sample size $n$, are reported in Figure~\ref{fig_weib_s1}-\ref{fig_weib_s3}, together with their mean absolute deviations (MAD) from the true values of $S_{r,n}$. See also Table~\ref{tab_weib_s1}-\ref{tab_weib_s3}. The estimate $\hat{S}^{\text{\tiny{[Q-B]}}}_{r,n}$ is compared with: the EB estimate $\hat{S}^{\text{\tiny{[ML]}}}_{r,n}$ and Bayes EB estimate $\hat{S}^{\text{\tiny{[B]}}}_{r,n}$ from Proposition \ref{prop:poisson}; ii) the ``$u,v$" estimate $\hat{S}^{\text{\tiny[``u,v"]}}_{r,n}$ in \eqref{uv_poiss}, which is available only for $S_{1,n}$. We also report the oracle Bayes estimate of $S_{r,n}$, namely $\hat{S}^{\text{\tiny{[O]}}}_{r,n}=\E_{G}[S_{r,n}\,|\,X_{1:n}]$ with $G$ being the Weibull distribution with scale $5$ and shape $3$.  An extended version of Table~\ref{tab_weib_s1}-\ref{tab_weib_s3} is reported in \ref{app_num1}. A sensitivity analysis of  $\hat{S}^{\text{\tiny{[Q-B]}}}_{1,n}$ with respect to the grid resolution $d$, which is provided in \ref{app_num1}, shows that the quasi-Bayes EB estimates is robust to the choice of $d$.

\begin{figure}[h!]
\begin{center}
\includegraphics[width=1\linewidth,height=0.36\textheight,keepaspectratio]{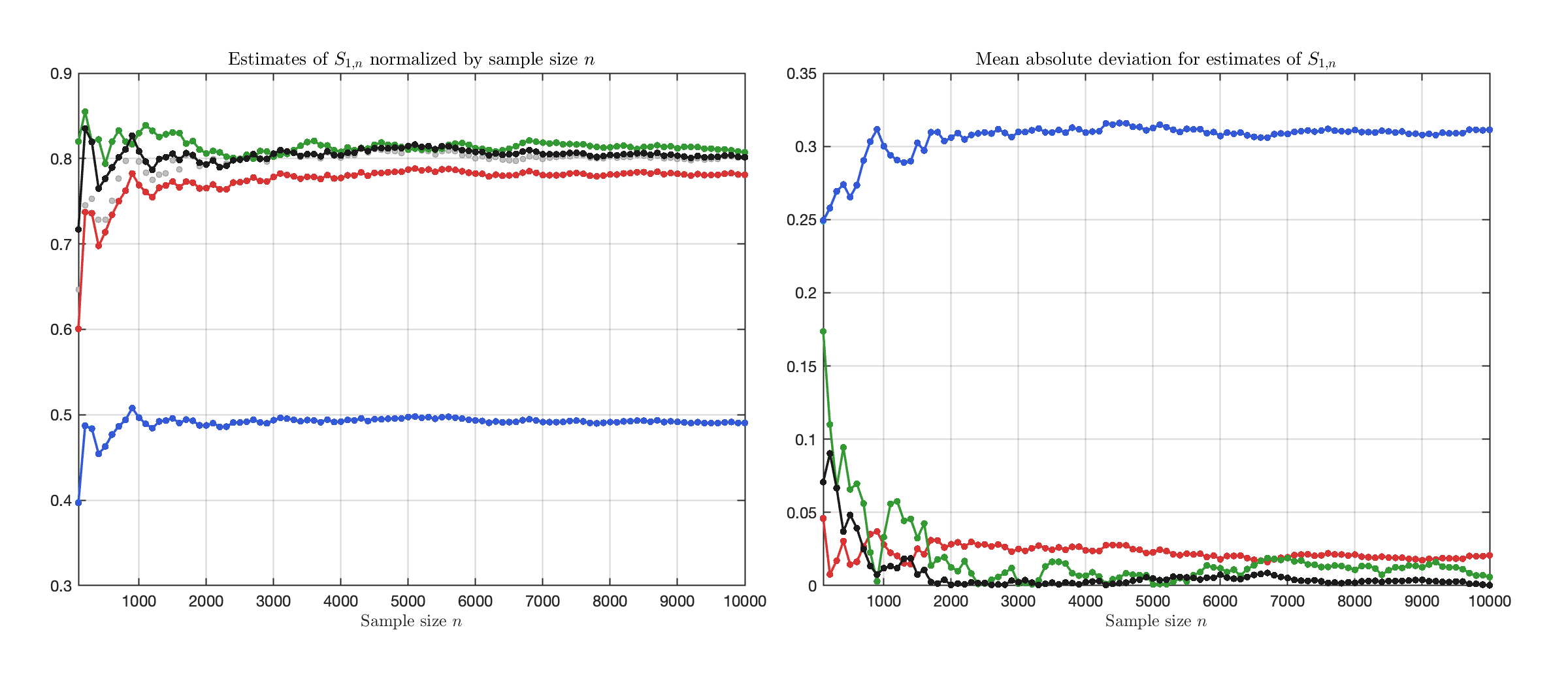}
\end{center}
\caption{\scriptsize{Weibull prior, $S_{1,n}$. Left panel: true values $n^{-1}S_{1,n}$ (Grey o-) and estimates $n^{-1}\hat{S}^{\text{\tiny{[O]}}}_{r,n}$ (Black .-); $n^{-1}\hat{S}^{\text{\tiny{[ML]}}}_{1,n}$ (Blue .-), $n^{-1}\hat{S}^{\text{\tiny{[B]}}}_{1,n}$ (Cyan .-), $n^{-1}\hat{S}^{\text{\tiny[``u,v"]}}_{1,n}$ (Green .-) and $n^{-1}\hat{S}^{\text{\tiny{[Q-B]}}}_{1,n}$ (Red .-). Right panel: MAD of $\hat{S}^{\text{\tiny{[O]}}}_{r,n}$ (Black .-); $\hat{S}^{\text{\tiny{[ML]}}}_{1,n}$ (Blue .-), $\hat{S}^{\text{\tiny{[B]}}}_{1,n}$ (Cyan .-), $\hat{S}^{\text{\tiny[``u,v"]}}_{1,n}$ (Green .-) and $\hat{S}^{\text{\tiny{[Q-B]}}}_{1,n}$ (Red .-)}}
\label{fig_weib_s1}
\end{figure}

\begin{table}[ht]
\centering
\caption{Weibull prior, $S_{1,n}$: MAD as $n$ varies}
{
\setlength{\tabcolsep}{0pt}
\begin{tabular}{@{}l@{\hspace{1.2cm}}*{5}{>{\centering\arraybackslash}p{2.15cm}}@{}}
\hline
\hline
& $\hat{S}^{\text{\tiny{[O]}}}_{1,n}$ & $\hat{S}^{\text{\tiny{[ML]}}}_{1,n}$ & $\hat{S}^{\text{\tiny{[B]}}}_{1,n}$ & $\hat{S}^{\text{\tiny{[``u,v'']}}}_{1,n}$ & $\hat{S}^{\text{\tiny{[Q-B]}}}_{1,n}$ \\[0.1cm]
\hline
\multicolumn{6}{@{}l}{\underline{$n=1{,}000$}} \\[0.05cm]
MAD & 0.0118 & 0.3002 & 0.3002 & 0.0331 & \textbf{0.0279} \\[0.2cm]

\multicolumn{6}{@{}l}{\underline{$n=3{,}000$}} \\[0.05cm]
MAD & 0.0022 & 0.3099 & 0.3099 & 0.0012 & \textbf{0.0250} \\[0.2cm]

\multicolumn{6}{@{}l}{\underline{$n=5{,}000$}} \\[0.05cm]
MAD & 0.0048 & 0.3126 & 0.3126 & 0.0007 & \textbf{0.0227} \\[0.2cm]

\multicolumn{6}{@{}l}{\underline{$n=7{,}000$}} \\[0.05cm]
MAD & 0.0051 & 0.3085 & 0.3085 & 0.0188 & \textbf{0.0193} \\[0.2cm]

\multicolumn{6}{@{}l}{\underline{$n=9{,}000$}} \\[0.05cm]
MAD & 0.0038 & 0.3079 & 0.3079 & 0.0124 & \textbf{0.0173} \\[0.1cm]
\hline
\hline
\end{tabular}
}
\label{tab_weib_s1}
\end{table}

\begin{figure}[h!]
\begin{center}
\includegraphics[width=1\linewidth,height=0.36\textheight,keepaspectratio]{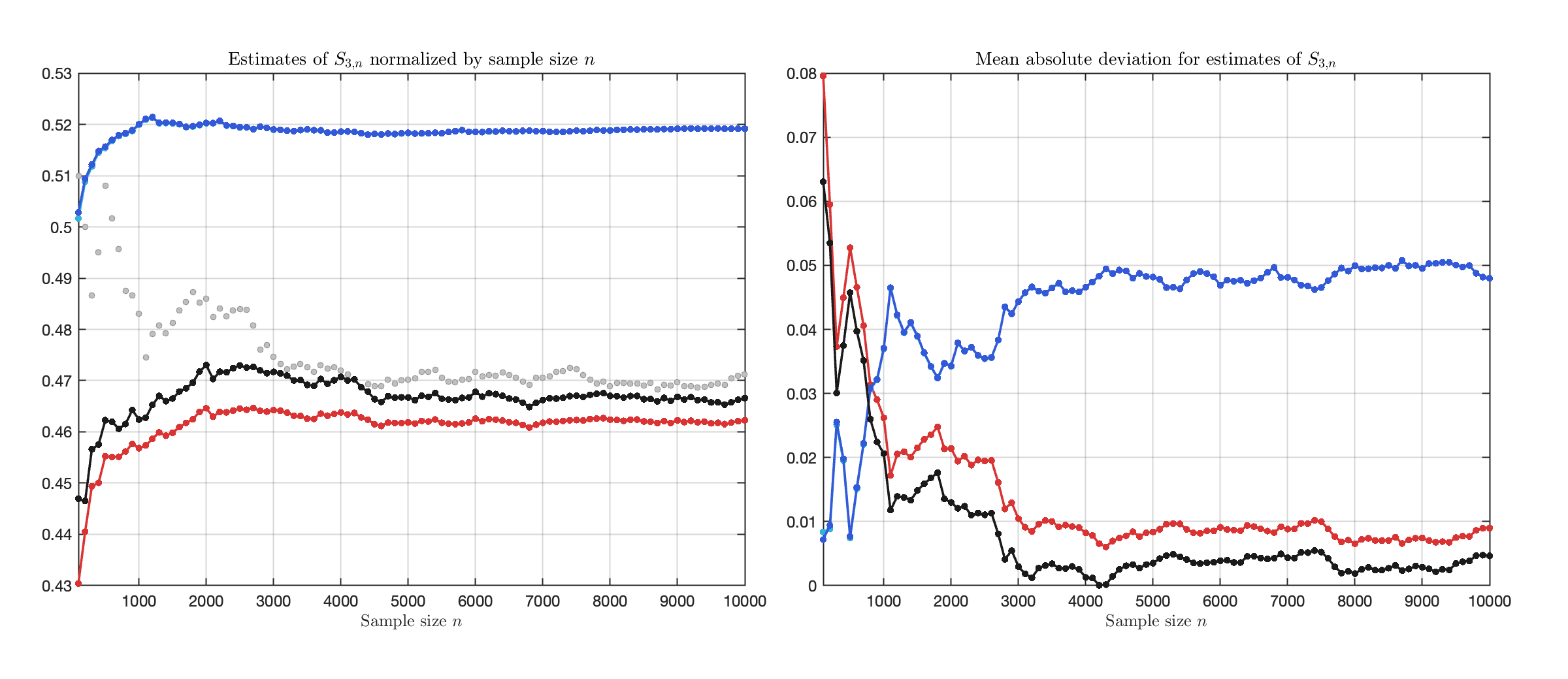}
\end{center}
\caption{\scriptsize{Weibull prior, $S_{3,n}$. Left panel: true values of $n^{-1}S_{3,n}$ (Grey o-) and estimates $n^{-1}\hat{S}^{\text{\tiny{[O]}}}_{3,n}$ (Black .-); $n^{-1}\hat{S}^{\text{\tiny{[ML]}}}_{3,n}$ (Blue .-), $n^{-1}\hat{S}^{\text{\tiny{[B]}}}_{3,n}$ (Cyan .-)  and $n^{-1}\hat{S}^{\text{\tiny{[Q-B]}}}_{3,n}$ (Red .-). Right panel: MAD of $\hat{S}^{\text{\tiny{[O]}}}_{3,n}$ (Black .-); $\hat{S}^{\text{\tiny{[ML]}}}_{3,n}$ (Blue .-), $\hat{S}^{\text{\tiny{[B]}}}_{3,n}$ (Cyan .-) and $\hat{S}^{\text{\tiny{[Q-B]}}}_{3,n}$ (Red .-)}}
\label{fig_weib_s3}
\end{figure}

\begin{table}[ht]
\centering
\caption{Weibull prior, $S_{3,n}$: MAD as $n$ varies}
{
\setlength{\tabcolsep}{0pt}
\begin{tabular}{@{}l@{\hspace{1.2cm}}*{4}{>{\centering\arraybackslash}p{2.15cm}}@{}}
\hline
\hline
& $\hat{S}^{\text{\tiny{[O]}}}_{3,n}$ & $\hat{S}^{\text{\tiny{[ML]}}}_{3,n}$ & $\hat{S}^{\text{\tiny{[B]}}}_{3,n}$ & $\hat{S}^{\text{\tiny{[Q-B]}}}_{3,n}$ \\[0.1cm]
\hline
\multicolumn{5}{@{}l}{\underline{$n=1{,}000$}} \\[0.05cm]
MAD & 0.0206 & 0.0371 & 0.0369 & \textbf{0.0262}\\[0.2cm]

\multicolumn{5}{@{}l}{\underline{$n=3{,}000$}} \\[0.05cm]
MAD & 0.0030 & 0.0444 & 0.0443 & \textbf{0.0104}\\[0.2cm]

\multicolumn{5}{@{}l}{\underline{$n=5{,}000$}} \\[0.05cm]
MAD & 0.0035 & 0.0482 & 0.0482 & \textbf{0.0084} \\[0.2cm]

\multicolumn{5}{@{}l}{\underline{$n=7{,}000$}} \\[0.05cm]
MAD & 0.0044 & 0.0481 & 0.0481 & \textbf{0.0088} \\[0.2cm]

\multicolumn{5}{@{}l}{\underline{$n=9{,}000$}} \\[0.05cm]
MAD & 0.0028 & 0.0495 & 0.0495 & \textbf{0.0074} \\[0.1cm]
\hline
\hline
\end{tabular}
}
\label{tab_weib_s3}
\end{table}

The results in Figure~\ref{fig_weib_s1}--\ref{fig_weib_s3} and Table~\ref{tab_weib_s1}--\ref{tab_weib_s3} highlight clear differences among the competing estimates; see also \ref{app_num1}. Focusing first on the estimation of $S_{1,n}$, the quasi-Bayes EB estimate exhibits a performance comparable to that of the ``$u,v$'' estimate, with both approaches providing more accurate estimates than the parametric EB and Bayes EB estimates, which are nearly indistinguishable. With regards to $S_{3,n}$, the ``$u,v$'' estimate is not available. In this case, the quasi-Bayes EB estimate outperforms both the parametric EB and Bayes EB estimates.

Turning to computation, the ``$u,v$'' estimate is computationally trivial whenever it is available, as it is based on a closed-form expression, and is therefore superior from a purely computational standpoint. The same holds for the parametric EB and Bayes EB estimates. By contrast, the quasi-Bayes EB estimate incurs the computational cost of updating the estimate of the mixing distribution $G_n$ via Newton’s algorithm \eqref{eq:newton}. This cost is constant per observation and depends only on the grid resolution $d$. For $d=1,000$, the CPU time required to update $G_n$ is approximately $0.001$ seconds, and it scales linearly with $d$. See \ref{app_num1} for details. Overall, in terms of the trade-off between accuracy (in terms of MAD), computational scalability, and generality of the class of functionals, the quasi-Bayes EB method is the most attractive compromise among the parametric and nonparametric methods considered.

A more detailed and exhaustive synthetic-data analysis is reported in \ref{app_num1}. This includes a study of the quasi-Bayes credible interval in \eqref{cred_interval}, together with further experiments under other choices of the prior $G$, i.e., the Uniform, half-Gaussian, and square-root of half-Cauchy distributions. 

\section{Real-data analysis}\label{sec5}

We apply the quasi-Bayes EB approach to a dataset of the total number of goals scored by National Hockey League players in the 2017--2018 and 2018--2019 seasons, available from \url{https://www.hockey-reference.com/}. We restrict attention to the $n=745$ players for whom goal totals are available in both seasons. We analyze these data under the Poisson mixture model \eqref{model_distr_poi}, where, for player $i$, $Y_{i}$ is the total number of goals scored in the 2018--2019 season, $X_{i}$ is the total number of goals scored in the 2017--2018 season, and $\theta_i$ represents the latent scoring intensity. In this setting, the $X_{i}$'s are observed data used for estimation, whereas the $Y_{i}$'s are treated as future unobserved data that are used only for out-of-sample validation. Thus, the 2017--2018 season is used to learn the distribution $G$ of latent scoring intensities, and the 2018--2019 season provides a natural future-season benchmark. 

Given the 2017--2018 season observed $X_{1:n}$, we consider the estimation of the following 2018--2019 season functionals: $T_{1,n}(\kappa)=\sum_{1\leq i\leq n} Y_i I(X_i\leq \kappa)$ and $T_{3,n}=\sum_{1\leq i\leq n} I(Y_i<X_i)$, where $\kappa\in\mathbb N$ is a fixed threshold. $T_{1,n}(\kappa)$ represents the total number of goals scored in 2018--2019 season by players whose 2017--2018 season goal totals are below $\kappa$, whereas $T_{3,n}$ represents the the number of players whose goal total decreases from the past season to the future season.

We apply the quasi-Bayes EB approach to the dataset $X_{1:n}$, obtaining estimates of $T_{1,n}(\kappa)$, for $\kappa\in\{2i\text{ : }i=0,1,\ldots,2^{-1}\max\{X_{1:n}\}\}$, and $T_{3,n}$. For the implementation of Newton's algorithm, which provides an estimate $G_n$ of the mixing distribution $G$, we employ the same setting as in Section~\ref{sec4}. Therefore, from \eqref{eq:est}, the quasi-Bayes EB estimates of $T_{1,n}$ and $T_{3,n}$ are
\begin{displaymath}
\hat T^{\text{\tiny{[Q-B]}}}_{1,n}(\kappa)=\sum_{i=1}^n\E_{G_n}[\theta_i\mid X_i] I(X_i\leq \kappa)=\sum_{i=1}^n
I(X_i\leq \kappa)
\int_{\Theta}
\theta \,
G_n(d\theta\mid X_i)
\end{displaymath}
and 
\begin{displaymath}
\hat T^{\text{\tiny{[Q-B]}}}_{3,n}=\sum_{i=1}^n\E_{G_n}\!\left[I(Y_i<X_i)\mid X_i\right]=\sum_{i=1}^n\int_{\Theta}\left\{\sum_{y=0}^{X_i-1}e^{-\theta}\frac{\theta^y}{y!}\right\}G_n(d\theta\mid X_i),
\end{displaymath}
respectively, with the convention that the inner sum in the estimate $\hat T^{\text{\tiny{[Q-B]}}}_{3,n}$ is equal to zero when $X_{i}=0$. Estimates $\hat T^{\text{\tiny{[Q-B]}}}_{1,n}(\kappa)$ and $\hat T^{\text{\tiny{[Q-B]}}}_{3,n}$ are reported in Figure~\ref{fig_realdata_s1} and Table~\ref{tab_realdata_s3}, respectively, together with their absolute deviations (AD) from the true values of $T_{1,n}(\kappa)$ and $T_{3,n}$. Quasi-Bayes EB estimates are compared with: i) the EB estimates $\hat{T}^{\text{\tiny{[ML]}}}_{1,n}(\kappa)$ and $\hat{T}^{\text{\tiny{[ML]}}}_{3,n}$, as well as the Bayes EB estimatse $\hat{T}^{\text{\tiny{[B]}}}_{1,n}(\kappa)$ and $\hat{T}^{\text{\tiny{[B]}}}_{3,n}$, which are obtained along the same lines as in Proposition \ref{prop:poisson}; ii) the ``$u,v$" estimate $\hat{T}^{\text{\tiny[``u,v"]}}_{1,n}(\kappa)$ obtained along the same lines as in \eqref{uv_poiss}.

\begin{figure}[h!]
\begin{center}
\includegraphics[width=1\linewidth,height=0.36\textheight,keepaspectratio]{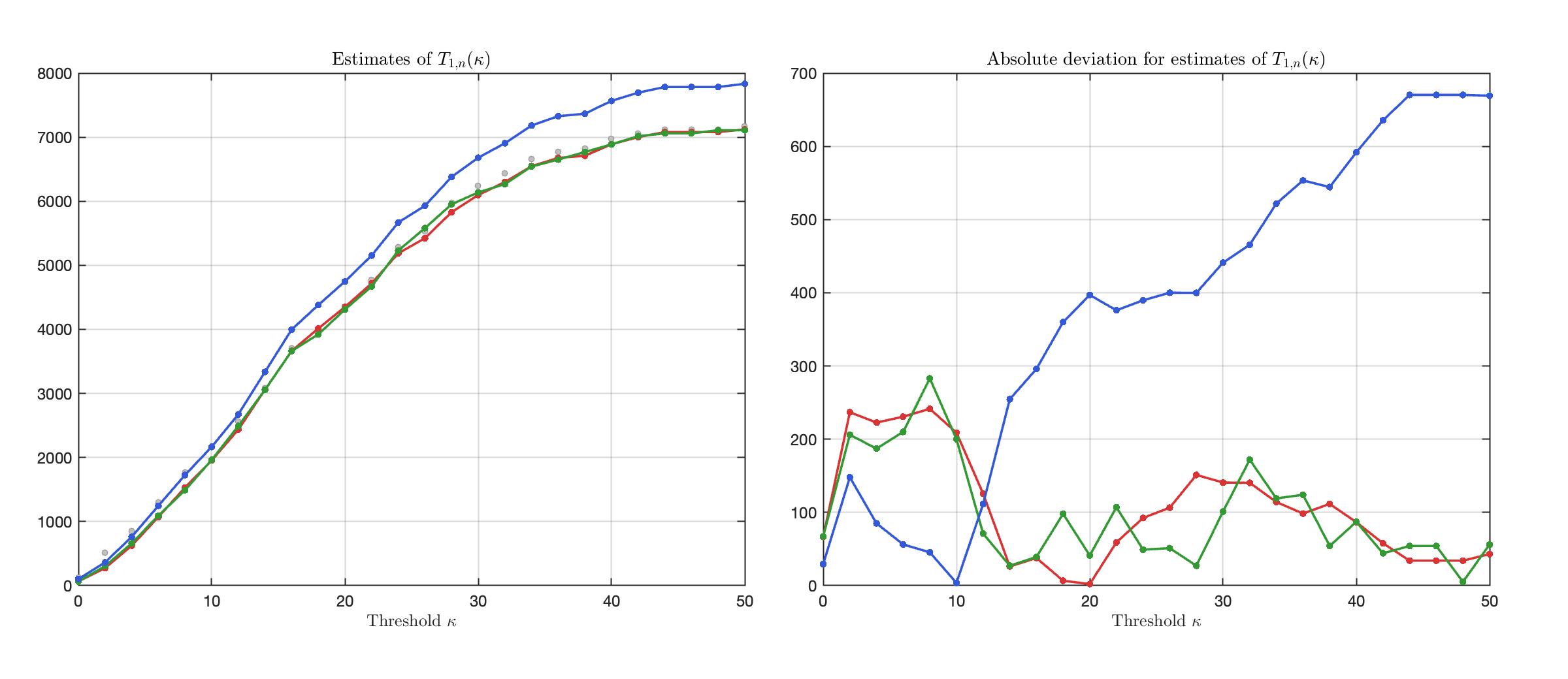}
\end{center}
\caption{\scriptsize{Left panel: true values $T_{1,n}(\kappa)$ (Grey o-) and estimates $\hat{T}^{\text{\tiny{[ML]}}}_{1,n}(\kappa)$ (Blue .-), $\hat{T}^{\text{\tiny{[B]}}}_{1,n}(\kappa)$ (Cyan .-), $\hat{T}^{\text{\tiny[``u,v"]}}_{1,n}(\kappa)$ (Green .-) and $\hat{T}^{\text{\tiny{[Q-B]}}}_{1,n}(\kappa)$ (Red .-). Right panel: AD of $\hat{T}^{\text{\tiny{[ML]}}}_{1,n}(\kappa)$ (Blue .-), $\hat{T}^{\text{\tiny{[B]}}}_{1,n}(\kappa)$ (Cyan .-), $\hat{T}^{\text{\tiny[``u,v"]}}_{1,n}(\kappa)$ (Green .-) and $\hat{T}^{\text{\tiny{[Q-B]}}}_{1,n}(\kappa)$ (Red .-)}}
\label{fig_realdata_s1}
\end{figure}

\begin{table}[ht]
\centering
\caption{Estimates of $T_{3,n}$ and AD}
{
\setlength{\tabcolsep}{0pt}
\begin{tabular}{@{}l@{\hspace{3.6cm}}*{3}{>{\centering\arraybackslash}p{2.15cm}}@{}}
\hline
\hline
& $\hat{T}^{\text{\tiny{[ML]}}}_{3,n}$ & $\hat{T}^{\text{\tiny{[B]}}}_{3,n}$ & $\hat{T}^{\text{\tiny{[Q-B]}}}_{3,n}$ \\[0.1cm]
\hline
\multicolumn{4}{@{}l}{\underline{True value = 318}} \\[0.05cm]
Estimates & 302.4713 & 302.5334 & \textbf{303.0321} \\
AD      & 15.5287 & 15.4666 & \textbf{14.9679} \\[0.1cm]
\hline
\hline
\end{tabular}
}
\label{tab_realdata_s3}
\end{table}

As $\kappa$ increases, Figure~\ref{fig_realdata_s1} shows how well each method estimates the cumulative future scoring contribution of players with increasingly large past-season goal totals. Figure~\ref{fig_realdata_s1} shows that the quasi-Bayes EB estimates remain close to the true values of  $T_{1,n}(\kappa)$ across most thresholds $\kappa$, and the ``$u,v$'' estimates display a similarly stable behavior. By contrast, the EB and Bayes EB estimates tend to overestimate $T_{1,n}(\kappa)$ for larger thresholds. Table~\ref{tab_realdata_s3} shows how well each method estimates the number of players whose goal total decreases from the 2017--2018 season to the 2018--2019 season; this is an aggregate measure of regression to the mean in scoring performance. The quasi-Bayes EB estimate is closer to the value of $T_{3,n}$ than the EB and Bayes EB estimates, yielding the smallest absolute deviation. It is also worth noting that the ``$u,v$'' estimate is not available for $T_{3,n}$, whereas the quasi-Bayes EB construction applies directly. Overall, both the analysis of $T_{1,n}(\kappa)$ and that of $T_{3,n}$ indicate that the quasi-Bayes EB approach provides stable and competitive out-of-sample estimation.

The same analysis reported in Figure~\ref{fig_realdata_s1} and Table~\ref{tab_realdata_s3}, is repeated after stratifying players by their role, namely center (C), left wing (LW), right wing (RW) and defenseman (D). In this case, the quasi-Bayes EB estimates and the competing estimates are applied separately to each role-specific subset of players, and the estimates are reported in Figure~\ref{fig_realdata_s1_role} and Table~\ref{tab_realdata_s3_role}.

\begin{figure}[h!]
\begin{center}
\includegraphics[width=1\linewidth,height=0.56\textheight,keepaspectratio]{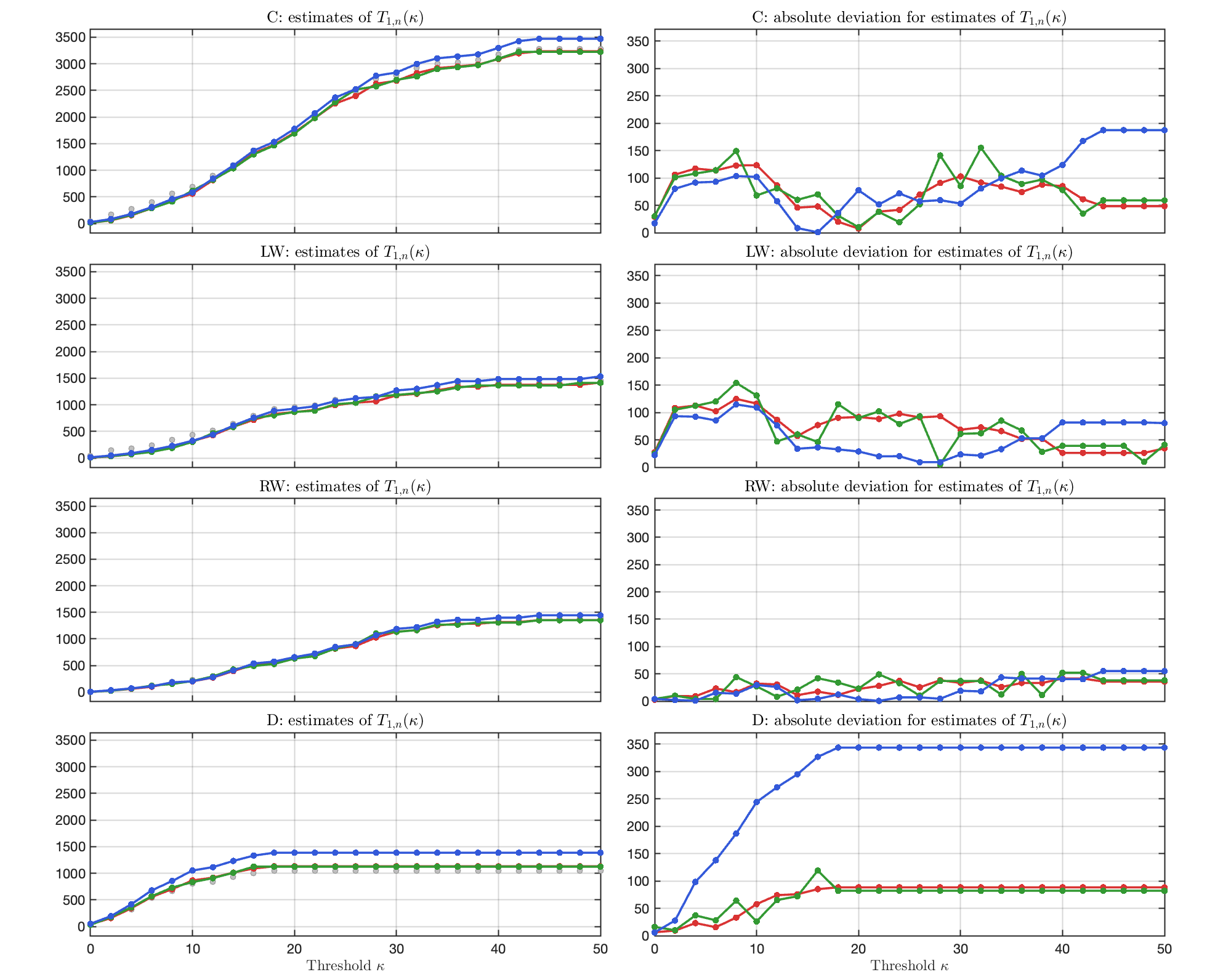}
\end{center}
\caption{\scriptsize{Left panels with centers (C, first row), left wings (LW, second row), right wings (RW, third row) and defenseman (D, fourth row): true values $T_{1,n}(\kappa)$ (Grey o-) and estimates $\hat{T}^{\text{\tiny{[ML]}}}_{1,n}(\kappa)$ (Blue .-), $\hat{T}^{\text{\tiny{[B]}}}_{1,n}(\kappa)$ (Cyan .-), $\hat{T}^{\text{\tiny[``u,v"]}}_{1,n}(\kappa)$ (Green .-) and $\hat{T}^{\text{\tiny{[Q-B]}}}_{1,n}(\kappa)$ (Red .-). Right panels with centers (C, first row), left wings (LW, second row), right wings (RW, third row) and defenseman (D, fourth row): AD of $\hat{T}^{\text{\tiny{[ML]}}}_{1,n}(\kappa)$ (Blue .-), $\hat{T}^{\text{\tiny{[B]}}}_{1,n}(\kappa)$ (Cyan .-), $\hat{T}^{\text{\tiny[``u,v"]}}_{1,n}(\kappa)$ (Green .-) and $\hat{T}^{\text{\tiny{[Q-B]}}}_{1,n}(\kappa)$ (Red .-)}}
\label{fig_realdata_s1_role}
\end{figure}

\begin{table}[ht]
\centering
\caption{Estimates of $T_{3,n}$ and AD}
{
\setlength{\tabcolsep}{0pt}
\begin{tabular}{@{}l@{\hspace{5.6cm}}*{3}{>{\centering\arraybackslash}p{2.15cm}}@{}}
\hline
\hline
& $\hat{T}^{\text{\tiny{[ML]}}}_{3,n}$ & $\hat{T}^{\text{\tiny{[B]}}}_{3,n}$ & $\hat{T}^{\text{\tiny{[Q-B]}}}_{3,n}$ \\[0.1cm]
\hline
\multicolumn{4}{@{}l}{\underline{Centers (C); True value = 109}} \\[0.05cm]
Estimates & 109.7222 & 109.7817 & \textbf{107.6732} \\
MAD      & 0.7222 & 0.7817 & \textbf{1.3268} \\[0.2cm]

\multicolumn{4}{@{}l}{\underline{Right wings (RW); True value = 40}} \\[0.05cm]
Estimates & 43.4283 & 43.4870 & \textbf{43.0061} \\
MAD      & 3.4283 & 3.4871 & \textbf{3.0061} \\[0.2cm]

\multicolumn{4}{@{}l}{\underline{Left wings (LW); True value = 51}} \\[0.05cm]
Estimates & 52.5085 & 52.5692 & \textbf{52.0260} \\
MAD      & 1.5085 & 1.5692 & \textbf{1.0260} \\[0.2cm]

\multicolumn{4}{@{}l}{\underline{Defensman (D); True value = 118}} \\[0.05cm]
Estimates & 101.7880 & 101.8582 & \textbf{100.5159}\\
MAD      & 16.2120 & 16.1418 & \textbf{17.4841} \\[0.1cm]

\hline
\hline
\end{tabular}
}
\label{tab_realdata_s3_role}
\end{table}


\section{Discussion and concluding remarks}\label{sec6}

\subsection{Examples with Gaussian kernel}

Although the examples in the paper are developed for the Poisson kernel, the proposed quasi-Bayes EB methodology can be applied to other kernels, a natural example being the Gaussian kernel. Let $n\geq1$ and consider real-valued random vectors $(X_1,\theta_1),\ldots,(X_n,\theta_n)$ distributed as follows:
\begin{align}\label{model_distr_gaussian}
X_i\mid\theta_i & \quad\simind\quad \mathrm{Gaussian}(\cdot\mid\theta_i,1),
\qquad i=1,\ldots,n,\\[-0.2cm]
\notag
\theta_i & \quad\simiid\quad G,
\end{align}
where $\mathrm{Gaussian}(\cdot\mid\theta,1)$ denotes the Gaussian kernel with mean $\theta$ and variance $1$, and $G$ is an unknown mixing distribution on $\Theta\subset\mathbb{R}$. Following \citet[Section 2.4]{Zha(05)}, in \ref{app_gauss} we consider the estimation of $S_{1,n}=\sum_{1\leq i\leq n}^n \theta_i I(X_i\leq\kappa)$ and $S_{3,n}=\sum_{1\leq i\leq n} I(X_i>\theta_i)$, where $\kappa\in\mathbb{R}$ is a fixed threshold. For these functionals, the quasi-Bayes EB estimate is compared with parametric EB and Bayes EB estimates, as well as with the ``$u,v$'' estimate, which is available for $S_{3,n}$ only. 

The numerical results reported in \ref{app_gauss} show broadly comparable performance across the competing estimates. For $S_{1,n}$, the quasi-Bayes EB estimate is competitive with the parametric estimates, while the classical ``$u,v$'' estimator is not available. For $S_{3,n}$, all estimates, including the ``$u,v$'' estimate, achieve similar empirical accuracy, with no method clearly dominating the others. Overall, the quasi-Bayes EB estimate remains competitive for both functionals.

\subsection{Multidimensional extension}

Throughout the paper, we have focused on univariate mixture models, where $\theta$ take values in a subset of $\mathbb{R}$. The proposed quasi-Bayes EB methodology extends naturally, with only minor modifications, to multidimensional settings under a coordinate-wise independence assumption. For $i=1,\ldots,n$, let $\Xb_i=(X_{i,1},\dots,X_{i,d})$ be a $d$-dimensional random vector with independent coordinates, where each $X_{i,j}$ has density function $k(\cdot\mid \theta_{i,j})$. Writing $\thetab_i=(\theta_{i,1},\dots,\theta_{i,d})\in\Theta^d$,
\begin{displaymath}
k(\xb_i\mid \thetab_i)=\prod_{j=1}^k k(x_{i,j}\mid \theta_{i,j}),
\end{displaymath}
and the corresponding marginal density is obtained by mixing over a distribution $G$ on $\Theta^d$. Given a measurable vector-valued utility function $\ub:\mathbb X^d\times\Theta^d\to\mathbb R^s$, we consider the estimation of
\begin{displaymath}
\Sb_n=\sum_{i=1}^n \ub(\Xb_i,\thetab_i).
\end{displaymath}
In this multidimensional framework, Newton’s algorithm extends verbatim, yielding a recursive update of the mixing distribution $G_n$ on $\Theta^d$. The quasi-Bayes EB estimate $\hat{\Sb}^{\text{\tiny{[Q-B]}}}_n$ of $\Sb_n$ is then obtained by replacing $G$ with $G_n$ in the posterior mean $\E_{G}[\Sb_n\,|\,\mathbf{X}_{1:n}]$. See \ref{app_mult} for details.

All the theoretical properties established in Sections \ref{sec2} and \ref{sec3} extend to the multidimensional framework with minor technical modifications. While the computational cost increases with the dimension $d$, Newton's algorithm preserves constant per-observation complexity.

\subsection{Alternative nonparametric EB strategies}

An important direction for future research is the development and theoretical analysis of alternative nonparametric $g$-modeling strategies for estimating sums of random variables $S_n$ under the mixture model \eqref{eq:model}. A natural approach is nonparametric maximum likelihood estimation of $G$ \citep{Lin(95)}, which yields a principled plug-in EB procedure, though it may be computationally demanding. Another promising alternative is nonparametric Bayes EB, which places a prior on $G$, for instance via Dirichlet process mixture models \citep{Fer(73),Lo(84)}. While such methods provide EB estimates of $S_n$ based on posterior means and offer substantial flexibility, they typically require MCMC or related sampling techniques for implementation \citep{Fav(13)}. Establishing rigorous theoretical guarantees for these nonparametric EB approaches and systematically comparing them with the quasi-Bayes EB methodology developed here remains an open and interesting problem.

From a purely empirical perspective, for the estimation of $S_{1,n}$ and $S_{3,n}$ under the Poisson mixture model \eqref{model_distr_poi}, we implemented a nonparametric EB strategy by replacing $\tilde{G}$ with its nonparametric maximum likelihood estimation, computed via the vertex direction algorithm \citep{Lin(95),Jan(24)}. Numerical results are reported in \ref{app_neb}, including a detailed comparison with respect to the quasi-Bayes EB and the ``$u,v$'' estimates from Sections \ref{sec4}. For $S_{1,n}$, the resulting nonparametric EB estimates are competitive with both the quasi-Bayes EB and the ``$u,v$'' estimates. In contrast, for $S_{3,n}$ the nonparametric EB estimates perform substantially worse than the quasi-Bayes EB estimates, while ``$u,v$'' estimates are not available in this case. We also report CPU times, which indicate that the nonparametric EB approach is noticeably computationally slower than the quasi-Bayes EB procedure.

\subsection{Specie sampling problems}

Another promising direction for future research is to adapt the quasi-Bayes EB methodology developed in this paper to species sampling problems \citep{Goo(53),Goo(56),Efr(76),Mao(02)}. In this setting, quantities of interest such as the number of species in the population, the missing mass, and the number of unseen species can be expressed as sums of random variables under Poisson mixture models. Unlike the examples considered in Section~\ref{sec4}, however, the associated ``$u,v$'' integral equations do not admit solutions, and corresponding $u,v$ estimates are therefore unavailable \citep[Section~3.2]{Zha(05)}. While parametric EB approaches remain feasible under parametric assumptions on the mixing distribution, developing quasi-Bayes EB estimates with theoretical guarantees in this setting would be of independent interest. Additional details on the species sampling formulation and its connection with EB estimation of sums are provided in \ref{app_ssp}.

\subsection{Further directions for future research}

As in \citet{Fav(24)}, establishing convergence rates for quasi-Bayes EB estimates currently relies on discretizing the parameter space $\Theta$, a standard and practically convenient assumption that would ideally be avoided; obtaining rate results in infinite-dimensional settings would require new tools to control stochastic approximation errors without finite-dimensional approximations. In addition, while this paper adopts a $g$-modeling strategy, an alternative $f$-modeling approach, directly modeling the marginal distribution of the observations, remains largely unexplored within the quasi-Bayes framework and may offer complementary insights. Finally, although the quasi-Bayes approach yields asymptotic credible intervals for the target sum, their frequentist coverage properties remain unknown; identifying conditions under which such intervals achieve valid coverage is an important open problem.



\section*{Appendix}
\renewcommand{\thesection}{\Alph{section}}
\renewcommand{\theequation}{\thesection.\arabic{equation}}
\setcounter{section}{1}
\setcounter{equation}{0}
\setcounter{thm}{0}

\section{Proofs of Section 2}

\subsection{Proof of Theorem \ref{th:convg}}\label{appB}
 Conditional to $\tilde G$, the random variables $X_n$ are i.i.d. with density function  $f_{\tilde G}(x)=\int_\Theta k(x\mid\theta)\tilde G(\ddr\theta)$ with respect to $\lambda$. Let $\mathcal H$ be a countable convergence determining class of bounded continuous functions on $\Theta$. \\
i) By Corollary 4.7 in \cite{MarTok(09)}, for every $h\in\mathcal H$,
$$
\P\left(\int_\Theta h(\theta)G_n(\ddr\theta)\rightarrow \int h(\theta)\tilde G(\ddr\theta)\mid \tilde G\right)=1,
$$
which implies, by countability of $\mathcal H$,  that
$$
\P\left(\int_\Theta h(\theta)G_n(\ddr\theta)\rightarrow \int h(\theta)\tilde G(\ddr\theta)\;\forall h\in\mathcal H\right)=1.
$$
\\
ii) By Corollary 4.6 in \cite{MarTok(09)}, 
$$
\P\left(
\lambda\{x:f_{G_n}(x)\not\rightarrow f_{\tilde G}(x)\}=0\mid\tilde G
\right)=1,
$$
and
$$
\P\left(\int |f_{G_n}(x)-f_{\tilde G}(x)|\lambda(\ddr x){\rightarrow}0\mid \tilde G\right)=1\quad \P\mbox{-a.s.}
$$
The above equations imply that
$$
\P\left(
\lambda\{x:f_{G_n}(x)\not\rightarrow f_{\tilde G}(x)\}=0
\right)=1.
$$
and
$$
\P\left(\int |f_{G_n}(x)-f_{\tilde G}(x)|\lambda(\ddr x){\rightarrow}0\right)=1.
$$
       \\
       iii) By i) and ii), for every $h\in\mathcal H$,
       $$\frac{\int_\Theta h(\theta)k(x\mid\theta)G_n(\ddr\theta)}{f_{G_n}(x)}
       \rightarrow \frac{\int_\Theta h(\theta)k(x\mid\theta)\tilde G(\ddr\theta)}{f_{\tilde G}(x)}\quad \P\mbox{-a.s.}$$
       Since $\mathcal H$ is countable, the above equation holds on a set of probability one for every $h\in\mathcal H$.


\subsection{Proof of  Theorem \ref{th:central}}\label{app:central}

\begin{lem}\label{lemma:sigman}
    Under the assumptions of Theorem \ref{th:central},
    \begin{align}
        &\frac 1 n  \sum_{j=1}^n \sigma_{n,j}^2 \rightarrow 
        \iint_\Theta (u(x,\theta)-\tilde u(x))^2\tilde G(\ddr\theta\mid x)f_{\tilde G}(\ddr x)\lambda(\ddr x)\quad \PP\mbox{-a.s.}\nonumber
    \end{align}
\end{lem}
\begin{proof}
Denoting 
$$
u_n(x)=\int_\Theta u(x,\theta)G_n(\ddr\theta\mid x),\quad\quad v_n(x)=\int_\Theta u^2(x,\theta)G_n(\ddr\theta\mid x),
$$
we can write
    \begin{align*}
        \frac 1 n  \sum_{j=1}^n \sigma_{n,j}^2
&=\frac 1 n \sum_{j=1}^n v_n(X_j)
-\frac 1 n \sum_{j=1}^n u_n^2(X_j),
      \end{align*}
      and the claim holds if
      $$
      \frac 1 n \sum_{j=1}^n v_n(X_j)\rightarrow\int\tilde v(x)f_{\tilde G}(x)\lambda(\ddr x)\quad \mbox{ and }\quad
      \frac 1 n \sum_{j=1}^n u_n^2(X_j)\rightarrow\int\tilde u^2(x)f_{\tilde G}(x)\lambda(\ddr x)\quad\PP\mbox{-a.s.}
      $$
      where
      $$
\quad\quad\tilde v(x)=\int_\Theta u^2(x,\theta)\tilde G(\ddr\theta\mid x),\quad\quad \tilde u(x)=\int_\Theta u(x,\theta)\tilde G(\ddr\theta\mid x).
$$
    We prove the convergence separately for the two terms.    
    \\
By assumption $(\mathcal A3)$,  $k(x\mid\theta)$ is continuous in $\theta$ for $\lambda$ a.e. $x$, and dominated by $h(x)\in L^1(\lambda)$.  Then $T(\theta)=\int_{\cdot}k(x\mid\theta)\lambda(\ddr x)$ is a continuous map defined on $\Theta$ and taking values in the space of probability measures on $\mathbb R$ with the weak topology. Indeed, if $\theta_n\rightarrow\theta$, then for every bounded continuous function $r(x)$, 
$$
\int r(x)k(x\mid\theta_n)\lambda(\ddr x)\rightarrow \int r(x)k(x\mid\theta)\lambda(\ddr\theta),$$
where the convergence holds since $r(x)k(x\mid\theta_n)$ converges to $r(x)k(x\mid\theta)$ for $\lambda$-a.e. $x$ and  is dominated by the function  $r(x)h(x)\in L^1(\lambda)$. 
Since $\Theta$ is compact and $T$ is continuous, then $T(\Theta)$ is compact. Therefore the class $\{\int_\cdot k(x\mid\theta)\lambda(\ddr x):\theta\in\Theta\}$ is tight. It follows that for every $\epsilon >0$ there exists a compact set $K\subset\mathbb R$ such that 
\begin{equation}
    \label{eq:compact}
\sup_{\theta\in\Theta}\int_{K^c}k(x\mid \theta)\lambda(\ddr x)<\frac{\epsilon}{||u||_\infty^2},
\end{equation}
where $||\cdot ||_\infty$ denotes the sup-norm, which implies that
$$
\PP[X_j\in K^c\mid\tilde G]=\int_{K^c}\int_\Theta k(x\mid\theta)\tilde G(\ddr\theta)\lambda(\ddr x)<\frac{\epsilon}{||u||_\infty^2}\quad\PP\mbox{-a.s.}
$$
Fix $\epsilon$ and let $K$ be such that \eqref{eq:compact} holds.
Then for $\omega$ in a set of probability one:
$\PP[X_j\in K^c\mid\tilde G](\omega)<\epsilon$; $G_n(\omega)$ converges weakly to $\tilde G(\omega)$; and the function $\theta\mapsto u^2(x,\theta)$ is $\tilde G(\omega)$-a.s. continuous. Fix $\omega$ in this set. In what follows, we omit the dependence on $\omega$. \\
The function $\theta\mapsto u^2(x,\theta)$ is bounded and $\tilde G$-a.s. continuous. Thus, by Theorem \ref{th:convg}, $v_n(x)$ converges to $\tilde v(x)$, $\lambda$-a.e. Moreover,   by assumption (ii) and $G_n\in\mathbb G$, the functions $v_n$ are equicontinuous on the compact set $K$. By the Ascoli-Arzel\`a theorem the sequence $(v_n)$ is relatively compact
in $(C(K),\|\cdot\|_\infty)$. Hence, every subsequence of $(v_n)$ admits a further
subsequence that converges uniformly on $K$ to some continuous limit function.
Since $v_n(x)\to \tilde v(x)$ uniformly for $\lambda$-almost every $x$ and the limit is unique,
it follows that the whole sequence $(v_n)$ converges uniformly to $\tilde v$ on $K$.
This implies that
$$
\frac 1 n \sum_{j=1}^n (v_n(X_j)-\tilde v(X_j))\indicator(X_j\in K)\rightarrow 0,
$$
as $n\rightarrow\infty$. 
Moreover, for every $n$,
$$
\frac 1 n \sum_{j=1}^n (v_n(X_j)-\tilde v(X_j))\indicator(X_j\in K^c)<\epsilon.
$$
Since $\epsilon$ is arbitrary, we conclude that
\begin{equation}\label{eq:vntov}
\frac 1 n \sum_{j=1}^n (v_n(X_j)-\tilde v(X_j))\rightarrow 0\quad\PP\mbox{-a.s.}
\end{equation}
On the other hand, since the $X_j$ are conditionally i.i.d., given $\tilde G$ then
$$
\PP\left[\frac 1 n \sum_{j=1}^n\tilde v(X_j)\rightarrow \int\tilde v(x)f_{\tilde G}(x)\lambda(\ddr x)\mid\tilde G\right]=1,
$$
which implies that
$$
\frac 1 n \sum_{j=1}^n\tilde v(X_j)\rightarrow \int\tilde v(x)f_{\tilde G}(x)\lambda(\ddr x)\quad\PP\mbox{-a.s.}
$$
This, together with \eqref{eq:vntov} implies that
$$
\frac 1 n \sum_{j=1}^n v_n(X_j)\rightarrow \int\tilde v(x)f_{\tilde G}(x)\lambda(\ddr x)\quad\PP\mbox{-a.s.}
$$
The same argument, applied to $u_n^2(x)$ leads to 
$$
\frac 1 n \sum_{j=1}^n u_n^2(X_j)\rightarrow \int\tilde u^2(x)f_{\tilde G}(x)\lambda(\ddr x)\quad \PP\mbox{-a.s.}
$$
\end{proof}


\begin{proof}[Proof of theorem \ref{th:central}]
With the same notations as in Lemma \ref{lemma:sigman}, we can write
$$
S_n-\hat{S}^{\text{\tiny{[Q-B]}}}_{n}=\sum_{j=1}^n\left( u(X_j,\theta_j)-u_n(X_j)\right).
$$
For every $n\geq 1$ and $j=1,\dots,n$, let
$$
Y_{n,j}=\frac{u(X_j,\theta_j)-u_n(X_j)}{\left(\sum_{j=1}^n \sigma_{n,j}^2\right)^{1/2}}.
$$
Since $\sigma_{n,j}^2$ is measurable with respect to the sigma-algebra generated by $X_{1:n}$, then, under $\Pn(\cdot\mid X_{1:n})$, the random variables $Y_{n,1},\dots,Y_{n,n}$ are independent and centered.   \\
By Lemma \ref{lemma:sigman},
$$
\frac 1 n\sum_{j=1}^n\sigma^2_{n,j}\rightarrow \tilde{\sigma}^2\quad\PP\mbox{-a.s.}
$$
with $\tilde \sigma^2>0$ by assumption (iii).
Thus,
\begin{align*}\label{appeq:ineq}
&\limsup_{{n}\rightarrow\infty}\sum_{j=1}^{n} \E^{(n)}(Y_{{n},j}^2 \indicator{(|Y_{{n},j}|>\epsilon)}\mid X_{1:n})\\
&=
\limsup_{{n}\rightarrow\infty}\frac{\sum_{j=1}^{n} 
\int (u(X_j,\theta)-u_{n}(X_j))^2 \indicator{((u(X_j,\theta)-u_{n}(X_j))^2>\epsilon^2\sum_{j=1}^{n} \sigma_{{n},j}^2)}G_{n}(\ddr\theta\mid X_j)}{\sum_{j=1}^{n} \sigma_{{n},j}^2}\nonumber \\
&= \frac{1}{\tilde{\sigma}^2}\int \limsup_{{n}\rightarrow\infty}\frac{1}{{n}} \sum_{j=1}^{n} 
 (u(X_j,\theta)-u_n(X_j))^2 \indicator((u(X_j,\theta)-u_n(X_j))^2>\epsilon^2 \sum_{j=1}^n\sigma^2_{n,j})G_n(\ddr\theta\mid X_j)\quad\PP\mbox{-a.s.}\nonumber
\end{align*}
Since the function $u$ is bounded and $\sum_{j=1}^n \sigma^2_{n,j}$ diverges $\PP$-a.s., then the limit is zero.
By a conditional version of the Lindeberg-Feller central limit theorem
$$
\mathcal L_{\P^{(n)}}\left(\sum_{j=1}^{n} Y_{n,j}\mid X_{1:n}\right)\rightarrow\mathcal N(0,1)\quad \P-a.s.
$$
\end{proof}
\section{Proofs and auxiliary results for Section 3}

\subsection{Proof of Theorem \ref{th:consist}}\label{app:proofconsistency}
Define
$
u_G(x)=\;\int_\Theta u(x,\theta)\,G(d\theta\mid x)
$ and $
f_G(x)=\int_\Theta k(x\mid\theta)\,G(d\theta),
$ for any $G\in{\mathbb G}$, and let 
\[
\mathcal U =\{\,u_G(\cdot):\, G\in{\mathbb G}\,\}.
\]
We first prove that $\mathcal U$ has an integrable envelope with respect to the measure $F_{G^*}(\ddr x):=f_{G^*}(x)\lambda(\ddr x)$. At this aim, denote
$
A= \sup_{\theta\in\Theta} a(\theta)$ and  $M(x) = A b(x).
$
Then, by \eqref{eq:cond1}, $$\sup_{G\in \mathbb G}|u_G(x)|=\sup_{G\in\mathbb G}\int |u(x,\theta)|G(\ddr\theta\mid x)\leq M(x)$$ and $$\int M(x)f_{G^*}(x)\lambda(\ddr x)<\infty,$$ 
proving the envelope property. 
Denoting for $n\geq 1$
$$
u_n(x)=\frac{\int u(x,\theta)k(x\mid\theta)G_n(\ddr\theta)}{f_{G_n}(x)},\quad e_n(\ddr x)=\frac 1 n \sum_{i=1}^n \delta_{X_i}(\ddr x),
$$
we can write
\begin{align*}
   & \frac 1 n |\hat{S}^{\text{\tiny{[Q-B]}}}_{n}-\hat S_n^{G^*}|=
   |\int u_n(x)e_n(\ddr x)-\int u_{G^*}(x)e_n(\ddr x)|\\
&\leq|\int u_n(x)e_n(\ddr x)-\int u_{G^*}(x)F_{G^*}(\ddr x)|+
|\int u_{G^*}(x)e_n(\ddr x)-\int u_{G^*}(x)F_{G^*}(\ddr x)|.
\end{align*}
Since $|u_{G^*}(x)|\leq M(x)$ for all $x$ and since $M$ is integrable with respect to $F_{G^*}$, then the last term converges to zero $\P^*$-a.s. as $n\rightarrow\infty$ by the strong law of large numbers. On the other hand, 
\begin{align*}
    &|\int u_n(x)e_n(\ddr x)-\int u_{G^*}(x)F_{G^*}(\ddr x)|\\
    &\leq 
\left|\int u_n(x)e_n(\ddr x)-\int u_n(x)F_{G^*}(\ddr x)\right|+
\int |u_n(x)-u_{G^*}(x)|F_{G^*}(\ddr x)
\end{align*}
To prove that both terms of the sum converge to zero $\P^*$-a.s. as $n\rightarrow\infty$, we first prove that $u_n(x)$ converges to $u_{G^*}(x)$, $\lambda$-a.e. as $n\rightarrow\infty$. 
By Corollary 4.6 in  \cite{MarTok(09)}
$$
f_{G_n}(x)=\int k(x\mid\theta)G_n(\ddr\theta)\rightarrow f_{G^*}(x)=\int k(x\mid\theta)G^*(\ddr\theta)\quad \lambda\mbox{-a.e. and in }L^1,
$$
as $n\rightarrow\infty$. Moreover, by Corollary 4.7 in \cite{MarTok(09)}, $G_n$ converges to $G^*$ $\PP^*$-a.s. in the weak topology. For a fixed $x$,  
$$
\sup_{G\in\mathbb G}\int_\Theta|u(x,\theta)|k(x\mid\theta)G(\ddr\theta)\leq A\,b(x)\, \sup_{x,\theta}k(x\mid\theta) .
$$
Hence, the function $u(x,\theta)k(x\mid\theta) $ is uniformly integrable with respect to $G_n$. Moreover, as a function of $\theta$, it is $G^*$-a.s.
 continuous by assumption. Therefore,
$$
\int u(x,\theta)k(x\mid\theta)G_n(\ddr\theta)\rightarrow \int u(x,\theta)k(x\mid\theta)G^*(\ddr\theta), \quad \lambda\mbox{-a.e.}
$$
It follows that $u_n(x)$ converges to $u_{G^*}(x)$, $\lambda$-a.e. as $n\rightarrow\infty$.\\ 
The map $G\mapsto u_G$ is continuous as a map
from ${\mathbb G}$ with the weak topology into $L^1(F_{G^*})$. Indeed, if $G_n'$ is a sequence in $\mathbb G$ converging weakly to $G'\in\mathbb G$, then
\begin{align*}
&\int \left| u_{G_n'}(x)  -u_{G'}(x) \right|F_{G^*}(\ddr x)\\&
=\int \left| \frac{\int_\Theta u(x,\theta)k(x\mid\theta)G_n'(\ddr\theta)}{\int_\Theta k(x\mid\theta)G_n'(\ddr\theta)}    -\frac{\int_\Theta u(x,\theta)k(x\mid\theta)G'(\ddr\theta)}{\int_\Theta k(x\mid\theta)G'(\ddr\theta)}   \right|F_{G^*}(\ddr x)\rightarrow 0,
\end{align*}
where the convergence holds since $\mathcal U$ has an integrable envelope with respect to $F_{G^*}$. \\  Since the map $G\mapsto u_G$ is continuous and ${\mathbb G}$ is precompact by ($\mathcal A$1),
its image $\mathcal U$ is precompact (hence totally bounded) in $L^1(F_{G^*})$ (continuous image of a precompact set is precompact).
Therefore $\mathcal U$ admits finite $L^1(F_{G^*})$-bracketing numbers at every
$\varepsilon>0$, and it follows that
$\mathcal U$ is a $F_{G^*}$-Glivenko--Cantelli class; see Theorem 2.4.1 in
van der Vaart and Wellner (1996).
Consequently,
\[
\sup_{v\in\mathcal U}\Big|\int v(x)e_n(\ddr x)-\int v(x)F_{G^*}(\ddr x)\Big|\;\longrightarrow\;0
\qquad P^*\text{-a.s.},
\]
and in particular,
\[
\Big|\int u_{G_n}(x)e_n(\ddr x)-\int u_{G_n}(x)F_{G^*}(\ddr x)\Big|
\;\le\;
\sup_{v\in\mathcal U}\Big|\int v(x)e_n(\ddr x)-\int v(x)F_{G^*}(\ddr x)\Big|
\;\longrightarrow\;0
\qquad P^*\text{-a.s.}
\]
On the other hand, since $|u_n(x)-u_{G^*}(x)|$ are dominated by $2M(x)\in L^1(F_{G^*})$, then by dominated convergence theorem
$$
\int |u_n(x)-u_{G^*}(x)|F_{G^*}(\ddr x)\rightarrow 0,
$$
as $n\rightarrow\infty$.

\subsection{Proof of Theorem \ref{th:regret}}\label{app:proofregret}

In this section, we use the same letter in lower case to denote the probability mass function associated with a given probability distribution. Moreover we will occasionally represent probability mass functions $g$ on $\Theta_\diamond=\{\vartheta_1,\dots,\vartheta_{d}\}$ as vectors $\gb$, taking values either in the simplex $\Delta_d=\{\vb=(v_1,\dots,v_{{d}-1}):v_i\geq 0 \;(i=1,\dots,{d}-1),\;\sum_{i=1}^{{d}-1}v_i\leq 1\}$,  or in $\mathbb R^{{d}}$, depending on the context, which will be specified when needed. Similarly, the kernel $k$ is represented as a vector in $\Delta_d$ or $\mathbb R^{{d}}$, by defining $[\kb_x]_j=k(x\mid\vartheta_j)$. We denote component-wise multiplication of vectors by $\circ$ and the Euclidean norm by $||\cdot||$. Let $\mu$ denote the counting measure on $\Theta_\diamond$.

The proof of Theorem \ref{th:regret} is based on the following Lemma \ref{lem:ineq*}, and on the rate of convergence of $g_n$ to $g_\diamond^*$.

\begin{lem}\label{lem:ineq*}
Under the assumptions of Theorem \ref{th:regret}, 
    \begin{equation}\label{eq:ineq*}
        \frac{1}{n}|\hat{S}^{\text{\tiny{[Q-B]}}}_{n}-\hat S_n(G^*_\diamond)| =\sup_{\theta\in{\Theta_\diamond}}|g_n(\theta)-g^*_\diamond(\theta)|\;O(1)\quad \PP^*\mbox{-a.s}.
    \end{equation}
    \end{lem}
\begin{proof}
    We can write that
    \begin{align*}
        &\frac 1 n |\hat{S}^{\text{\tiny{[Q-B]}}}_{n}-\hat S_n(G^*_\diamond)|
        =\left|\frac 1 n \sum_{j=1}^n\int_{\Theta_\diamond} u(X_j,\theta)(g_n(\theta\mid X_j)-g^*_\diamond(\theta\mid X_j))\mu(\ddr\theta)\right|\\
        &=\bigl|\int_{\Theta_\diamond} \frac 1 n \sum_{j=1}^n u(X_j,\theta)k(X_j\mid\theta)\left(\frac{g_n(\theta)}{f_{G_n}(X_j)}-\frac{g^*_\diamond(\theta)}{f_{G^*_\diamond}(X_j)}\right)\mu(\ddr\theta)\bigr|\\
        &\leq \int_{\Theta_\diamond} \frac 1 n \sum_{j=1}^n |u(X_j,\theta)|k(X_j\mid\theta)\frac{|g_n(\theta)-g^*_\diamond(\theta)|}{f_{G^*_\diamond}(X_j)}\mu(\ddr\theta)\\
        &\quad\quad +\int_{\Theta_\diamond} \frac 1 n \sum_{j=1}^n |u(X_j,\theta)|k(X_j\mid\theta)g_n(\theta)\bigl|\frac{1}{f_{G_n}(X_j)}-\frac{1}{f_{G^*_\diamond}(X_j)}\bigr|\mu(\ddr\theta)\\
        &\leq \sup_{\theta\in {\Theta_\diamond}}|g_n(\theta)-g^*_\diamond(\theta)|\int_{\Theta_\diamond} \frac 1 n \sum_{j=1}^n |u(X_j,\theta)|\frac{k(X_j\mid\theta)}{f_{G^*_\diamond}(X_j)}\mu(\ddr\theta)\\
        &\quad\quad +\int_{\Theta_\diamond}\frac 1 n \sum_{j=1}^n\frac{|u(X_j,\theta)|k(X_j\mid\theta)g_n(\theta)}{f_{G_n}(X_j)f_{G^*_\diamond}(X_j)}\\
        &\quad\quad \quad\quad \int_{\Theta_\diamond} k(X_j\mid\theta')|g_n(\theta')-g^*_\diamond(\theta')|\mu(\ddr\theta')\mu(\ddr\theta)\\
        &\leq \sup_{\theta\in {\Theta_\diamond}}|g_n(\theta)-g^*_\diamond(\theta)|\\
        &
       \quad\quad \cdot \frac 1 n \sum_{j=1}^n \bigl(\int_{\Theta_\diamond} |u(X_j,\theta)|\frac{k(X_j\mid\theta)}{f_{G^*_\diamond}(X_j)}\mu(\ddr\theta)\\
       &\quad\quad \quad\quad +
       \int_{\Theta_\diamond}\frac{|u(X_j,\theta)|k(X_j\mid\theta)g_n(\theta)}{f_{G_n}(X_j)f_{G^*_\diamond}(X_j)}\int_{\Theta_\diamond} k(X_j\mid\theta')\mu(\ddr\theta')\mu(\ddr\theta)\bigr)\\
       &\leq  (1+dC')\sup_{\theta\in {\Theta_\diamond}}|g_n(\theta)-g^*_\diamond(\theta)|
       \frac 1 n \sum_{j=1}^n \left(\int_{\Theta_\diamond} |u(X_j,\theta)|\frac{k(X_j\mid\theta)}{f_{G^*_\diamond}(X_j)}\mu(\ddr\theta)\right),
    \end{align*}
    with $C'=\sup_{x,\theta}k(x\mid\theta)$. The thesis follows from the strong law of large number, since
    \begin{align*}
        &\E^*\left[\int_{\Theta_\diamond} |u(X_j,\theta)|\frac{k(X_j\mid\theta)}{f_{G^*_\diamond}(X_j)}\mu(\ddr\theta)\right]
        =\int_{\mathbb R} \int_{\Theta_\diamond} |u(x,\theta)|\frac{k(x\mid\theta)}{f_{G^*_\diamond}(x)}f_{G^*}(x)\lambda(\ddr x)<\infty,
    \end{align*}
    by the assumption \eqref{eq:finitesecder} and boundedness of $u$.
\end{proof}


The proof for the rate of convergence of $\sup_\theta|g_n(\theta)-g_\diamond^*(\theta)|$
is based on stochastic approximation theory.



\begin{lem}\label{lem:stochapp3}
Under the assumptions of Theorem \ref{th:regret}, for every $\zeta<2-1/\gamma$
    $$
    ||\gb_n-\gb_\diamond^*||^2=o(n^{-\zeta})\quad P^*-a.s.
    $$
\end{lem}

\begin{proof}
 By Corollary 4.7 in \cite{MarTok(09)}, $\gb_n$ converges to $\gb_\diamond^*$. Moreover, $(\gb_n)$ satisfies the 
 stochastic approximation in $\mathbb R^{d}$:
\begin{align*}
\gb_{n+1}&=\gb_{n}+\alpha_{n+1} \gb_{n}\circ \left(\frac{\kb_{X_{n+1}}}{f_{g_{n}}(X_{n+1})}-\mathbf 1\right)\\
&=\gb_{n}+\alpha_{n+1} \hb (\gb_{n})+\alpha_{n+1}\epsilonb_{n+1},
\end{align*}
with initial value $\gb_0$, where
$$
\hb (\gb)=\gb\circ \left(
\int \frac{\kb_x}{f_{\gb}(x)}
f_{G^*}(x)\lambda(\ddr x)-\bf 1
\right),
$$
and
$$
\epsilonb_{n+1}=\gb_n\circ \left(
\frac{\kb_{X_{n+1}}}{f_{\gb_n}(X_{n+1})}
-
\int \frac{\kb_x}{f_{\gb_n}(x)}
f_{G^*}(x)\lambda(\ddr x)
\right).
$$
   The thesis follows from Theorem 3.1.1 in \cite{Chen(02)} if we can prove that the following conditions hold:\begin{itemize}
        \item[$(\mathcal B1)$] $\alpha_n\rightarrow 0$, $\sum_{n=1}^\infty \alpha_n=\infty$, $\alpha_{n+1}^{-1}-\alpha_n^{-1}\rightarrow 0$;
        \item[$(\mathcal B2)$] $\sum_{n=1}^\infty \alpha_n^{1-\zeta/2}\epsilon_n$ converges, $P^*$-a.s.;
        \item[$(\mathcal B3)$] $\hb $ is measurable and locally bounded, and is differentiable at $\gb_\diamond^*$. Let $\Hb $ be a matrix such that, as $\gb\rightarrow \gb_\diamond^*$,
        $$
        \hb (\gb)=\Hb (\gb-\gb_\diamond^*)+\rb(\gb),\quad \rb(\gb_\diamond^*)=\mathbf 0,\quad ||\rb(\gb)||=o(||\gb-\gb_\diamond^*||).
        $$
        All the eigenvalues of the matrix $\Hb $ have negative real parts.
    \end{itemize}

\noindent{\bf Remark:}        Theorem 3.1.1 in \cite{Chen(02)} also requests:\\
  $(\mathcal  B4)$  {\em   There exists a continuously differentiable function $\ell:\Delta^{d}\rightarrow \mathbb R$ such that $$\sup_{\delta_1\leq {\rm d}(\gb,\gb_\diamond^*)\leq \delta_2}\hb (\gb)^T\nabla_{\gb}\ell<0$$
    for every $\delta_2>\delta_1>0$.}
   However, a close inspection of the proof of Theorem 3.1.1 in \cite{Chen(02)} shows that $(\mathcal B4)$ is employed only to ensure convergence of $g_n$ towards $g^*_\diamond$, which holds in our case.

   \medskip

    $(\mathcal B1)$ is obvious. To prove $(\mathcal B2)$, notice that $(\sum_{k=1}^n\alpha_k^{1-\zeta/2}\epsilonb_k)_{n\geq 1}$ is a martingale. Moreover, there exists a constant $C$ such that, for every $i=1,\dots,{d}$,
    \begin{align*}
        \sup_n\E^*[(\sum_{k=1}^n\alpha_k^{1-\zeta/2}\epsilon_{k,i})^2]
        &= \sup_n\sum_{k=1}^n\alpha_k^{2-\zeta}\E^*[\epsilon_{k,i}^2]\\
        &\leq C \sum_{k=1}^\infty\left(\frac{1}{\alpha+k}\right)^{\gamma(2-\zeta)}<\infty,
        \end{align*}
        since $\gamma(2-\zeta)>1$ and $\E^*[\epsilonb_{n,i}^2]$ are uniformly bounded by \eqref{eq:finitesecder}.
        We now prove $(\mathcal B3)$. The function $\hb $ is continuously differentiable at $\gb_\diamond^*$ and
        $$
        \Hb_{i,j}=-g_{\diamond,i}^*\int_{\R} \frac{k(x\mid\vartheta_i)k(x\mid\vartheta_j)}{f_{\gb_\diamond^*}(x)^2}f_{G^*}(x)\lambda(\ddr x),
        $$
       which is finite by \eqref{eq:finitesecder}. The matrix $\Hb $ can be seen as the product of the diagonal matrix ${\rm diag}(\gb_\diamond^*)$, which is positive definite, times $-\Mb$, with $\Mb_{i,j}=\int_{\R} \frac{k(x\mid \vartheta_i)k(x\mid \vartheta_j)}{f_{\gb_\diamond^*}(x)^2}f_{G^*}(x)\lambda(\ddr x)$. We now prove that $\Mb$ is positive definite. Given a vector $\ab\in \mathbb R^{d}$ with $||\ab||=1$, we can write that
        \begin{align*}
        \ab^T \Mb\ab=&\sum_{i,j=1}^{d}a_ia_j\int_{\R} \frac{k(x\mid\vartheta_i)k(x\mid\vartheta_j)}{f_{\gb_\diamond^*}(x)^2}f_{G^*}(x)\lambda(\ddr x)\\&=
        \int_{\R} \frac{\left(\sum_{i=1}^{d} a_i k(x\mid\vartheta_i)\right)^2}{f_{g^*_\diamond}(x)^2}f_{G^*}(x)\lambda(\ddr x)\geq 0,
        \end{align*}
        with equality if and only if $\sum_{i=1}^{d} a_i k(x\mid\vartheta_i)=0$ for $\lambda$-almost every $x\in \R$. The only solution to the last equation in $\ab$ is the zero vector, which is inconsistent with $||\ab||=1$. Since the function $\ab^T\Mb\ab$ is continuous on the compact set $||\ab||=1$, then the minimum is attained, and by the above reasoning it is strictly positive. It follows that $\Mb$ is positive definite.
\end{proof}

The proof of Theorem \ref{th:regret} is an immediate consequence of lemmas \ref{lem:ineq*} and \ref{lem:stochapp3}.

\subsection{Proof of Proposition \ref{prp_grid}}\label{sec_ex_grid}
   


Existence and uniqueness of the KL minimizer hold for the Poisson and Gaussian kernels since the assumptions $(\mathcal A1)$, $(\mathcal A3)$ and $(\mathcal A4)$ hold (see Lemma 3.1 and Remark 3.2 in \cite{MarTok(09)}).
For both the Poisson and Gaussian kernel and $k\in\N$
$$
\int_\Theta |\theta|^k G^*(\ddr\theta)\leq \int_\Theta \int_{\mathbb R} |x|^kk(x\mid\theta)\lambda(\ddr x)G^*(\ddr\theta).
$$
    Let $J=\{i\in \mathbb \Z :(\vartheta_{i-1},\vartheta_i]\subset \Theta_\diamond\}$, $\underline\vartheta=\inf\{\vartheta_{i-1}:i\in J\}$ and $\overline\vartheta=\sup\{\vartheta_i:i\in J\}$. For every $i\in J$, let $w_i=G^*((\vartheta_{i-1},\vartheta_i])$. Let $G^\dagger$ be the  probability measure
    $$
    G^\dagger=G^*((-\infty,\underline\vartheta])\delta_{\underline\vartheta}+\sum_{i\in J}w_i \delta_{\vartheta_i}+G^*((\overline\vartheta,+\infty))\delta_{\overline\vartheta},
    $$
    (the first term being zero in the Poisson case).
    Then
    $$
    G^\dagger=\delta_{\underline\vartheta}\int_{-\infty}^{\underline\vartheta}G^*(\ddr\theta)+
    \sum_{i\in J}\delta_{\vartheta_i}\int_{I_i}
    G^*(\ddr\theta)+\delta_{\overline\vartheta}\int_{\overline\vartheta}^{+\infty}G^*(\ddr\theta),
    $$
    where $I_i=(\vartheta_{i-1},\vartheta_i]$ for $i\in J$.
  \begin{description}
        \item \textit{Poisson kernel}: $\underline\vartheta=0$, and denoting by $M=\overline\vartheta=d \eta/2$, we can write, by the convexity of the Kullback-Leibler divergence,
        \begin{align*}
        &\mathrm{KL}(f_{G^*}||f_{G^\dagger})\\&
        \leq \sum_{i=1}^{d}\int_{I_i}\sum_{y\in\mathbb N_0}\log\left(\frac{e^{-\theta}\theta^y}{e^{-\vartheta_i}\vartheta_i^y}\right)\frac{e^{-\theta}\theta^y}{y!}G^*(\ddr\theta)
        +\int_M^\infty \sum_{y\in\mathbb N_0}\log\left(\frac{e^{-\theta}\theta^y}{e^{-M}M^y}\right)\frac{e^{-\theta}\theta^y}{y!}G^*(\ddr\theta)\\
       &\leq \int_0^M\frac{\eta}{2} G^*(\ddr\theta)+\int_M^\infty \sum_{y\in\mathbb N_0}y\log \theta\frac{e^{-\theta}\theta^y}{y!}G^*(\ddr\theta)\\
       &\leq \eta/2 +\int_M^\infty \theta\log \theta \,G^*(\ddr\theta)\\
       &\leq \eta/2+\frac{\log M}{M^{k-1}}\int_M^\infty \theta^k\,G^*(\ddr\theta)\\
        &\leq \eta/2+\frac{m_k\log M}{M^{k-1}}\\
       &\leq \eta,
       \end{align*}
      where  we have used the fact that, for $\theta \in I_i := (\vartheta_{i-1},\vartheta_i]$,
\begin{align*}
\mathrm{KL}\big(\mathrm{Pois}(\theta)\,\|\,\mathrm{Pois}(\vartheta_i)\big)
&= \vartheta_i - \theta + \theta \log\!\left(\frac{\theta}{\vartheta_i}\right) \\
&\le \vartheta_i - \theta + \theta \frac{\theta-\vartheta_i}{\vartheta_i}
= \frac{(\vartheta_i-\theta)^2}{\vartheta_i}
\le \frac{(\eta/2)^2}{\eta/2}
= \frac{\eta}{2}.
\end{align*}
        \item \textit{Gaussian kernel}: denoting by $-M=\underline\vartheta=-d\sigma\sqrt{\eta}$ and by $M=\overline\vartheta=d\sigma\sqrt{\eta}$, we can write by the convexity of the Kullback-Leibler divergence, that
    \begin{align*}
        &\mathrm{KL}(f_{G^*}||f_{G^\dagger})\\
        &
        \leq \int_{-\infty}^{-M}\mathrm{KL}(k(\cdot \mid\theta)||k(\cdot\mid\! -\!M))G^*(\ddr \theta)        +\sum_{\vartheta_i\in\Theta_\diamond}\int_{I_i}\mathrm{KL}(k(\cdot\mid\theta)||k(\cdot\mid \vartheta_i))G^*(\ddr\theta)\\
        &+\int_{M}^{+\infty}\mathrm{KL}(k(\cdot\mid\theta)||k(\cdot\mid M))G^*(\ddr \theta)\\
             &\leq \int_{-\infty}^{-M}  \frac{(\theta+M)^2}{2\sigma^2}G^*(\ddr \theta)
        +\sum_{\vartheta_i\in\Theta_\diamond}\int_{I_i}
         \frac{(\theta-\vartheta_i)^2}{2\sigma^2}G^*(\ddr\theta)
        +\int_{M}^{+\infty}   \frac{(\theta-M)^2}{2\sigma^2}G^*(\ddr \theta)\\
        &\leq \int_{-\infty}^{-M}  \frac{\theta^2}{2\sigma^2}G^*(\ddr \theta)
        +\sum_{\vartheta_i\in\Theta_\diamond}\int_{I_i}
         \frac{\eta}{2}G^*(\ddr\theta)
        +\int_{M}^{+\infty}   \frac{\theta^2}{2\sigma^2}G^*(\ddr \theta)\\
         &\leq\frac{\eta}{2}+\frac{\int_{|\theta|\geq M}   |\theta|^kG^*(\ddr \theta)}{2\sigma^2M^{k-2}}\\
        &\leq \frac{\eta}{2}+\frac{m_k}{2\sigma^2M^{k-2}}\\
       &\leq \eta,
    \end{align*}
    where we have used that, for any $\theta \in I_i := (\vartheta_{i-1},\vartheta_i]$,
\begin{align*}
\mathrm{KL}\big(\mathcal N(\theta,\sigma^2)\,\|\,\mathcal N(\vartheta_i,\sigma^2)\big)
&= \frac{(\theta-\vartheta_i)^2}{2\sigma^2}
\le \frac{(\sigma\sqrt{\eta})^2}{2\sigma^2}
= \frac{\eta}{2}.
\end{align*}

    \end{description}
%
%
To establish the rate of convergence, we verify that condition \eqref{eq:finitesecder} is satisfied.
For the Poisson kernel, with $\lambda$ denoting the counting measure on $\mathbb N_0$, condition \eqref{eq:momentgf} implies that
\begin{align*}
    &\sup_{\vartheta_i,\vartheta_j\in\Theta_\diamond} \int_{\R}\frac{k(x\mid\vartheta_i)^2}{k(x\mid\vartheta_j)^2}f_{G^*}(x)\lambda(dx)\leq \exp((d-1)\eta)\int_{\R} d^{2x}f_{G^*}(x)\lambda(dx)<+\infty.
    \end{align*}
For the Gaussian kernel, with $\lambda$ denoting the Lebesgue measure on $\R$,  \eqref{eq:momentgf} implies that
\begin{align*}
    &\sup_{\vartheta_i,\vartheta_j\in\Theta_\diamond} \int_{\R}\frac{k(x\mid\vartheta_i)^2}{k(x\mid\vartheta_j)^2}f_{G^*}(x)\lambda(dx)\leq \exp(4d^2\eta)
    \int_{\R}\exp(4\frac{d\sqrt{\eta} |x|}{\sigma})f_{G^*}(x)\lambda(dx)<+\infty.
\end{align*}

\subsection{Grid construction for general kernels}\label{thm_grid}
The next proposition establishes, under suitable assumptions on the kernel $k$, a constructive discretization of the oracle prior $G^*$ achieving a prescribed Kullback–Leibler divergence.

\begin{prp}
    Let $G^{\ast}$ be the oracle prior on $\mathbb{R}$, and let $\eta>0$. Denote by $\underline \theta=\sup\{\theta:G^*(-\infty,\theta)=0\}$ and $\overline \theta=\inf\{\theta:G^*(\theta,+\infty)=0\}$. Assume that there exists $M>0$ and $\psi >0$ such that the following conditions hold:
\begin{equation}\label{eq:defM}
    \int_{(-\infty,M]}\!\!\!\!\!\!\!\!\!\!\!\!\!\mathrm{KL}(k(\cdot\mid \theta)\,||\,k(\cdot\mid - M))G^*(\ddr \theta)+ \int_{[M,+\infty)}\!\!\!\!\!\!\!\!\!\!\!\!\!\mathrm{KL}(k(\cdot\mid \theta)\,||\,k(\cdot\mid  M))G^*(\ddr \theta)<\frac{\eta}{2},
\end{equation}
\begin{equation}\label{eq:defeta}
  \sup_{|\theta|,|\theta'|\leq M, |\theta-\theta'|<\psi}\mathrm{KL}(k(\cdot\mid \theta)\,||\,k(\cdot\mid \theta'))<\frac \eta 2.
\end{equation}
Let $\Theta_\diamond=\{\vartheta_0,\dots,\vartheta_{d}\}$ (with $\vartheta_{i-1}<\vartheta_i $ for every $i=1,\dots,d$) be a grid of equally spaced points in $\Theta$ such that $\inf_{\Theta_\diamond}\theta=\max(-M,\underline\theta)$, $\sup_{\Theta_\diamond}\theta=\min (M,\overline\theta)$,   and $d$ is the smallest integer such that $\vartheta_i-\vartheta_{i-1}\leq \psi$ for every $i=1,\dots,d$. Then, $\min_{G \in \mathcal{S}_\diamond} \mathrm{KL}(f_{G^{\ast}} \,\|\, f_{G})\leq \eta $, where $ \mathcal{S}_\diamond $  is the set of all probability distributions supported on $ \Theta_\diamond $.
 \end{prp}
    
\begin{proof}
    For every $i=1,\dots,d$, let $w_i=G^*((\vartheta_{i-1},\vartheta_i])$ and let $G^\dagger$ be the  probability measure
    $$
    G^\dagger=G^*((-\infty,\vartheta_0])\delta_{\vartheta_{0}}+\sum_{i=1}^{d} w_i \delta_{\vartheta_i}+G^*((\vartheta_{{d}},+\infty))\delta_{\vartheta_{{d}}}.
    $$
        Then
    $$
    G^\dagger=\delta_{\vartheta_{0}}\int_{-\infty}^{\vartheta_0}G^*(\ddr\theta)+
    \sum_{i=1}^{d}\delta_{\vartheta_i}\int_{I_i}
    G^*(\ddr\theta)+\delta_{\vartheta_{{d}}}\int_{\vartheta_{d}}^{+\infty}G^*(\ddr\theta),
    $$
    where $I_i=(\vartheta_{i-1},\vartheta_i]$ for $i=1,\dots,{d}$. Since the KL divergence is convex, then
    \begin{align*}
        \mathrm{KL}(f_{G^*}||f_{G^\dagger})&
              \leq\int_{-\infty}^{\vartheta_0} \mathrm{KL}(k(\cdot\mid\theta)||k(\cdot\mid \vartheta_0))G^*(\ddr\theta)+\int_{\vartheta_{d}}^{+\infty} \mathrm{KL}(k(\cdot\mid\theta)||k(\cdot\mid \vartheta_{d}))G^*(\ddr\theta)\\&+
    \sum_{i=1}^{d}\int_{I_i}  \mathrm{KL}(k(\cdot\mid\theta)||k(\cdot\mid\vartheta_i))G^*(\ddr\theta)\leq \eta.
    \end{align*}
\end{proof}

\section{Synthetic-data illustrations: Poisson kernel}\label{app_num}

We present a detailed and exhaustive synthetic-data analysis for the estimation of $S_{r,n}$, $r=1,2,3$, under the Poisson mixture model. Beyond estimation, we also consider quasi-Bayes credible intervals, which, for a fixed leve $1-\alpha$, are constructed according to \eqref{cred_interval}. The quasi-Bayes credible interval is then compared with the corresponding oracle credible interval of level $1-\alpha$, which is the $(1-\alpha)$-level posterior quantile obtained under the ``true'', or oracle, prior $G^\ast$ that generate the $\theta_{i}$'s in the Poisson and Gaussian mixture models. For instance, consider the Poisson mixture model \eqref{model_distr_poi}. To define the oracle credible intervals, let
\begin{displaymath}
G^\ast(d\theta_i\mid X_i)=
\frac{\text{Poisson}(X_i\mid \theta_i)\,G^\ast(\ddr\theta_i)}
{\int_\Theta \text{Poisson}(X_i\mid \vartheta)\,G^\ast(\ddr\vartheta)},
\qquad i=1,\ldots,n,
\end{displaymath}
denote the posterior distribution of $\theta_i$ given $X_i$ under the prior $G^\ast$. Because of the definition of the model \eqref{model_distr_poi}, conditionally on $(X_1,\ldots,X_n)$, the posterior distribution of $(\theta_1,\ldots,\theta_n)$ factorizes as $\prod_{1\leq i\leq n} G^\ast(d\theta_i\mid X_i)$, the oracle credible interval for $S_{1,n}$ is defined as the equal-tail posterior credible interval of level $1-\alpha$ associated with the posterior distribution function
\begin{displaymath}
F^\ast_{S_{1,n}}(t\mid X_1,\ldots,X_n)=\PP_{G^\ast}(S_{1,n}\leq t\mid X_1,\ldots,X_n),\qquad t\in\mathbb R,
\end{displaymath}
that is $I^\ast_{S_{1,n}}=[L^\ast_{S_{1,n}},U^\ast_{S_{1,n}}]$, where
\begin{displaymath}
L^\ast_{T_{\kappa,n}}=\inf\{t\in\mathbb R:F^\ast_{S_{1,n}}(t\mid X_1,\ldots,X_n)\geq \alpha/2\},
\end{displaymath}
and
\begin{displaymath}
U^\ast_{T_{\kappa,n}}=\inf\{t\in\mathbb R:F^\ast_{S_{1,n}}(t\mid X_1,\ldots,X_n)\geq 1-\alpha/2\}.
\end{displaymath}
Similarly, we define the oracle credible interval for $S_{3,n}$. In particular, this is defined as the equal-tail posterior credible interval of level $1-\alpha$ associated with the posterior distribution function
\begin{displaymath}
F^\ast_{S_{3,n}}(t\mid X_1,\ldots,X_n)=\PP_{G^\ast}(S_{3,n}\leq t\mid X_1,\ldots,X_n),\qquad t\in\mathbb R,
\end{displaymath}
that is $I^\ast_{S_{3,n}}=[L^\ast_{S_{3,n}},U^\ast_{S_{3,n}}]$, where
\begin{displaymath}
L^\ast_{S_{3,n}}=\inf\{t\in\mathbb R:F^\ast_{S_{3,n}}(t\mid X_1,\ldots,X_n)\geq \alpha/2\}
\end{displaymath}
and
\begin{displaymath}
U^\ast_{S_{3,n}}=\inf\{t\in\mathbb R:F^\ast_{S_{3,n}}(t\mid X_1,\ldots,X_n)\geq 1-\alpha/2\}.
\end{displaymath}
In the numerical experiments, the oracle posterior distribution functions $F^\ast_{S_{1,n}}$ and $F^\ast_{S_{3,n}}$ are evaluated numerically, and the corresponding equal-tail quantiles are used to construct the oracle credible intervals.

\subsection{Classes of prior distributions}\label{app_num1}

Under the Poisson mixture model \eqref{model_distr_poi}, we consider various choices of $G$, First, we assume $G$ to be a Uniform distribution on $[a,b]$; precisely, we set $a=0$ and $b=10$. Then, we consider two examples of $G$ belonging to the class $\mathcal{G}$ of sub-exponential distribution of order $s$, which is defined as
\begin{displaymath}
\mathcal{G}=\left\{G\text{ on }\mathbb{R}^{+}:\;G([t,\infty)) \leq 2e^{-t/s}\ \text{for all } t>0\right\},\qquad s>0.
\end{displaymath}
In particular, we assume $G$ to be: i) a Weibull distribution with scale $a$ and shape $b$, which belongs to $\mathcal{G}$ for $b\geq1$; ii) a half-Gaussian distribution, namely the distribution of the positive part of a Gaussian random variable with mean $0$ and variance $\sigma^{2}$, which belongs $\mathcal{G}$ for $\sigma>0$. Precisely we set $a=5$ and $b=3$ for the Weibull, and $\sigma=1$ for the half-Gaussian; under this parameterization the Weibull tail is lighter than the half-Gaussian tail. Finally, we consider an example of $G$ belonging to the moment class $\mathcal{M}$ defined, for any real $M_{p}$, as
\begin{displaymath}
\mathcal{M}=\left\{G\text{ on }\mathbb{R}^{+}:\;\int_{\mathbb{R}^{+}}\theta^{p} G(\ddr\theta)<M_{p}\right\},\qquad p>0.
\end{displaymath}
In particular, we assume $G$ to be square-root of half-Cauchy distribution, namely the distribution of the square-root of the positive part of a standard Cauchy random variable. This distribution has heavier tail than both the Weibull distribution and the half-Gaussian distribution.

\subsubsection{Uniform prior}

For $i=1,\ldots,100$, let $\mathbf{X}_{i}=X_{1:100 i}$ denote a dataset of size $n=100i$ generated from the Poisson mixture model \eqref{model_distr_poi}, with Uniform prior  $G$ on the set $[0,10]$. The $\mathbf{X}_{i}$'s are nested, so that the sample size increases progressively: at stage $i$, we have $n=100i$ data, obtained by adding $100$ new data at each step, starting from $n=100$. For each dataset $\mathbf{X}_{i}$ of size $n=100i$, we apply the quasi-Bayes EB approach to estimate $S_{1,n}$ and $S_{2,n}$, with $\kappa=2$, and $S_{3,n}$.

The integral in Newton's algorithm \eqref{eq:newton} is evaluated numerically via the trapezoidal rule. To perform this evaluation on a dataset $X_{1:n}$ of size $n$, the density function of $G_{n}$ is represented through its values on a fixed grid of $d\in\{5,000;\,1,000;\,500;\,100;\,50\}$ quadrature points over $\Theta=(0,U_{\Theta})$, where $U_\Theta=\max\{\max\{X_{1:n}\},\lceil Q_{n,0.99}+4\sqrt{\max\{Q_{n,0.99},1\}}\rceil\}$, with $Q_{n,0.99}=\text{Quantile}(X_{1:n};0.99)$. Further, we set $G_{0}$ to be Uniform over $\Theta$, and set $\alpha_{n}=(1+n)^{-0.99}$.  Under this setting for Newton's algorithm, we obtain an estimate $G_n$ of the mixing distribution $G$, which, when substituted into \eqref{eq:est}, gives the quasi-Bayes EB estimates $\hat{S}^{\text{\tiny{[Q-B]}}}_{r,n}$ of $S_{r,n}$, for $r=1,2,3$. With regards to the estimation of $S_{1,n}$, Table~\ref{tab_unif_s1_sens_supp} provides the mean absolute deviations (MAD) and CPU times as $n$ and $d$ vary. The CPU time refers to the time (in seconds) for processing a new observation on a laptop MacBook Pro (M1 type processor).

\begin{table}[ht]
\centering
\caption{Uniform prior: MAD and CPU time (in seconds) of $\hat{S}^{\text{\tiny{[Q-B]}}}_{1,n}$ as $n$ and $d$ vary}
{
\setlength{\tabcolsep}{0pt}
\begin{tabular}{@{}l@{\hspace{1.2cm}}*{5}{>{\centering\arraybackslash}p{2.15cm}}@{}}
\hline
\hline
 & $d=5{,}000$ & $d=1{,}000$ & $d=500$ & $d=100$ & $d=50$ \\[0.1cm]
\hline
\multicolumn{6}{@{}l}{\underline{$n=1{,}000$}} \\[0.05cm]
MAD      & 0.0918 & 0.0918 & 0.0918 & 0.0917 & 0.0916 \\
CPU time & 0.0023 & 0.0009 & 0.0007 & 0.0004 & 0.0004 \\[0.2cm]

\multicolumn{6}{@{}l}{\underline{$n=2{,}000$}} \\[0.05cm]
MAD      & 0.0852 & 0.0852 & 0.0852 & 0.0852 & 0.0851 \\
CPU time & 0.0027 & 0.0009 & 0.0007 & 0.0004 & 0.0004 \\[0.2cm]

\multicolumn{6}{@{}l}{\underline{$n=3{,}000$}} \\[0.05cm]
MAD      & 0.0822 & 0.0822 & 0.0822 & 0.0821 & 0.0820 \\
CPU time & 0.0027 & 0.0009 & 0.0007 & 0.0004 & 0.0004 \\[0.2cm]

\multicolumn{6}{@{}l}{\underline{$n=4{,}000$}} \\[0.05cm]
MAD      & 0.0738 & 0.0738 & 0.0738 & 0.0738 & 0.0737 \\
CPU time & 0.0023 & 0.0009 & 0.0007 & 0.0004 & 0.0004 \\[0.2cm]

\multicolumn{6}{@{}l}{\underline{$n=5{,}000$}} \\[0.05cm]
MAD      & 0.0694 & 0.0694 & 0.0694 & 0.0694 & 0.0693 \\
CPU time & 0.0023 & 0.0009 & 0.0007 & 0.0004 & 0.0004 \\[0.2cm]

\multicolumn{6}{@{}l}{\underline{$n=6{,}000$}} \\[0.05cm]
MAD      & 0.0655 & 0.0655 & 0.0655 & 0.0655 & 0.0654 \\
CPU time & 0.0023 & 0.0009 & 0.0007 & 0.0004 & 0.0004 \\[0.2cm]

\multicolumn{6}{@{}l}{\underline{$n=7{,}000$}} \\[0.05cm]
MAD      & 0.0609 & 0.0609 & 0.0609 & 0.0609 & 0.0607 \\
CPU time & 0.0023 & 0.0009 & 0.0007 & 0.0004 & 0.0004 \\[0.2cm]

\multicolumn{6}{@{}l}{\underline{$n=8{,}000$}} \\[0.05cm]
MAD      & 0.0595 & 0.0595 & 0.0595 & 0.0595 & 0.0594 \\
CPU time & 0.0024 & 0.0009 & 0.0007 & 0.0004 & 0.0004 \\[0.2cm]

\multicolumn{6}{@{}l}{\underline{$n=9{,}000$}} \\[0.05cm]
MAD      & 0.0545 & 0.0545 & 0.0545 & 0.0544 & 0.0543 \\
CPU time & 0.0023 & 0.0009 & 0.0007 & 0.0004 & 0.0004 \\[0.2cm]

\multicolumn{6}{@{}l}{\underline{$n=10{,}000$}} \\[0.05cm]
MAD      & 0.0552 & 0.0552 & 0.0552 & 0.0552 & 0.0551 \\
CPU time & 0.0023 & 0.0009 & 0.0007 & 0.0004 & 0.0004 \\[0.1cm]
\hline
\hline
\end{tabular}
}
\label{tab_unif_s1_sens_supp}
\end{table}

Figure~\ref{fig_unif_s1_supp} and Table~\ref{tab_unif_s1_supp} compares the MADs for $\hat{S}^{\text{\tiny{[O]}}}_{1,n}$, $\hat{S}^{\text{\tiny{[ML]}}}_{1,n}$, $\hat{S}^{\text{\tiny{[B]}}}_{1,n}$, $\hat{S}^{\text{\tiny[``u,v"]}}_{1,n}$ and $\hat{S}^{\text{\tiny{[Q-B]}}}_{1,n}$. Here, $\hat{S}^{\text{\tiny{[O]}}}_{1,n}=\E_{G}[S_{1,n}\,|\,X_{1:n}]$ with $G$ being the Uniform distribution on the set $[0,10]$. In particular, for $\hat{S}^{\text{\tiny{[Q-B]}}}_{1,n}$ we consider: i) the fixed uniform grid of $d=1,000$ quadrature points over $\Theta=(0,U_{\Theta})$; ii) $G_{0}$ to be Uniform over $\Theta$; iii) the learning rate $\alpha_{n}=(1+n)^{-0.99}$. This is precisely the initialization considered in the second column of Table~\ref{tab_unif_s1_sens_supp}, which also provides CPU time.

\begin{figure}[h!]
\begin{center}
\includegraphics[width=1\linewidth,height=0.36\textheight,keepaspectratio]{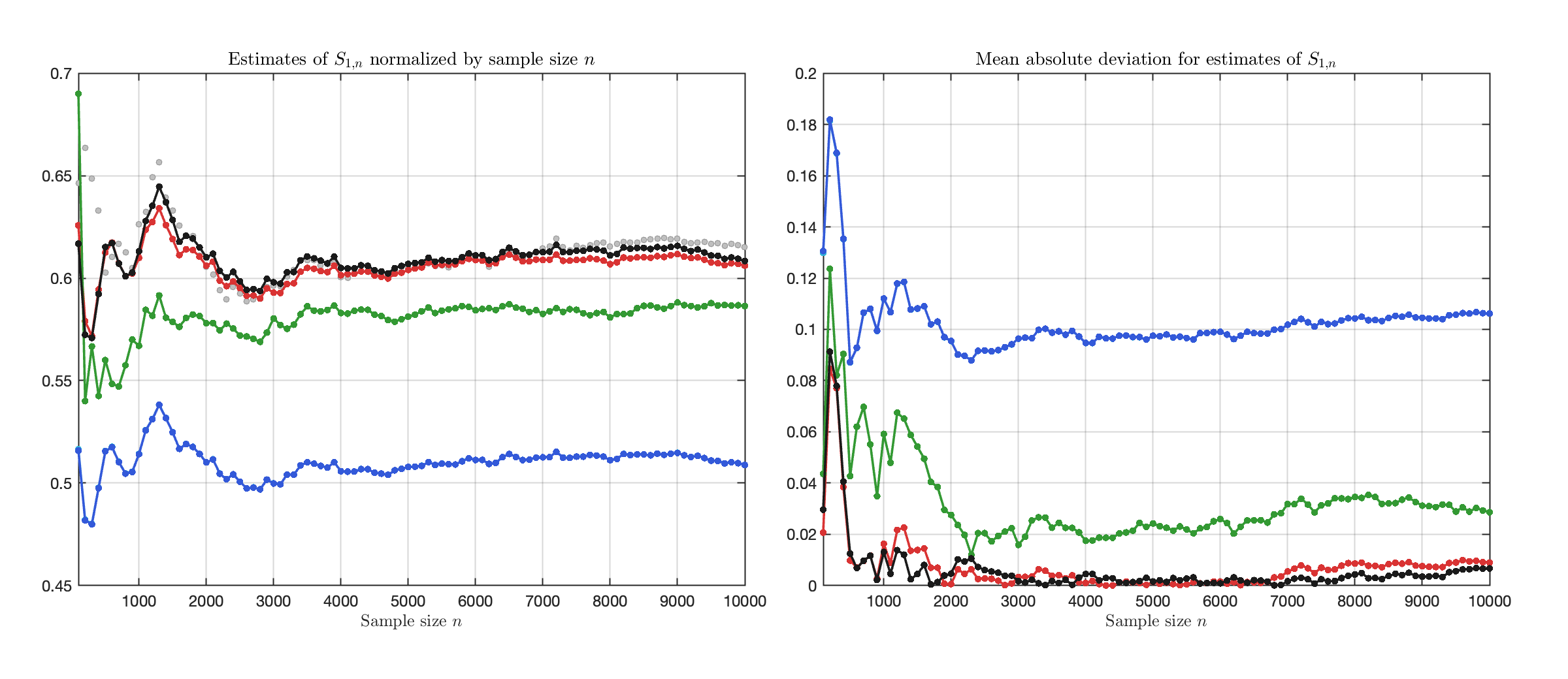}
\end{center}
\caption{\scriptsize{Uniform prior, $S_{1,n}$. Left panel: true values $n^{-1}S_{1,n}$ (Grey o-) and estimates $n^{-1}\hat{S}^{\text{\tiny{[O]}}}_{1,n}$ (Black .-), $n^{-1}\hat{S}^{\text{\tiny{[ML]}}}_{1,n}$ (Blue .-), $n^{-1}\hat{S}^{\text{\tiny{[B]}}}_{1,n}$ (Cyan .-), $n^{-1}\hat{S}^{\text{\tiny[``u,v"]}}_{1,n}$ (Green .-) and $n^{-1}\hat{S}^{\text{\tiny{[Q-B]}}}_{1,n}$ (Red .-). Right panel: MAD of $\hat{S}^{\text{\tiny{[O]}}}_{1,n}$ (Black .-), $\hat{S}^{\text{\tiny{[ML]}}}_{1,n}$ (Blue .-), $\hat{S}^{\text{\tiny{[B]}}}_{1,n}$ (Cyan .-), $\hat{S}^{\text{\tiny[``u,v"]}}_{1,n}$ (Green .-) and $\hat{S}^{\text{\tiny{[Q-B]}}}_{1,n}$ (Red .-)}}
\label{fig_unif_s1_supp}
\end{figure}

\begin{table}[ht]
\centering
\caption{Uniform prior, $S_{1,n}$: MAD as $n$ varies}
{
\setlength{\tabcolsep}{0pt}
\begin{tabular}{@{}l@{\hspace{1.2cm}}*{5}{>{\centering\arraybackslash}p{2.15cm}}@{}}
\hline
\hline
& $\hat{S}^{\text{\tiny{[O]}}}_{1,n}$ & $\hat{S}^{\text{\tiny{[ML]}}}_{1,n}$ & $\hat{S}^{\text{\tiny{[B]}}}_{1,n}$ & $\hat{S}^{\text{\tiny{[``u,v'']}}}_{1,n}$ & $\hat{S}^{\text{\tiny{[Q-B]}}}_{1,n}$ \\[0.1cm]
\hline
\multicolumn{6}{@{}l}{\underline{$n=1{,}000$}} \\[0.05cm]
MAD & 0.0248 & 0.1900 & 0.1901 & 0.0920 & \textbf{0.0918} \\[0.2cm]

\multicolumn{6}{@{}l}{\underline{$n=2{,}000$}} \\[0.05cm]
MAD & 0.0113 & 0.1803 & 0.1803 & 0.0848 & \textbf{0.0852} \\[0.2cm]

\multicolumn{6}{@{}l}{\underline{$n=3{,}000$}} \\[0.05cm]
MAD & 0.0075 & 0.1770 & 0.1770 & 0.0746 & \textbf{0.0822} \\[0.2cm]

\multicolumn{6}{@{}l}{\underline{$n=4{,}000$}} \\[0.05cm]
MAD & 0.0079 & 0.1659 & 0.1659 & 0.0458 & \textbf{0.0738} \\[0.2cm]

\multicolumn{6}{@{}l}{\underline{$n=5{,}000$}} \\[0.05cm]
MAD & 0.0081 & 0.1616 & 0.1616 & 0.0344 & \textbf{0.0694} \\[0.2cm]

\multicolumn{6}{@{}l}{\underline{$n=6{,}000$}} \\[0.05cm]
MAD & 0.0040 & 0.1594 & 0.1594 & 0.0193 & \textbf{0.0655} \\[0.2cm]

\multicolumn{6}{@{}l}{\underline{$n=7{,}000$}} \\[0.05cm]
MAD & 0.0022 & 0.1545 & 0.1546 & 0.0096 & \textbf{0.0609} \\[0.2cm]

\multicolumn{6}{@{}l}{\underline{$n=8{,}000$}} \\[0.05cm]
MAD & 0.0025 & 0.1536 & 0.1537 & 0.0040 & \textbf{0.0595} \\[0.2cm]

\multicolumn{6}{@{}l}{\underline{$n=9{,}000$}} \\[0.05cm]
MAD & 0.0005 & 0.1483 & 0.1483 & 0.0043 & \textbf{0.0545} \\[0.2cm]

\multicolumn{6}{@{}l}{\underline{$n=10{,}000$}} \\[0.05cm]
MAD & 0.0019 & 0.1495 & 0.1495 & 0.0072 & \textbf{0.0552} \\[0.1cm]
\hline
\hline
\end{tabular}
}
\label{tab_unif_s1_supp}
\end{table}

Figure~\ref{fig_unif_s3_supp} and Table~\ref{tab_unif_s3_supp} report corresponding results with respect to the estimation of $S_{3,n}$.

\begin{figure}[h!]
\begin{center}
\includegraphics[width=1\linewidth,height=0.36\textheight,keepaspectratio]{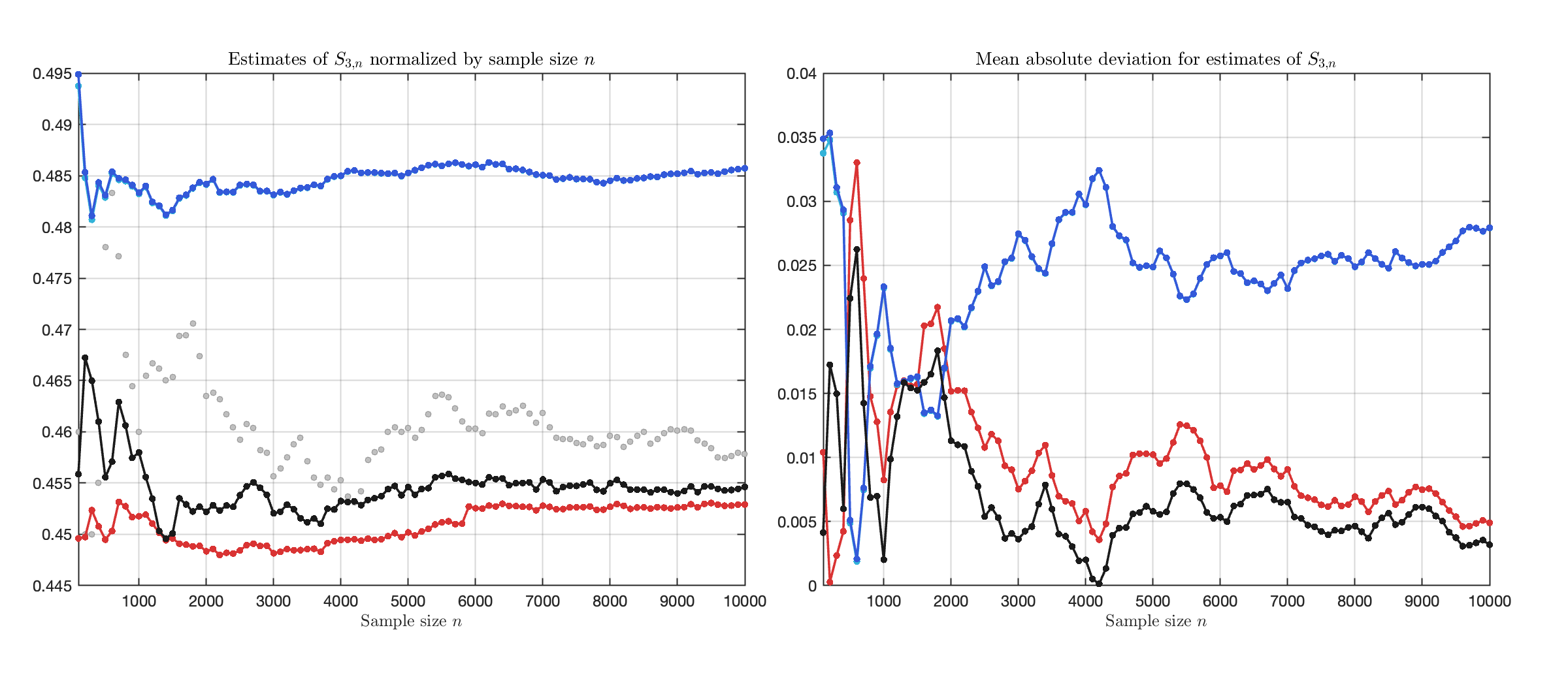}
\end{center}
\caption{\scriptsize{Uniform prior, $S_{3,n}$. Left panel: true values of $n^{-1}S_{3,n}$ (Grey o-) and estimates $n^{-1}\hat{S}^{\text{\tiny{[O]}}}_{3,n}$ (Black .-), $n^{-1}\hat{S}^{\text{\tiny{[ML]}}}_{3,n}$ (Blue .-), $n^{-1}\hat{S}^{\text{\tiny{[B]}}}_{3,n}$ (Cyan .-)  and $n^{-1}\hat{S}^{\text{\tiny{[Q-B]}}}_{3,n}$ (Red .-). Right panel: MAD of $\hat{S}^{\text{\tiny{[O]}}}_{3,n}$ (Black .-), $\hat{S}^{\text{\tiny{[ML]}}}_{3,n}$ (Blue .-), $\hat{S}^{\text{\tiny{[B]}}}_{3,n}$ (Cyan .-) and $\hat{S}^{\text{\tiny{[Q-B]}}}_{3,n}$ (Red .-)}}
\label{fig_unif_s3_supp}
\end{figure}

\begin{table}[ht]
\centering
\caption{Uniform prior, $S_{3,n}$: MAD as $n$ varies}
{
\setlength{\tabcolsep}{0pt}
\begin{tabular}{@{}l@{\hspace{1.2cm}}*{4}{>{\centering\arraybackslash}p{2.15cm}}@{}}
\hline
\hline
& $\hat{S}^{\text{\tiny{[O]}}}_{3,n}$ & $\hat{S}^{\text{\tiny{[ML]}}}_{3,n}$ & $\hat{S}^{\text{\tiny{[B]}}}_{3,n}$ & $\hat{S}^{\text{\tiny{[Q-B]}}}_{3,n}$ \\[0.1cm]
\hline
\multicolumn{5}{@{}l}{\underline{$n=1{,}000$}} \\[0.05cm]
MAD & 0.0016 & 0.0201 & 0.0202 & \textbf{0.0084} \\[0.2cm]

\multicolumn{5}{@{}l}{\underline{$n=2{,}000$}} \\[0.05cm]
MAD & 0.0048 & 0.0171 & 0.0171 & \textbf{0.0083} \\[0.2cm]

\multicolumn{5}{@{}l}{\underline{$n=3{,}000$}} \\[0.05cm]
MAD & 0.0003 & 0.0217 & 0.0218 & \textbf{0.0142} \\[0.2cm]

\multicolumn{5}{@{}l}{\underline{$n=4{,}000$}} \\[0.05cm]
MAD & 0.0015 & 0.0142 & 0.0142 & \textbf{0.0090} \\[0.2cm]

\multicolumn{5}{@{}l}{\underline{$n=5{,}000$}} \\[0.05cm]
MAD & 0.0003 & 0.0134 & 0.0134 & \textbf{0.0097} \\[0.2cm]

\multicolumn{5}{@{}l}{\underline{$n=6{,}000$}} \\[0.05cm]
MAD & 0.0006 & 0.0138 & 0.0138 & \textbf{0.0095} \\[0.2cm]

\multicolumn{5}{@{}l}{\underline{$n=7{,}000$}} \\[0.05cm]
MAD & 0.0014 & 0.0111 & 0.0111 & \textbf{0.0075} \\[0.2cm]

\multicolumn{5}{@{}l}{\underline{$n=8{,}000$}} \\[0.05cm]
MAD & 0.0011 & 0.0113 & 0.0113 & \textbf{0.0077} \\[0.2cm]

\multicolumn{5}{@{}l}{\underline{$n=9{,}000$}} \\[0.05cm]
MAD & 0.0014 & 0.0095 & 0.0095 & \textbf{0.0065} \\[0.2cm]

\multicolumn{5}{@{}l}{\underline{$n=10{,}000$}} \\[0.05cm]
MAD & 0.0009 & 0.0090 & 0.0090 & \textbf{0.0063} \\[0.1cm]
\hline
\hline
\end{tabular}
}
\label{tab_unif_s3_supp}
\end{table}

To conclude, we provide quasi-Bayes credible intervals at level $1-\alpha=0.95$. Specifically, credible intervals are constructed from \eqref{cred_interval} by relying on Newton's algorithm initialized as in Table~\ref{tab_unif_s1_supp}. Figure~\ref{uniform_fig_interval_s1}-\ref{uniform_fig_interval_s3} display the quasi-Bayes credible intervals for $S_{1,n}$ and $S_{3,n}$, respectively, with the corresponding oracle credible intervals $I^\ast_{S_{1,n}}$ and $I^\ast_{S_{3,n}}$ under the Uniform prior.

\begin{figure}[h!]
\begin{center}
\includegraphics[width=1\linewidth,height=0.36\textheight,keepaspectratio]{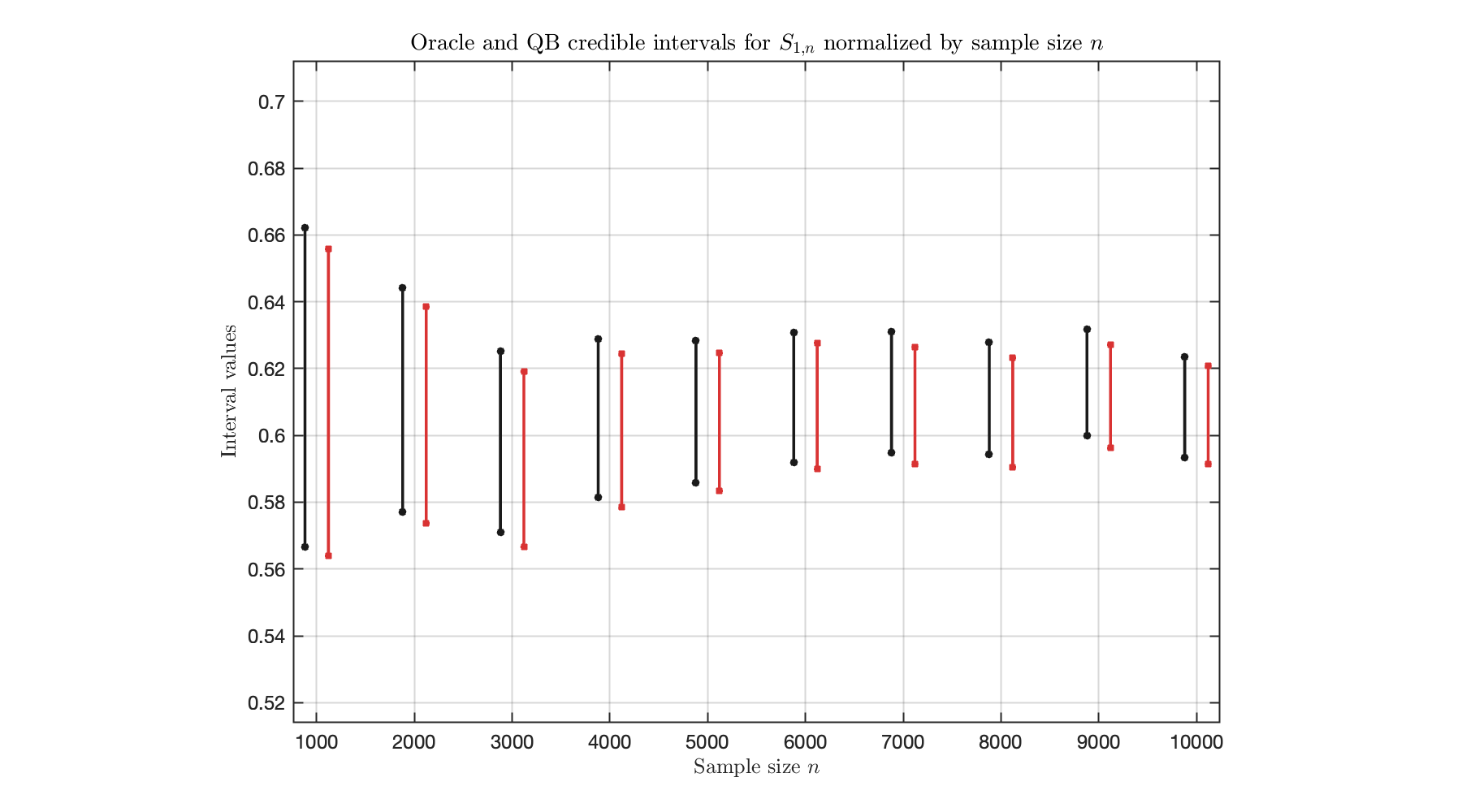}
\end{center}
\caption{\scriptsize{Uniform prior, $S_{1,n}$: oracle credible intervals (black) and quasi-Bayes credible intervals (red)}}
\label{uniform_fig_interval_s1}
\end{figure}

\begin{figure}[h!]
\begin{center}
\includegraphics[width=1\linewidth,height=0.36\textheight,keepaspectratio]{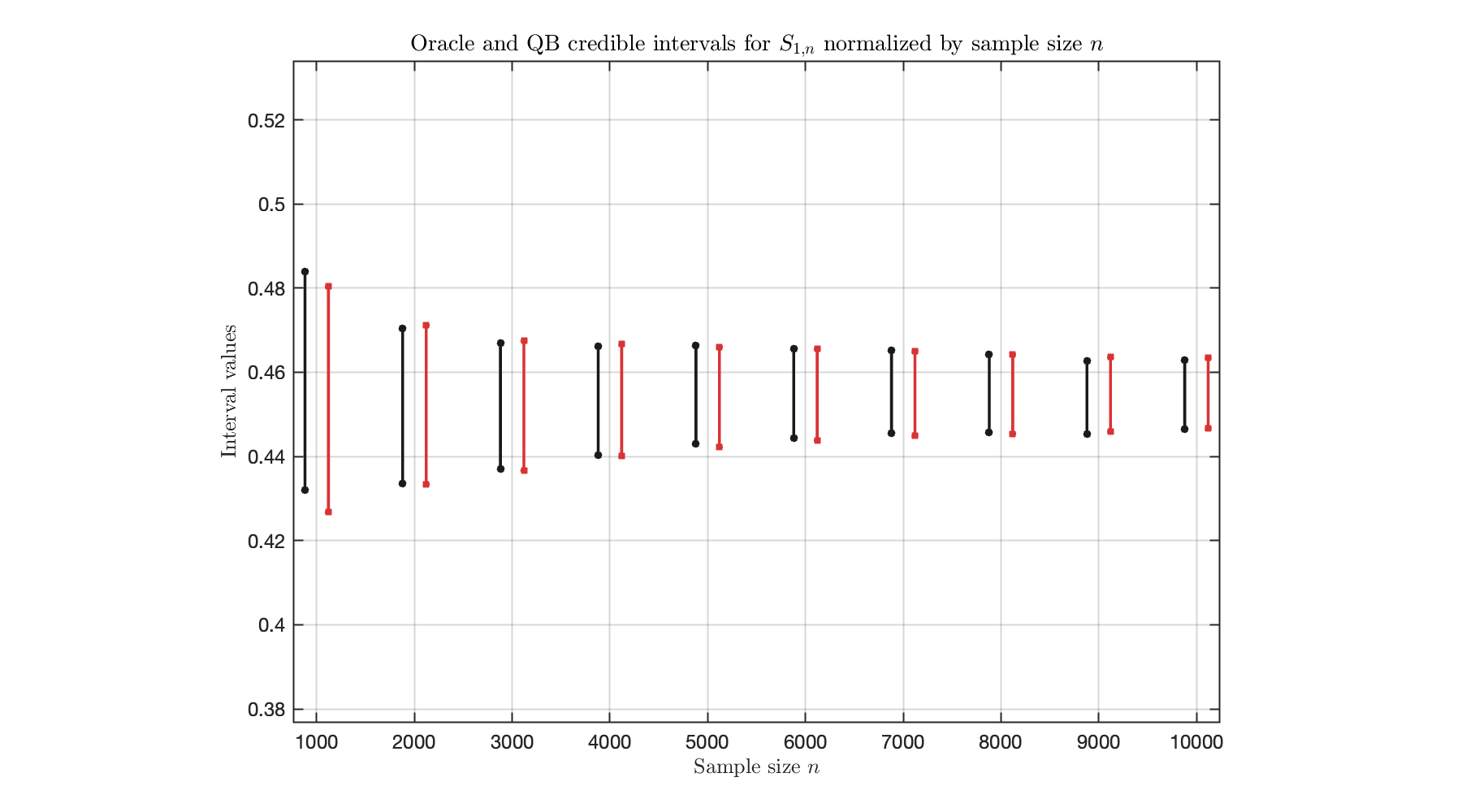}
\end{center}
\caption{\scriptsize{Uniform prior, $S_{3,n}$: oracle credible intervals (black) and quasi-Bayes credible intervals (red)}}
\label{uniform_fig_interval_s3}
\end{figure}

\subsubsection{Weibull prior}

For $i=1,\ldots,100$, let $\mathbf{X}_{i}=X_{1:100 i}$ denote a dataset of size $n=100i$ generated from the Poisson mixture model \eqref{model_distr_poi}, with Weibull prior  $G$ of scale $5$ and shape $3$. The $\mathbf{X}_{i}$'s are nested, so that the sample size increases progressively: at stage $i$, we have $n=100i$ data, obtained by adding $100$ new data at each step, starting from $n=100$. For each dataset $\mathbf{X}_{i}$ of size $n=100i$, we apply the quasi-Bayes EB approach to estimate $S_{1,n}$ and $S_{2,n}$, with $\kappa=2$, and $S_{3,n}$.

The integral in Newton's algorithm \eqref{eq:newton} is evaluated numerically via the trapezoidal rule. To perform this evaluation on a dataset $X_{1:n}$ of size $n$, the density function of $G_{n}$ is represented through its values on a fixed grid of $d\in\{5,000;\,1,000;\,500;\,100;\,50\}$ quadrature points over $\Theta=(0,U_{\Theta})$, where $U_\Theta=\max\{\max\{X_{1:n}\},\lceil Q_{n,0.99}+4\sqrt{\max\{Q_{n,0.99},1\}}\rceil\}$, with $Q_{n,0.99}=\text{Quantile}(X_{1:n};0.99)$. Further, we set $G_{0}$ to be Uniform over $\Theta$, and set $\alpha_{n}=(1+n)^{-0.99}$. Under this setting for Newton's algorithm, we obtain an estimate $G_n$ of the mixing distribution $G$, which, when substituted into \eqref{eq:est}, gives the quasi-Bayes EB estimates $\hat{S}^{\text{\tiny{[Q-B]}}}_{r,n}$ of $S_{r,n}$, for $r=1,2,3$. With regards to the estimation of $S_{1,n}$, Table~\ref{tab_weib_s1_sens_supp} provides the mean absolute deviations (MAD) and CPU times as $n$ and $d$ vary. The CPU time refers to the time (in seconds) for processing a new observation on a laptop MacBook Pro (M1 type processor).

\begin{table}[ht]
\centering
\caption{Weibull prior: MAD and CPU time (in seconds) of $\hat{S}^{\text{\tiny{[Q-B]}}}_{1,n}$ as $n$ and $d$ vary}
{
\setlength{\tabcolsep}{0pt}
\begin{tabular}{@{}l@{\hspace{1.2cm}}*{5}{>{\centering\arraybackslash}p{2.15cm}}@{}}
\hline
\hline
 & $d=5{,}000$ & $d=1{,}000$ & $d=500$ & $d=100$ & $d=50$ \\[0.1cm]
\hline
\multicolumn{6}{@{}l}{\underline{$n=1{,}000$}} \\[0.05cm]
MAD      & 0.0279 & 0.0279 & 0.0279 & 0.0279 & 0.0279 \\
CPU time & 0.0026 & 0.0010 & 0.0008 & 0.0005 & 0.0005 \\[0.2cm]

\multicolumn{6}{@{}l}{\underline{$n=2{,}000$}} \\[0.05cm]
MAD      & 0.0282 & 0.0282 & 0.0282 & 0.0282 & 0.0283 \\
CPU time & 0.0027 & 0.0011 & 0.0008 & 0.0005 & 0.0005 \\[0.2cm]

\multicolumn{6}{@{}l}{\underline{$n=3{,}000$}} \\[0.05cm]
MAD      & 0.0250 & 0.0250 & 0.0250 & 0.0250 & 0.0251 \\
CPU time & 0.0031 & 0.0011 & 0.0009 & 0.0005 & 0.0005 \\[0.2cm]

\multicolumn{6}{@{}l}{\underline{$n=4{,}000$}} \\[0.05cm]
MAD      & 0.0240 & 0.0240 & 0.0240 & 0.0240 & 0.0241 \\
CPU time & 0.0028 & 0.0011 & 0.0008 & 0.0005 & 0.0005 \\[0.2cm]

\multicolumn{6}{@{}l}{\underline{$n=5{,}000$}} \\[0.05cm]
MAD      & 0.0227 & 0.0227 & 0.0227 & 0.0227 & 0.0228 \\
CPU time & 0.0027 & 0.0010 & 0.0008 & 0.0005 & 0.0005 \\[0.2cm]

\multicolumn{6}{@{}l}{\underline{$n=6{,}000$}} \\[0.05cm]
MAD      & 0.0179 & 0.0179 & 0.0179 & 0.0180 & 0.0180 \\
CPU time & 0.0028 & 0.0011 & 0.0008 & 0.0005 & 0.0005 \\[0.2cm]

\multicolumn{6}{@{}l}{\underline{$n=7{,}000$}} \\[0.05cm]
MAD      & 0.0193 & 0.0193 & 0.0193 & 0.0193 & 0.0194 \\
CPU time & 0.0029 & 0.0011 & 0.0008 & 0.0005 & 0.0005 \\[0.2cm]

\multicolumn{6}{@{}l}{\underline{$n=8{,}000$}} \\[0.05cm]
MAD      & 0.0211 & 0.0211 & 0.0211 & 0.0212 & 0.0212 \\
CPU time & 0.0027 & 0.0011 & 0.0008 & 0.0005 & 0.0005 \\[0.2cm]

\multicolumn{6}{@{}l}{\underline{$n=9{,}000$}} \\[0.05cm]
MAD      & 0.0173 & 0.0173 & 0.0173 & 0.0173 & 0.0174 \\
CPU time & 0.0027 & 0.0011 & 0.0008 & 0.0005 & 0.0005 \\[0.2cm]

\multicolumn{6}{@{}l}{\underline{$n=10{,}000$}} \\[0.05cm]
MAD      & 0.0206 & 0.0206 & 0.0206 & 0.0207 & 0.0207 \\
CPU time & 0.0030 & 0.0011 & 0.0008 & 0.0005 & 0.0005 \\[0.1cm]
\hline
\hline
\end{tabular}
}
\label{tab_weib_s1_sens_supp}
\end{table}

Table~\ref{tab_weib_s1_supp} compares the MADs of $\hat{S}^{\text{\tiny{[O]}}}_{1,n}$, $\hat{S}^{\text{\tiny{[ML]}}}_{1,n}$, $\hat{S}^{\text{\tiny{[B]}}}_{1,n}$, $\hat{S}^{\text{\tiny[``u,v"]}}_{1,n}$ and $\hat{S}^{\text{\tiny{[Q-B]}}}_{1,n}$. Here, $\hat{S}^{\text{\tiny{[O]}}}_{1,n}=\E_{G}[S_{1,n}\,|\,X_{1:n}]$ with $G$ being the Weibull distribution with scale $5$ and shape $3$. In particular, for $\hat{S}^{\text{\tiny{[Q-B]}}}_{1,n}$ we consider: i) the fixed uniform grid of $d=1,000$ quadrature points over $\Theta=(0,U_{\Theta})$; ii) $G_{0}$ to be Uniform over $\Theta$; iii) the learning rate $\alpha_{n}=(1+n)^{-0.99}$. This is precisely the initialization considered in the second column of Table~\ref{tab_weib_s1_sens_supp}, which also provides CPU time. Furthermore, Table~\ref{tab_weib_s1ml_supp} reports the CPU times of $\hat{S}^{\text{\tiny{[Q-B]}}}_{1,n}$ as the sample size $n$ varies.

\begin{table}[ht]
\centering
\caption{Weibull prior, $S_{1,n}$: MAD as $n$ varies}
{
\setlength{\tabcolsep}{0pt}
\begin{tabular}{@{}l@{\hspace{1.2cm}}*{5}{>{\centering\arraybackslash}p{2.15cm}}@{}}
\hline
\hline
& $\hat{S}^{\text{\tiny{[O]}}}_{1,n}$ & $\hat{S}^{\text{\tiny{[ML]}}}_{1,n}$ & $\hat{S}^{\text{\tiny{[B]}}}_{1,n}$ & $\hat{S}^{\text{\tiny{[``u,v'']}}}_{1,n}$ & $\hat{S}^{\text{\tiny{[Q-B]}}}_{1,n}$ \\[0.1cm]
\hline
\multicolumn{6}{@{}l}{\underline{$n=1{,}000$}} \\[0.05cm]
MAD & 0.0118 & 0.3002 & 0.3002 & 0.0331 & \textbf{0.0279} \\[0.2cm]

\multicolumn{6}{@{}l}{\underline{$n=2{,}000$}} \\[0.05cm]
MAD & 0.0004 & 0.3060 & 0.3061 & 0.0124 & \textbf{0.0282} \\[0.2cm]

\multicolumn{6}{@{}l}{\underline{$n=3{,}000$}} \\[0.05cm]
MAD & 0.0022 & 0.3099 & 0.3099 & 0.0012 & \textbf{0.0250} \\[0.2cm]

\multicolumn{6}{@{}l}{\underline{$n=4{,}000$}} \\[0.05cm]
MAD & 0.0023 & 0.3094 & 0.3095 & 0.0066 & \textbf{0.0240} \\[0.2cm]

\multicolumn{6}{@{}l}{\underline{$n=5{,}000$}} \\[0.05cm]
MAD & 0.0048 & 0.3126 & 0.3126 & 0.0007 & \textbf{0.0227} \\[0.2cm]

\multicolumn{6}{@{}l}{\underline{$n=6{,}000$}} \\[0.05cm]
MAD & 0.0074 & 0.3071 & 0.3071 & 0.0118 & \textbf{0.0179} \\[0.2cm]

\multicolumn{6}{@{}l}{\underline{$n=7{,}000$}} \\[0.05cm]
MAD & 0.0051 & 0.3085 & 0.3085 & 0.0188 & \textbf{0.0193} \\[0.2cm]

\multicolumn{6}{@{}l}{\underline{$n=8{,}000$}} \\[0.05cm]
MAD & 0.0019 & 0.3112 & 0.3112 & 0.0107 & \textbf{0.0211} \\[0.2cm]

\multicolumn{6}{@{}l}{\underline{$n=9{,}000$}} \\[0.05cm]
MAD & 0.0038 & 0.3079 & 0.3079 & 0.0124 & \textbf{0.0173} \\[0.2cm]

\multicolumn{6}{@{}l}{\underline{$n=10{,}000$}} \\[0.05cm]
MAD & 0.0002 & 0.3114 & 0.3114 & 0.0058 & \textbf{0.0206} \\[0.1cm]
\hline
\hline
\end{tabular}
}
\label{tab_weib_s1_supp}
\end{table}

\begin{table}[ht]
\centering
\caption{Weibull prior: CPU time (in seconds) of $\hat{S}^{\text{\tiny{[Q-B]}}}_{1,n}$ for processing $n$ data}
{
\setlength{\tabcolsep}{0pt}
\begin{tabular}{@{}l@{\hspace{1.2cm}}>{\centering\arraybackslash}p{2.15cm}@{}}
\hline
\hline
& $\hat{S}^{\text{\tiny{[Q-B]}}}_{1,n}$ \\[0.1cm]
\hline
\multicolumn{2}{@{}l}{\underline{$n=1{,}000$}} \\[0.05cm]
CPU time & 1.0238 \\[0.2cm]

\multicolumn{2}{@{}l}{\underline{$n=2{,}000$}} \\[0.05cm]
CPU time & 2.1108 \\[0.2cm]

\multicolumn{2}{@{}l}{\underline{$n=3{,}000$}} \\[0.05cm]
CPU time & 3.3306 \\[0.2cm]

\multicolumn{2}{@{}l}{\underline{$n=4{,}000$}} \\[0.05cm]
CPU time & 2.1589 \\[0.2cm]

\multicolumn{2}{@{}l}{\underline{$n=5{,}000$}} \\[0.05cm]
CPU time & 4.3178 \\[0.2cm]

\multicolumn{2}{@{}l}{\underline{$n=6{,}000$}} \\[0.05cm]
CPU time & 3.3566 \\[0.2cm]

\multicolumn{2}{@{}l}{\underline{$n=7{,}000$}} \\[0.05cm]
CPU time & 6.7132 \\[0.2cm]

\multicolumn{2}{@{}l}{\underline{$n=8{,}000$}} \\[0.05cm]
CPU time & 8.5050 \\[0.2cm]

\multicolumn{2}{@{}l}{\underline{$n=9{,}000$}} \\[0.05cm]
CPU time & 9.6956 \\[0.2cm]

\multicolumn{2}{@{}l}{\underline{$n=10{,}000$}} \\[0.05cm]
CPU time & 10.9712 \\[0.1cm]
\hline
\hline
\end{tabular}
}
\label{tab_weib_s1ml_supp}
\end{table}

To conclude, we provide quasi-Bayes credible intervals at level $1-\alpha=0.95$. Specifically, credible intervals are constructed from \eqref{cred_interval} by relying on Newton's algorithm initialized as in Table~\ref{tab_weib_s1_supp}. Figure  \ref{weibull_fig_interval_s1}-\ref{weibull_fig_interval_s3} display the quasi-Bayes credible intervals for $S_{1,n}$ and $S_{3,n}$, respectively, with the corresponding oracle credible intervals $I^\ast_{S_{1,n}}$ and $I^\ast_{S_{3,n}}$ under the Weibull prior.

\begin{figure}[h!]
\begin{center}
\includegraphics[width=1\linewidth,height=0.36\textheight,keepaspectratio]{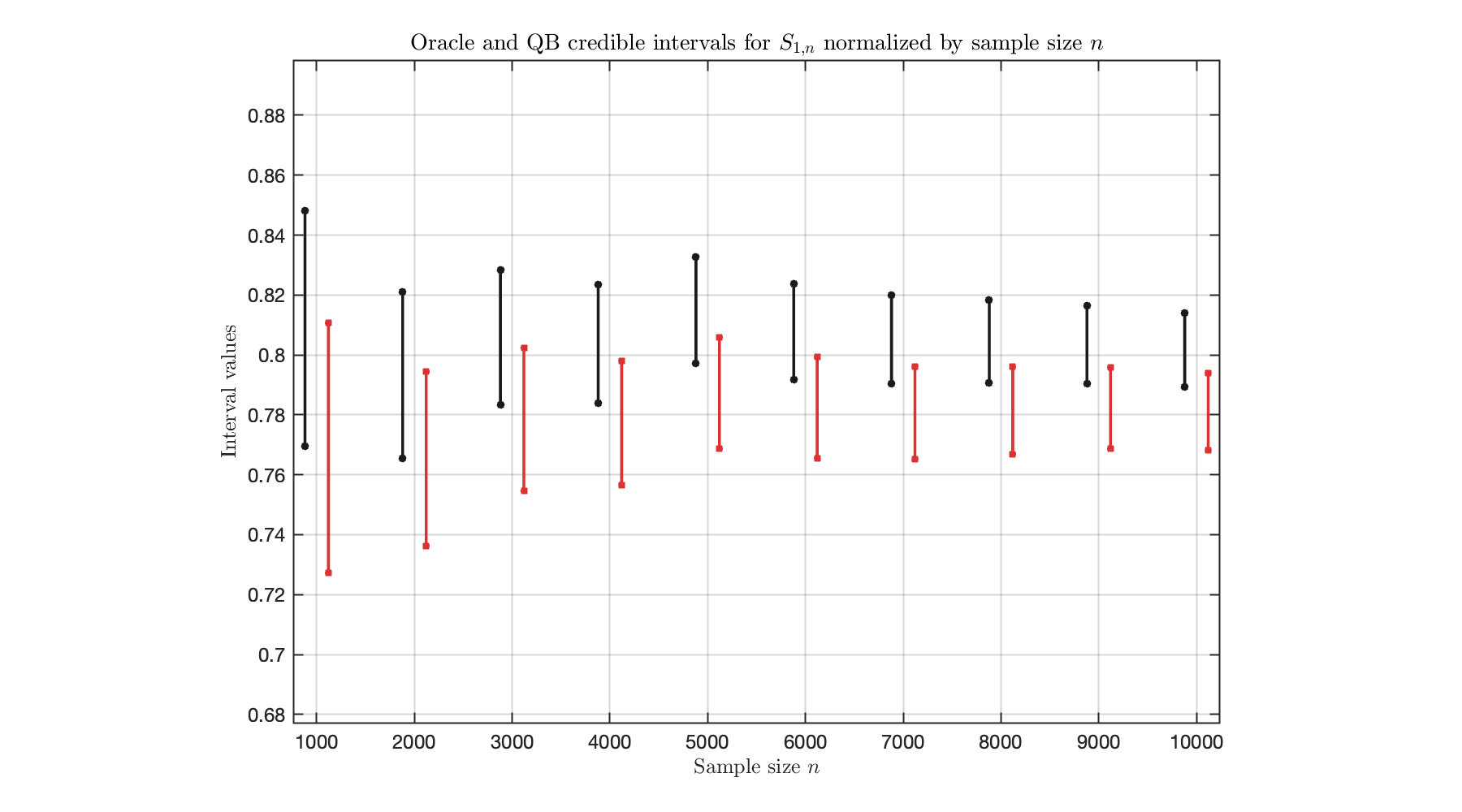}
\end{center}
\caption{\scriptsize{Weibull prior, $S_{1,n}$: oracle credible intervals (black) and quasi-Bayes credible intervals (red)}}
\label{weibull_fig_interval_s1}
\end{figure}

\begin{figure}[h!]
\begin{center}
\includegraphics[width=1\linewidth,height=0.36\textheight,keepaspectratio]{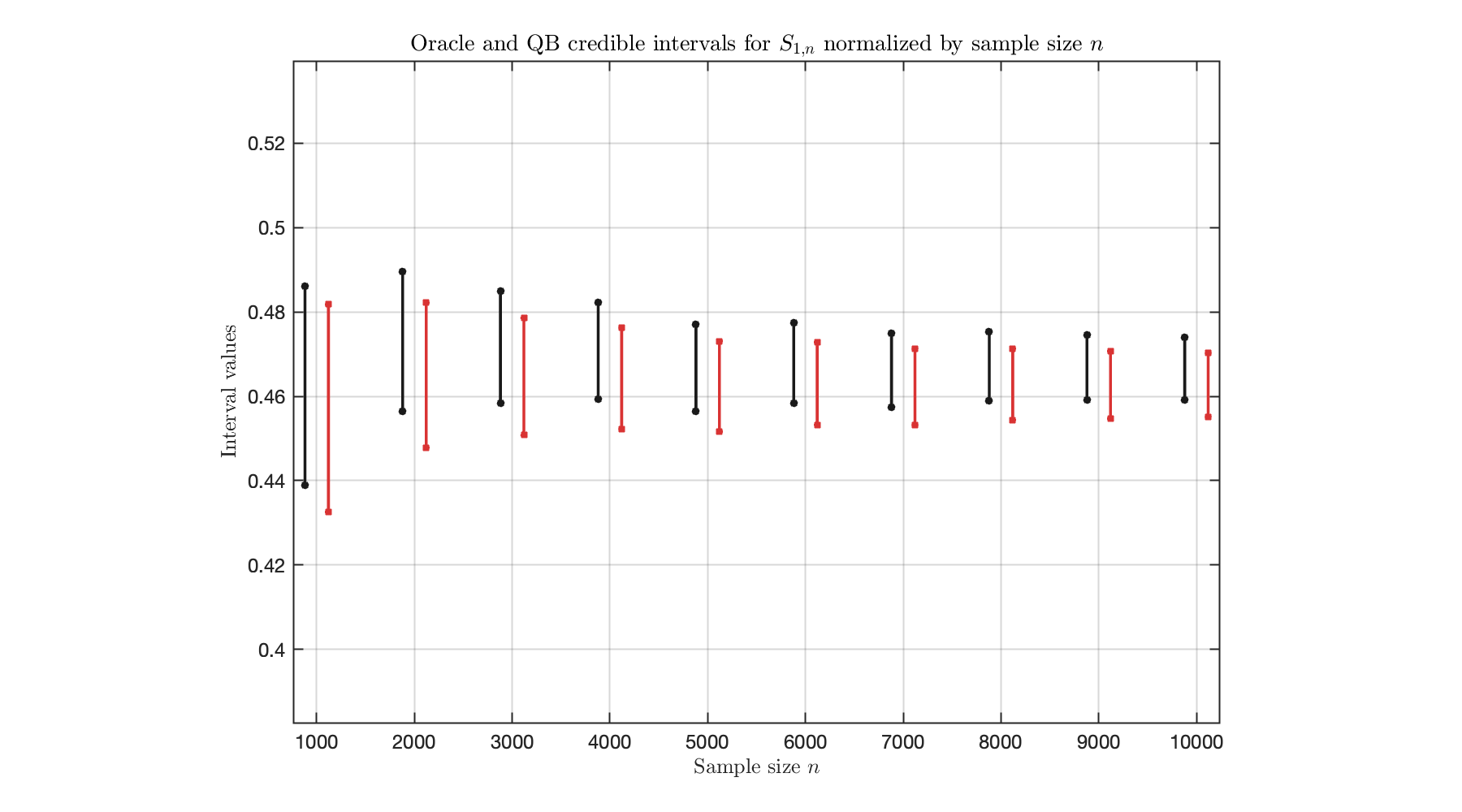}
\end{center}
\caption{\scriptsize{Weibull prior, $S_{3,n}$: oracle credible intervals (black) and quasi-Bayes credible intervals (red)}}
\label{weibull_fig_interval_s3}
\end{figure}

\subsubsection{Half-Gaussian prior}

For $i=1,\ldots,100$, let $\mathbf{X}_{i}=X_{1:100 i}$ denote a dataset of size $n=100i$ generated from the Poisson mixture model \eqref{model_distr_poi}, with a half-Gaussian prior $G$, namely the distribution of the positive part of a Gaussian random variable with mean $0$ and variance $1$. The $\mathbf{X}_{i}$'s are nested, so that the sample size increases progressively: at stage $i$, we have $n=100i$ data, obtained by adding $100$ new data at each step, starting from $n=100$. For each dataset $\mathbf{X}_{i}$ of size $n=100i$, we apply the quasi-Bayes EB approach to estimate $S_{1,n}$ and $S_{2,n}$, with $\kappa=2$, and $S_{3,n}$.

The integral in Newton's algorithm \eqref{eq:newton} is evaluated numerically via the trapezoidal rule. To perform this evaluation on a dataset $X_{1:n}$ of size $n$, the density function of $G_{n}$ is represented through its values on a fixed grid of $d\in\{5,000;\,1,000;\,500;\,100;\,50\}$ quadrature points over $\Theta=(0,U_{\Theta})$, where $U_\Theta=\max\{\max\{X_{1:n}\},\lceil Q_{n,0.99}+4\sqrt{\max\{Q_{n,0.99},1\}}\rceil\}$, with $Q_{n,0.99}=\text{Quantile}(X_{1:n};0.99)$. Further, we set $G_{0}$ to be Uniform over $\Theta$, and set $\alpha_{n}=(1+n)^{-0.99}$. Under this setting for Newton's algorithm, we obtain an estimate $G_n$ of the mixing distribution $G$, which, when substituted into \eqref{eq:est}, gives the quasi-Bayes EB estimates $\hat{S}^{\text{\tiny{[Q-B]}}}_{r,n}$ of $S_{r,n}$, for $r=1,2,3$. With regards to the estimation of $S_{1,n}$, Table~\ref{tab_halfgauss_s1_sens_supp} provides the mean absolute deviations (MAD) and CPU times as $n$ and $d$ vary. The CPU time refers to the time (in seconds) for processing a new observation on a laptop MacBook Pro (M1 type processor).

\begin{table}[ht]
\centering
\caption{Half-Gaussian prior: MAD and CPU time (in seconds) of $\hat{S}^{\text{\tiny{[Q-B]}}}_{1,n}$ as $n$ and $d$ vary}
{
\setlength{\tabcolsep}{0pt}
\begin{tabular}{@{}l@{\hspace{1.2cm}}*{5}{>{\centering\arraybackslash}p{2.15cm}}@{}}
\hline
\hline
 & $d=5{,}000$ & $d=1{,}000$ & $d=500$ & $d=100$ & $d=50$ \\[0.1cm]
\hline
\multicolumn{6}{@{}l}{\underline{$n=1{,}000$}} \\[0.05cm]
MAD      & 0.0108 & 0.0108 & 0.0108 & 0.0108 & 0.0109 \\
CPU time & 0.0021 & 0.0005 & 0.0004 & 0.0003 & 0.0003 \\[0.2cm]

\multicolumn{6}{@{}l}{\underline{$n=2{,}000$}} \\[0.05cm]
MAD      & 0.0098 & 0.0098 & 0.0098 & 0.0098 & 0.0100 \\
CPU time & 0.0021 & 0.0005 & 0.0004 & 0.0003 & 0.0003 \\[0.2cm]

\multicolumn{6}{@{}l}{\underline{$n=3{,}000$}} \\[0.05cm]
MAD      & 0.0037 & 0.0037 & 0.0037 & 0.0038 & 0.0040 \\
CPU time & 0.0022 & 0.0007 & 0.0005 & 0.0004 & 0.0004 \\[0.2cm]

\multicolumn{6}{@{}l}{\underline{$n=4{,}000$}} \\[0.05cm]
MAD      & 0.0043 & 0.0043 & 0.0043 & 0.0043 & 0.0046 \\
CPU time & 0.0023 & 0.0009 & 0.0007 & 0.0004 & 0.0004 \\[0.2cm]

\multicolumn{6}{@{}l}{\underline{$n=5{,}000$}} \\[0.05cm]
MAD      & 0.0031 & 0.0031 & 0.0031 & 0.0031 & 0.0033 \\
CPU time & 0.0023 & 0.0009 & 0.0007 & 0.0004 & 0.0004 \\[0.2cm]

\multicolumn{6}{@{}l}{\underline{$n=6{,}000$}} \\[0.05cm]
MAD      & 0.0049 & 0.0049 & 0.0049 & 0.0050 & 0.0052 \\
CPU time & 0.0025 & 0.0008 & 0.0006 & 0.0004 & 0.0004 \\[0.2cm]

\multicolumn{6}{@{}l}{\underline{$n=7{,}000$}} \\[0.05cm]
MAD      & 0.0024 & 0.0024 & 0.0024 & 0.0024 & 0.0027 \\
CPU time & 0.0023 & 0.0009 & 0.0007 & 0.0004 & 0.0004 \\[0.2cm]

\multicolumn{6}{@{}l}{\underline{$n=8{,}000$}} \\[0.05cm]
MAD      & 0.0024 & 0.0025 & 0.0025 & 0.0026 & 0.0032 \\
CPU time & 0.0024 & 0.0009 & 0.0007 & 0.0004 & 0.0004 \\[0.2cm]

\multicolumn{6}{@{}l}{\underline{$n=9{,}000$}} \\[0.05cm]
MAD      & 0.0002 & 0.0002 & 0.0003 & 0.0004 & 0.0010 \\
CPU time & 0.0026 & 0.0008 & 0.0006 & 0.0004 & 0.0003 \\[0.2cm]

\multicolumn{6}{@{}l}{\underline{$n=10{,}000$}} \\[0.05cm]
MAD      & 0.0018 & 0.0017 & 0.0552 & 0.0552 & 0.0551 \\
CPU time & 0.0023 & 0.0009 & 0.0017 & 0.0016 & 0.0010 \\[0.1cm]
\hline
\hline
\end{tabular}
}
\label{tab_halfgauss_s1_sens_supp}
\end{table}

Figure~\ref{fig_halfgauss_s1_supp} and Table~\ref{tab_halfgauss_s1_supp} compares the MADs for $\hat{S}^{\text{\tiny{[O]}}}_{1,n}$, $\hat{S}^{\text{\tiny{[ML]}}}_{1,n}$, $\hat{S}^{\text{\tiny{[B]}}}_{1,n}$, $\hat{S}^{\text{\tiny[``u,v"]}}_{1,n}$ and $\hat{S}^{\text{\tiny{[Q-B]}}}_{1,n}$. Here, $\hat{S}^{\text{\tiny{[O]}}}_{1,n}=\E_{G}[S_{1,n}\,|\,X_{1:n}]$ with $G$ being the half-Gaussian distribution. In particular, for $\hat{S}^{\text{\tiny{[Q-B]}}}_{1,n}$ we consider: i) the fixed uniform grid of $d=1,000$ quadrature points over $\Theta=(0,U_{\Theta})$; ii) $G_{0}$ to be Uniform over $\Theta$; iii) the learning rate to be $\alpha_{n}=(1+n)^{-0.99}$. This is precisely the initialization considered in the second column of Table~\ref{tab_halfgauss_s1_sens_supp}, which also provides CPU time. 

\begin{figure}[h!]
\begin{center}
\includegraphics[width=1\linewidth,height=0.36\textheight,keepaspectratio]{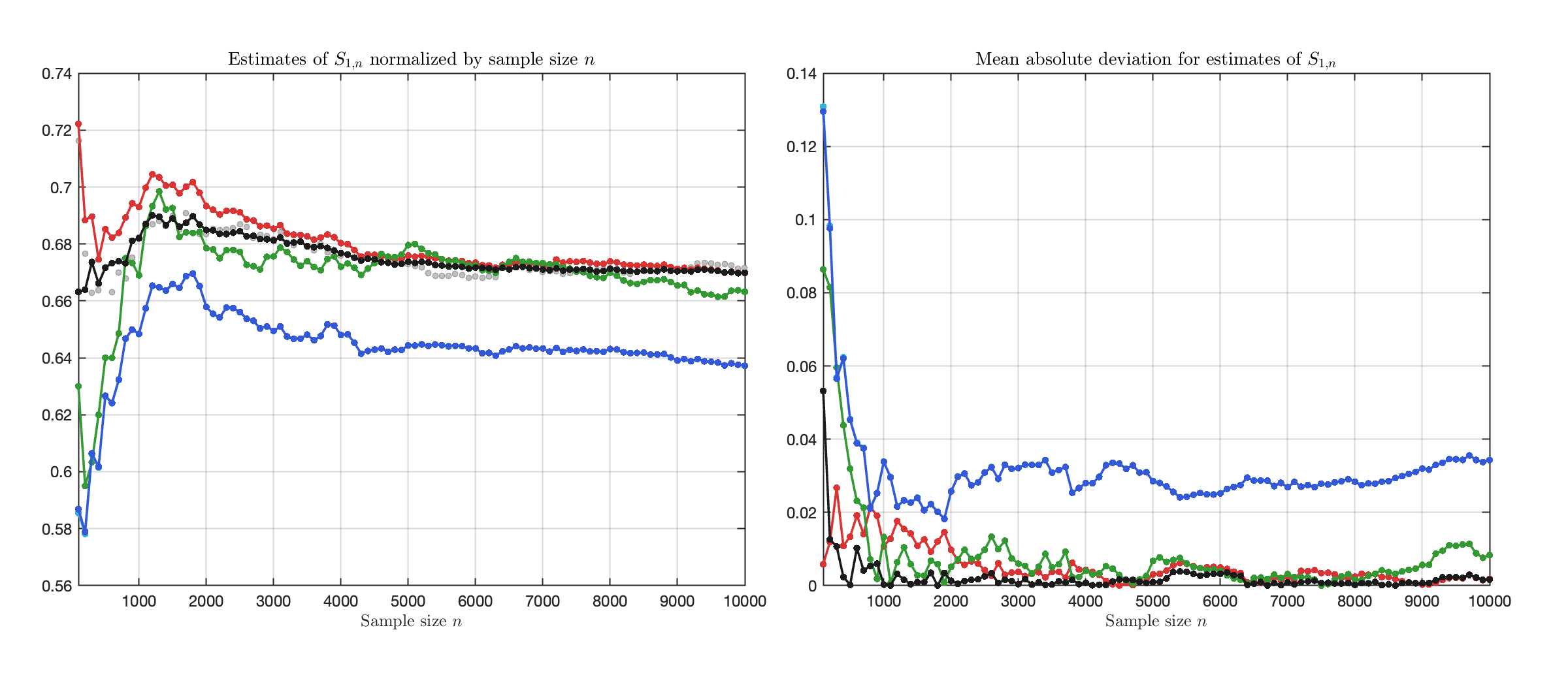}
\end{center}
\caption{\scriptsize{Half-Gaussian prior, $S_{1,n}$. Left panel: true values $n^{-1}S_{1,n}$ (Grey o-) and estimates $n^{-1}\hat{S}^{\text{\tiny{[O]}}}_{1,n}$ (Black .-), $n^{-1}\hat{S}^{\text{\tiny{[ML]}}}_{1,n}$ (Blue .-), $n^{-1}\hat{S}^{\text{\tiny{[B]}}}_{1,n}$ (Cyan .-), $n^{-1}\hat{S}^{\text{\tiny[``u,v"]}}_{1,n}$ (Green .-) and $n^{-1}\hat{S}^{\text{\tiny{[Q-B]}}}_{1,n}$ (Red .-). Right panel: MAD of $\hat{S}^{\text{\tiny{[O]}}}_{1,n}$ (Black .-), $\hat{S}^{\text{\tiny{[ML]}}}_{1,n}$ (Blue .-), $\hat{S}^{\text{\tiny{[B]}}}_{1,n}$ (Cyan .-), $\hat{S}^{\text{\tiny[``u,v"]}}_{1,n}$ (Green .-) and $\hat{S}^{\text{\tiny{[Q-B]}}}_{1,n}$ (Red .-)}}
\label{fig_halfgauss_s1_supp}
\end{figure}

\begin{table}[ht]
\centering
\caption{Half-Gaussian prior, $S_{1,n}$: MAD as $n$ varies}
{
\setlength{\tabcolsep}{0pt}
\begin{tabular}{@{}l@{\hspace{1.2cm}}*{5}{>{\centering\arraybackslash}p{2.15cm}}@{}}
\hline
\hline
& $\hat{S}^{\text{\tiny{[O]}}}_{1,n}$ & $\hat{S}^{\text{\tiny{[ML]}}}_{1,n}$ & $\hat{S}^{\text{\tiny{[B]}}}_{1,n}$ & $\hat{S}^{\text{\tiny{[``u,v'']}}}_{1,n}$ & $\hat{S}^{\text{\tiny{[Q-B]}}}_{1,n}$ \\[0.1cm]
\hline
\multicolumn{6}{@{}l}{\underline{$n=1{,}000$}} \\[0.05cm]
MAD & 0.0002 & 0.0339 & 0.0338 & 0.0132 & \textbf{0.0108} \\[0.2cm]

\multicolumn{6}{@{}l}{\underline{$n=2{,}000$}} \\[0.05cm]
MAD & 0.0013 & 0.0258 & 0.0257 & 0.0051 & \textbf{0.0098} \\[0.2cm]

\multicolumn{6}{@{}l}{\underline{$n=3{,}000$}} \\[0.05cm]
MAD & 0.0004 & 0.0322 & 0.0321 & 0.0060 & \textbf{0.0037} \\[0.2cm]

\multicolumn{6}{@{}l}{\underline{$n=4{,}000$}} \\[0.05cm]
MAD & 0.0007 & 0.0280 & 0.0280 & 0.0040 & \textbf{0.0043} \\[0.2cm]

\multicolumn{6}{@{}l}{\underline{$n=5{,}000$}} \\[0.05cm]
MAD & 0.0009 & 0.0285 & 0.0285 & 0.0067 & \textbf{0.0031} \\[0.2cm]

\multicolumn{6}{@{}l}{\underline{$n=6{,}000$}} \\[0.05cm]
MAD & 0.0032 & 0.0252 & 0.0252 & 0.0040 & \textbf{0.0049} \\[0.2cm]

\multicolumn{6}{@{}l}{\underline{$n=7{,}000$}} \\[0.05cm]
MAD & 0.0009 & 0.0269 & 0.0269 & 0.0031 & \textbf{0.0024} \\[0.2cm]

\multicolumn{6}{@{}l}{\underline{$n=8{,}000$}} \\[0.05cm]
MAD & 0.0002 & 0.0284 & 0.0284 & 0.0016 & \textbf{0.0024} \\[0.2cm]

\multicolumn{6}{@{}l}{\underline{$n=9{,}000$}} \\[0.05cm]
MAD & 0.0006 & 0.0320 & 0.0320 & 0.0056 & \textbf{0.0002} \\[0.2cm]

\multicolumn{6}{@{}l}{\underline{$n=10{,}000$}} \\[0.05cm]
MAD & 0.0016 & 0.0343 & 0.0343 & 0.0083 & \textbf{0.0018} \\[0.1cm]
\hline
\hline
\end{tabular}
}
\label{tab_halfgauss_s1_supp}
\end{table}

Figure~\ref{fig_halfgauss_s3_supp} and Table~\ref{tab_halfgauss_s3_supp} report corresponding results with respect to the estimation of $S_{3,n}$.

\begin{figure}[h!]
\begin{center}
\includegraphics[width=1\linewidth,height=0.36\textheight,keepaspectratio]{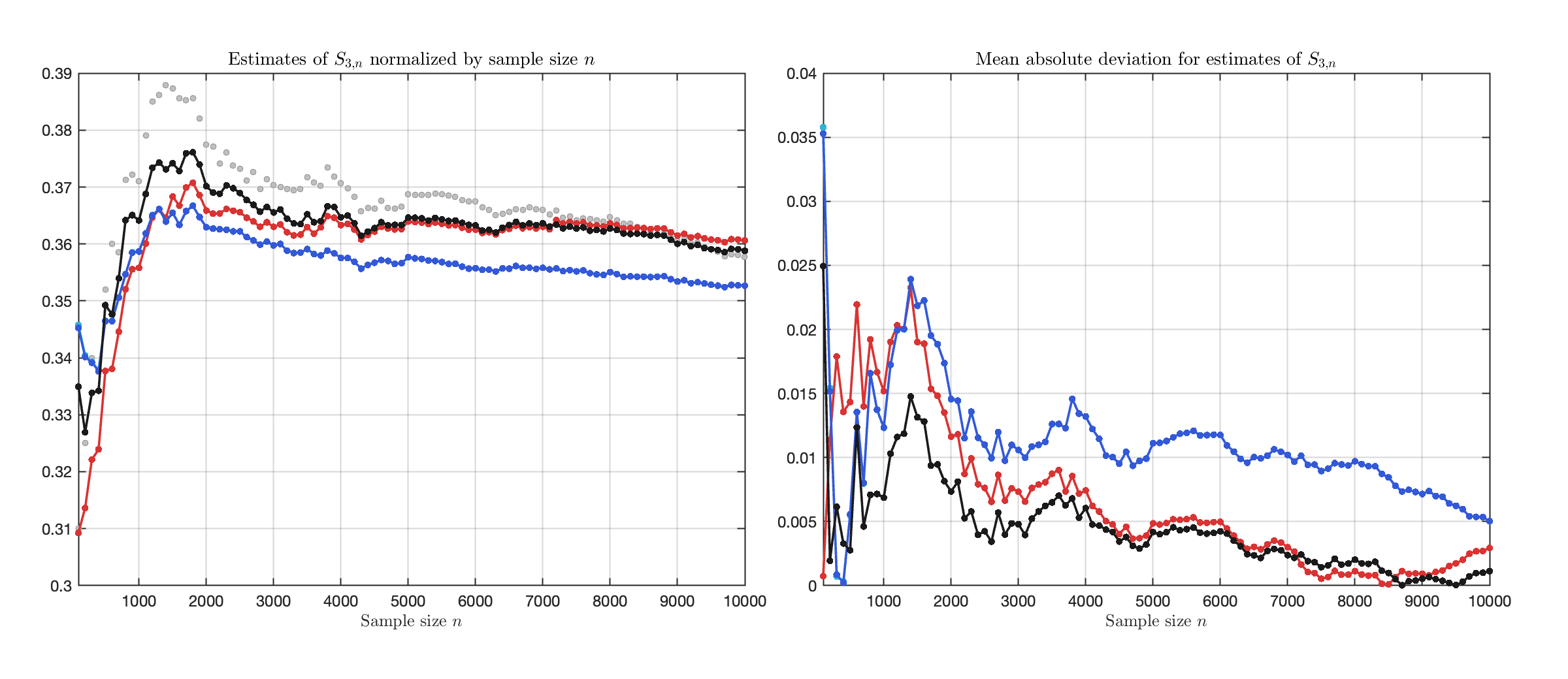}
\end{center}
\caption{\scriptsize{Half-Gaussian prior, $S_{3,n}$. Left panel: true values of $n^{-1}S_{3,n}$ (Grey o-) and estimates $n^{-1}\hat{S}^{\text{\tiny{[O]}}}_{3,n}$ (Black .-), $n^{-1}\hat{S}^{\text{\tiny{[ML]}}}_{3,n}$ (Blue .-), $n^{-1}\hat{S}^{\text{\tiny{[B]}}}_{3,n}$ (Cyan .-)  and $n^{-1}\hat{S}^{\text{\tiny{[Q-B]}}}_{3,n}$ (Red .-). Right panel: MAD of $\hat{S}^{\text{\tiny{[O]}}}_{3,n}$ (Black .-), $\hat{S}^{\text{\tiny{[ML]}}}_{3,n}$ (Blue .-), $\hat{S}^{\text{\tiny{[B]}}}_{3,n}$ (Cyan .-) and $\hat{S}^{\text{\tiny{[Q-B]}}}_{3,n}$ (Red .-)}}
\label{fig_halfgauss_s3_supp}
\end{figure}

\begin{table}[ht]
\centering
\caption{Half-Gaussian prior, $S_{3,n}$: MAD as $n$ varies}
{
\setlength{\tabcolsep}{0pt}
\begin{tabular}{@{}l@{\hspace{1.2cm}}*{4}{>{\centering\arraybackslash}p{2.15cm}}@{}}
\hline
\hline
& $\hat{S}^{\text{\tiny{[O]}}}_{3,n}$ & $\hat{S}^{\text{\tiny{[ML]}}}_{3,n}$ & $\hat{S}^{\text{\tiny{[B]}}}_{3,n}$ & $\hat{S}^{\text{\tiny{[Q-B]}}}_{3,n}$ \\[0.1cm]
\hline
\multicolumn{5}{@{}l}{\underline{$n=1{,}000$}} \\[0.05cm]
MAD & 0.0069 & 0.0123 & 0.0123 & \textbf{0.0152} \\[0.2cm]

\multicolumn{5}{@{}l}{\underline{$n=2{,}000$}} \\[0.05cm]
MAD & 0.0073 & 0.0145 & 0.0146 & \textbf{0.0116} \\[0.2cm]

\multicolumn{5}{@{}l}{\underline{$n=3{,}000$}} \\[0.05cm]
MAD & 0.0048 & 0.0106 & 0.0106 & \textbf{0.0073} \\[0.2cm]

\multicolumn{5}{@{}l}{\underline{$n=4{,}000$}} \\[0.05cm]
MAD & 0.0060 & 0.0132 & 0.0132 & \textbf{0.0074} \\[0.2cm]

\multicolumn{5}{@{}l}{\underline{$n=5{,}000$}} \\[0.05cm]
MAD & 0.0042 & 0.0111 & 0.0111 & \textbf{0.0049} \\[0.2cm]

\multicolumn{5}{@{}l}{\underline{$n=6{,}000$}} \\[0.05cm]
MAD & 0.0042 & 0.0117 & 0.0118 & \textbf{0.0050} \\[0.2cm]

\multicolumn{5}{@{}l}{\underline{$n=7{,}000$}} \\[0.05cm]
MAD & 0.0024 & 0.0102 & 0.0102 & \textbf{0.0030} \\[0.2cm]

\multicolumn{5}{@{}l}{\underline{$n=8{,}000$}} \\[0.05cm]
MAD & 0.0020 & 0.0097 & 0.0097 & \textbf{0.0011} \\[0.2cm]

\multicolumn{5}{@{}l}{\underline{$n=9{,}000$}} \\[0.05cm]
MAD & 0.0005 & 0.0071 & 0.0071 & \textbf{0.0009} \\[0.2cm]

\multicolumn{5}{@{}l}{\underline{$n=10{,}000$}} \\[0.05cm]
MAD & 0.0011 & 0.0050 & 0.0050 & \textbf{0.0029} \\[0.1cm]
\hline
\hline
\end{tabular}
}
\label{tab_halfgauss_s3_supp}
\end{table}

To conclude, we provide quasi-Bayes credible intervals at level $1-\alpha=0.95$. Specifically, credible intervals are constructed from \eqref{cred_interval} by relying on Newton's algorithm initialized as in Table~\ref{tab_halfgauss_s1_supp}. Figure~\ref{halfgauss_fig_interval_s1}-\ref{halfgauss_fig_interval_s3} display the quasi-Bayes credible intervals for $S_{1,n}$ and $S_{3,n}$, respectively, with the corresponding oracle credible intervals $I^\ast_{S_{1,n}}$ and $I^\ast_{S_{3,n}}$ under the half-Gaussian prior.

\begin{figure}[h!]
\begin{center}
\includegraphics[width=1\linewidth,height=0.36\textheight,keepaspectratio]{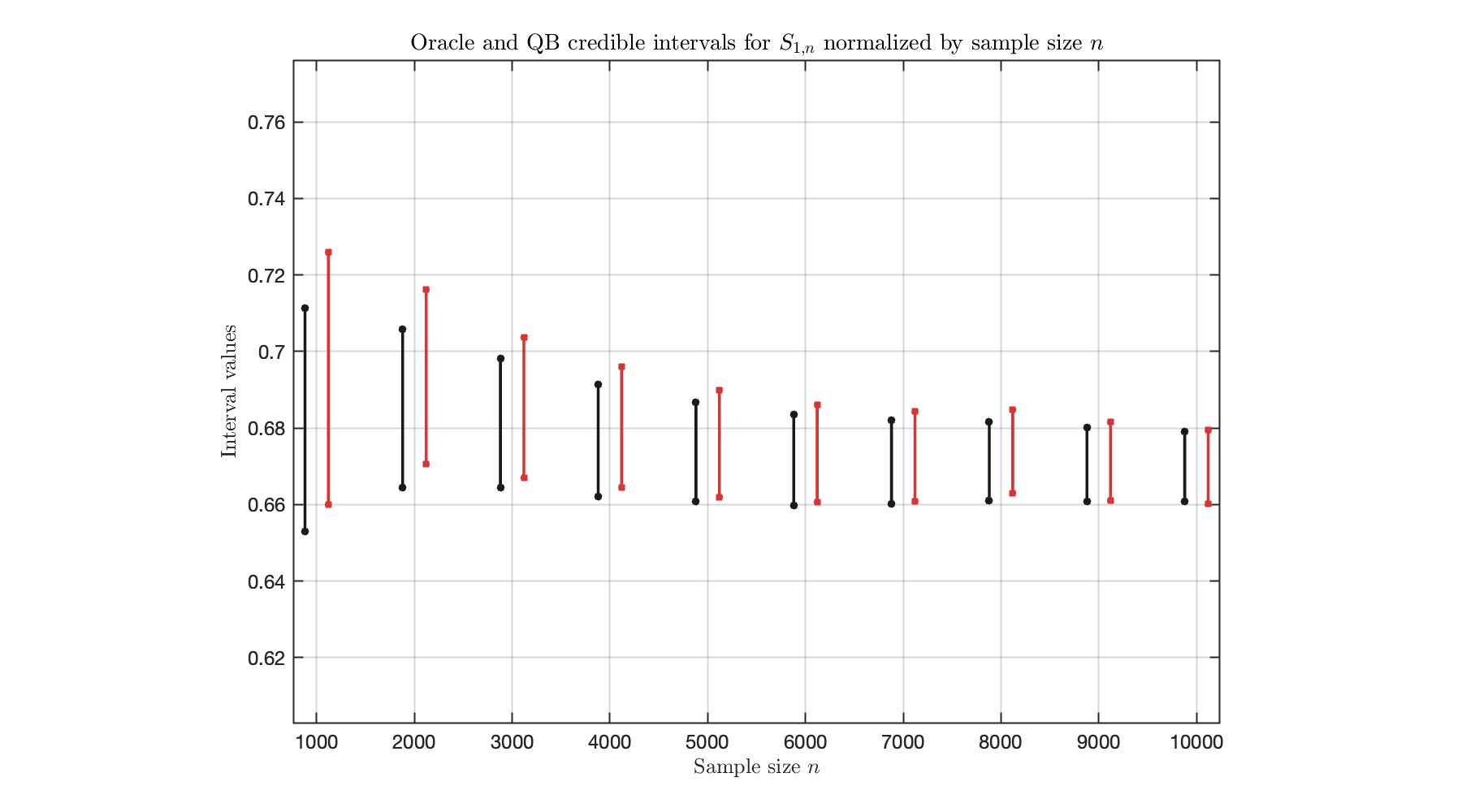}
\end{center}
\caption{\scriptsize{Half-Gaussian prior, $S_{1,n}$: oracle credible intervals (black) and quasi-Bayes credible intervals (red)}}
\label{halfgauss_fig_interval_s1}
\end{figure}

\begin{figure}[h!]
\begin{center}
\includegraphics[width=1\linewidth,height=0.36\textheight,keepaspectratio]{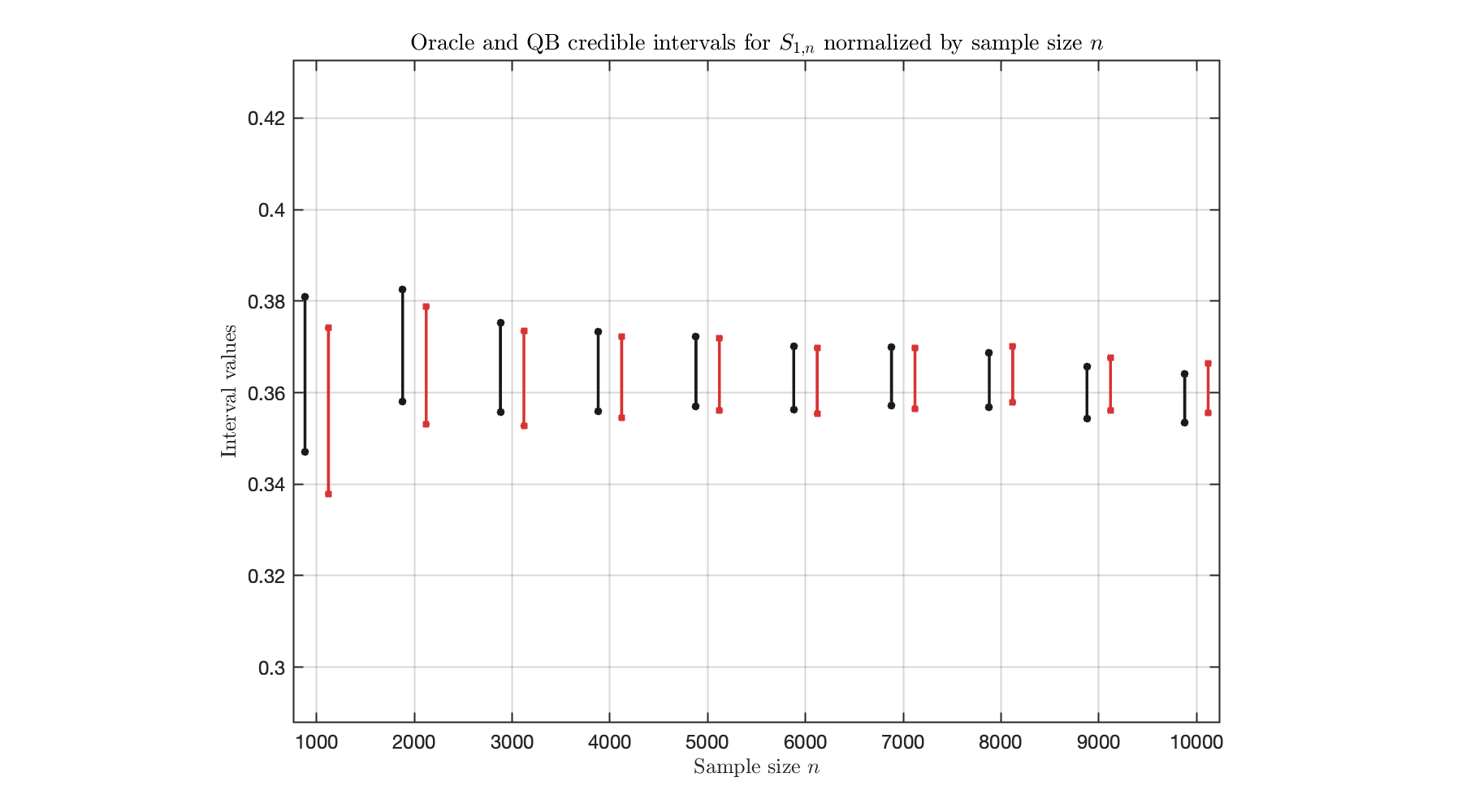}
\end{center}
\caption{\scriptsize{Half-Gaussian prior, $S_{3,n}$: oracle credible intervals (black) and quasi-Bayes credible intervals (red)}}
\label{halfgauss_fig_interval_s3}
\end{figure}

\subsubsection{Square-root of half-Cauchy prior}

For $i=1,\ldots,100$, let $\mathbf{X}_{i}=X_{1:100 i}$ denote a dataset of size $n=100i$ generated from the Poisson mixture model \eqref{model_distr_poi}, with a square-root of half-Cauchy prior $G$, namely the distribution of the square-root of the positive part of a standard Cauchy random variable. The $\mathbf{X}_{i}$'s are nested, so that the sample size increases progressively: at stage $i$, we have $n=100i$ data, obtained by adding $100$ new data at each step, starting from $n=100$. For each dataset $\mathbf{X}_{i}$ of size $n=100i$, we apply the quasi-Bayes EB approach to estimate $S_{1,n}$ and $S_{2,n}$, with $\kappa=2$, and $S_{3,n}$.

The integral in Newton's algorithm \eqref{eq:newton} is evaluated numerically via the trapezoidal rule. To perform this evaluation on a dataset $X_{1:n}$ of size $n$, the density function of $G_{n}$ is represented through its values on a fixed grid of $d\in\{5,000;\,1,000;\,500;\,100;\,50\}$ quadrature points over $\Theta=(0,U_{\Theta})$, where $U_\Theta=\max\{\max\{X_{1:n}\},\lceil Q_{n,0.99}+4\sqrt{\max\{Q_{n,0.99},1\}}\rceil\}$, with $Q_{n,0.99}=\text{Quantile}(X_{1:n};0.99)$. Further, we set $G_{0}$ to be Uniform over $\Theta$, and set $\alpha_{n}=(1+n)^{-0.99}$. Under this setting for Newton's algorithm, we obtain an estimate $G_n$ of the mixing distribution $G$, which, when substituted into \eqref{eq:est}, gives the quasi-Bayes EB estimates $\hat{S}^{\text{\tiny{[Q-B]}}}_{r,n}$ of $S_{r,n}$, for $r=1,2,3$. With regards to the estimation of $S_{1,n}$, Table~\ref{tab_hafcauchy_s1_sens_supp}  provides the mean absolute deviations (MAD) and CPU times as $n$ and $d$ vary. The CPU time refers to the time (in seconds) for processing a new observation on a laptop MacBook Pro (M1 type processor).

\begin{table}[ht]
\centering
\caption{Square-root of half-Cauchy prior: MAD and CPU time (in seconds) of $\hat{S}^{\text{\tiny{[Q-B]}}}_{1,n}$ as $n$ and $d$ vary}
{
\setlength{\tabcolsep}{0pt}
\begin{tabular}{@{}l@{\hspace{1.2cm}}*{5}{>{\centering\arraybackslash}p{2.15cm}}@{}}
\hline
\hline
 & $d=5{,}000$ & $d=1{,}000$ & $d=500$ & $d=100$ & $d=50$ \\[0.1cm]
\hline
\multicolumn{6}{@{}l}{\underline{$n=1{,}000$}} \\[0.05cm]
MAD      & 0.0410 & 0.0412 & 0.0420 & 0.1476 & 0.3713 \\
CPU time & 0.0022 & 0.0008 & 0.0007 & 0.0004 & 0.0004 \\[0.2cm]

\multicolumn{6}{@{}l}{\underline{$n=2{,}000$}} \\[0.05cm]
MAD      & 0.0352 & 0.0354 & 0.0360 & 0.1350 & 0.3521 \\
CPU time & 0.0028 & 0.0009 & 0.0007 & 0.0004 & 0.0004 \\[0.2cm]

\multicolumn{6}{@{}l}{\underline{$n=3{,}000$}} \\[0.05cm]
MAD      & 0.0433 & 0.0434 & 0.0440 & 0.1388 & 0.3671 \\
CPU time & 0.0027 & 0.0008 & 0.0006 & 0.0004 & 0.0004 \\[0.2cm]

\multicolumn{6}{@{}l}{\underline{$n=4{,}000$}} \\[0.05cm]
MAD      & 0.0327 & 0.0328 & 0.0333 & 0.1239 & 0.3596 \\
CPU time & 0.0024 & 0.0009 & 0.0007 & 0.0004 & 0.0004 \\[0.2cm]

\multicolumn{6}{@{}l}{\underline{$n=5{,}000$}} \\[0.05cm]
MAD      & 0.0273 & 0.0274 & 0.0279 & 0.1173 & 0.3472 \\
CPU time & 0.0023 & 0.0007 & 0.0006 & 0.0004 & 0.0004 \\[0.2cm]

\multicolumn{6}{@{}l}{\underline{$n=6{,}000$}} \\[0.05cm]
MAD      & 0.0280 & 0.0281 & 0.0285 & 0.1150 & 0.3452 \\
CPU time & 0.0025 & 0.0007 & 0.0007 & 0.0004 & 0.0004 \\[0.2cm]

\multicolumn{6}{@{}l}{\underline{$n=7{,}000$}} \\[0.05cm]
MAD      & 0.0229 & 0.0233 & 0.0253 & 0.1917 & 0.8017 \\
CPU time & 0.0021 & 0.0009 & 0.0007 & 0.0004 & 0.0004 \\[0.2cm]

\multicolumn{6}{@{}l}{\underline{$n=8{,}000$}} \\[0.05cm]
MAD      & 0.0223 & 0.0226 & 0.0245 & 0.1914 & 0.8055 \\
CPU time & 0.0024 & 0.0008 & 0.0007 & 0.0004 & 0.0004 \\[0.2cm]

\multicolumn{6}{@{}l}{\underline{$n=9{,}000$}} \\[0.05cm]
MAD      & 0.0208 & 0.0212 & 0.0229 & 0.1904 & 0.8083 \\
CPU time & 0.0024 & 0.0009 & 0.0006 & 0.0004 & 0.0004 \\[0.2cm]

\multicolumn{6}{@{}l}{\underline{$n=10{,}000$}} \\[0.05cm]
MAD      & 0.0195 & 0.0198 & 0.0216 & 0.1888 & 0.8088 \\
CPU time & 0.0023 & 0.0008 & 0.0007 & 0.0004 & 0.0004 \\[0.1cm]
\hline
\hline
\end{tabular}
}
\label{tab_hafcauchy_s1_sens_supp}
\end{table}

Figure~\ref{fig_hafcauchy_s1_supp} and Table~\ref{tab_hafcauchy_s1_supp} compares the MADs for $\hat{S}^{\text{\tiny{[O]}}}_{1,n}$, $\hat{S}^{\text{\tiny{[ML]}}}_{1,n}$, $\hat{S}^{\text{\tiny{[B]}}}_{1,n}$, $\hat{S}^{\text{\tiny[``u,v"]}}_{1,n}$ and $\hat{S}^{\text{\tiny{[Q-B]}}}_{1,n}$. Here, $\hat{S}^{\text{\tiny{[O]}}}_{1,n}=\E_{G}[S_{1,n}\,|\,X_{1:n}]$ with $G$ being the square-root of half-Cauchy distribution. In particular, for $\hat{S}^{\text{\tiny{[Q-B]}}}_{1,n}$ we consider: i) the fixed uniform grid of $d=1,000$ quadrature points over $\Theta=(0,U_{\Theta})$; ii) $G_{0}$ to be Uniform over $\Theta$; iii) the learning rate to be $\alpha_{n}=(1+n)^{-0.99}$. This is precisely the initialization considered in the second column of Table~ \ref{tab_hafcauchy_s1_sens_supp}, which also provides CPU time. 

\begin{figure}[h!]
\begin{center}
\includegraphics[width=1\linewidth,height=0.36\textheight,keepaspectratio]{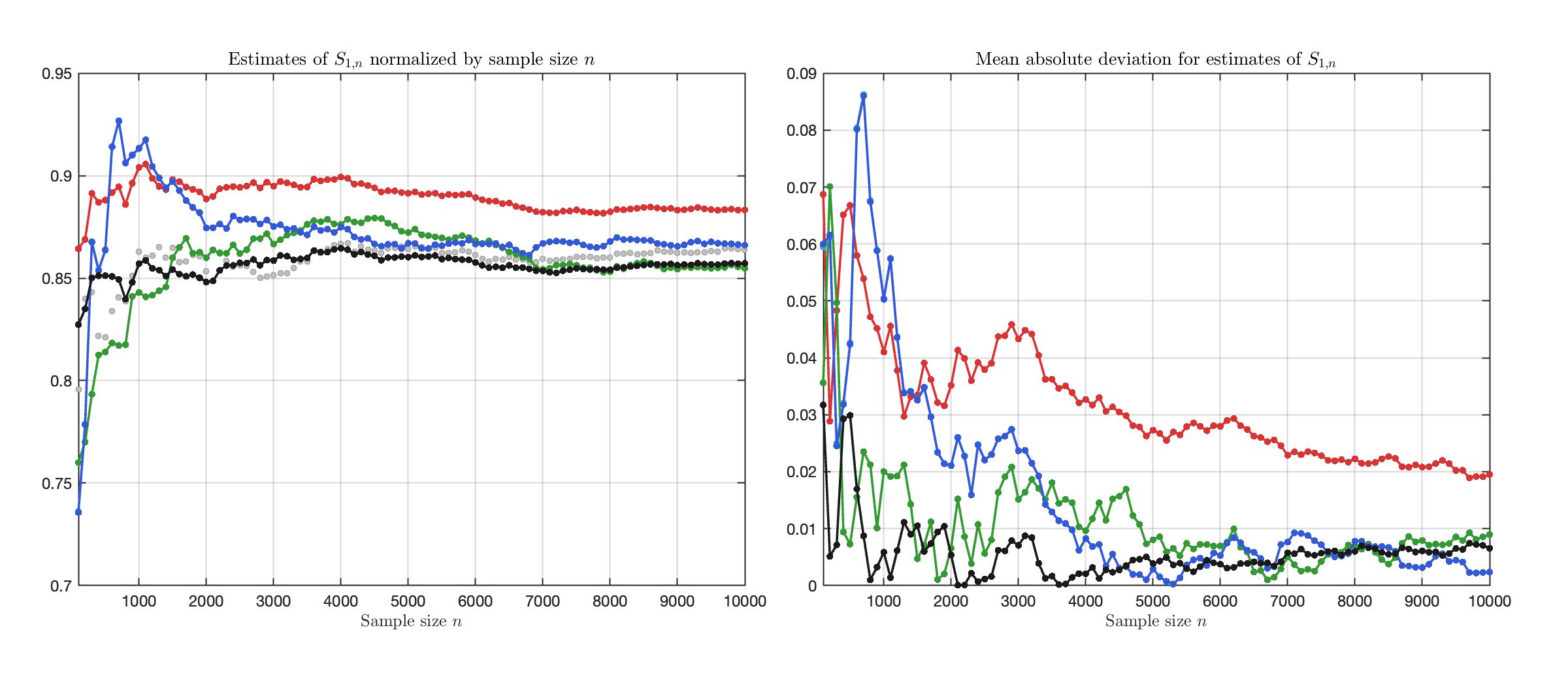}
\end{center}
\caption{\scriptsize{Square-root of half-Cauchy prior, $S_{1,n}$. Left panel: true values $n^{-1}S_{1,n}$ (Grey o-) and estimates $n^{-1}\hat{S}^{\text{\tiny{[O]}}}_{1,n}$ (Black .-), $n^{-1}\hat{S}^{\text{\tiny{[ML]}}}_{1,n}$ (Blue .-), $n^{-1}\hat{S}^{\text{\tiny{[B]}}}_{1,n}$ (Cyan .-), $n^{-1}\hat{S}^{\text{\tiny[``u,v"]}}_{1,n}$ (Green .-) and $n^{-1}\hat{S}^{\text{\tiny{[Q-B]}}}_{1,n}$ (Red .-). Right panel: MAD of $\hat{S}^{\text{\tiny{[O]}}}_{1,n}$ (Black .-), $\hat{S}^{\text{\tiny{[ML]}}}_{1,n}$ (Blue .-), $\hat{S}^{\text{\tiny{[B]}}}_{1,n}$ (Cyan .-), $\hat{S}^{\text{\tiny[``u,v"]}}_{1,n}$ (Green .-) and $\hat{S}^{\text{\tiny{[Q-B]}}}_{1,n}$ (Red .-)}}
\label{fig_hafcauchy_s1_supp}
\end{figure}

\begin{table}[ht]
\centering
\caption{Square-root of half-Cauchy prior, $S_{1,n}$: MAD as $n$ varies}
{
\setlength{\tabcolsep}{0pt}
\begin{tabular}{@{}l@{\hspace{1.2cm}}*{5}{>{\centering\arraybackslash}p{2.15cm}}@{}}
\hline
\hline
& $\hat{S}^{\text{\tiny{[O]}}}_{1,n}$ & $\hat{S}^{\text{\tiny{[ML]}}}_{1,n}$ & $\hat{S}^{\text{\tiny{[B]}}}_{1,n}$ & $\hat{S}^{\text{\tiny{[``u,v'']}}}_{1,n}$ & $\hat{S}^{\text{\tiny{[Q-B]}}}_{1,n}$ \\[0.1cm]
\hline
\multicolumn{6}{@{}l}{\underline{$n=1{,}000$}} \\[0.05cm]
MAD & 0.0059 & 0.0504 & 0.0503 & 0.0200 & \textbf{0.0410} \\[0.2cm]

\multicolumn{6}{@{}l}{\underline{$n=2{,}000$}} \\[0.05cm]
MAD & 0.0054 & 0.0211 & 0.0211 & 0.0066 & \textbf{0.0352} \\[0.2cm]

\multicolumn{6}{@{}l}{\underline{$n=3{,}000$}} \\[0.05cm]
MAD & 0.0070 & 0.0237 & 0.0237 & 0.0151 & \textbf{0.0433} \\[0.2cm]

\multicolumn{6}{@{}l}{\underline{$n=4{,}000$}} \\[0.05cm]
MAD & 0.0021 & 0.0083 & 0.0083 & 0.0096 & \textbf{0.0327} \\[0.2cm]

\multicolumn{6}{@{}l}{\underline{$n=5{,}000$}} \\[0.05cm]
MAD & 0.0038 & 0.0029 & 0.0029 & 0.0080 & \textbf{0.0273} \\[0.2cm]

\multicolumn{6}{@{}l}{\underline{$n=6{,}000$}} \\[0.05cm]
MAD & 0.0037 & 0.0053 & 0.0053 & 0.0070 & \textbf{0.0280} \\[0.2cm]

\multicolumn{6}{@{}l}{\underline{$n=7{,}000$}} \\[0.05cm]
MAD & 0.0057 & 0.0077 & 0.0076 & 0.0050 & \textbf{0.0229} \\[0.2cm]

\multicolumn{6}{@{}l}{\underline{$n=8{,}000$}} \\[0.05cm]
MAD & 0.0060 & 0.0078 & 0.0078 & 0.0070 & \textbf{0.0223} \\[0.2cm]

\multicolumn{6}{@{}l}{\underline{$n=9{,}000$}} \\[0.05cm]
MAD & 0.0061 & 0.0032 & 0.0031 & 0.0079 & \textbf{0.0208} \\[0.2cm]

\multicolumn{6}{@{}l}{\underline{$n=10{,}000$}} \\[0.05cm]
MAD & 0.0065 & 0.0024 & 0.0023 & 0.0090 & \textbf{0.0195} \\[0.1cm]
\hline
\hline
\end{tabular}
}
\label{tab_hafcauchy_s1_supp}
\end{table}

Figure~\ref{fig_hafcauchy_s3_supp} and Table~\ref{tab_hafcauchy_s3_supp} report corresponding results with respect to the estimation of $S_{3,n}$.

\begin{figure}[h!]
\begin{center}
\includegraphics[width=1\linewidth,height=0.36\textheight,keepaspectratio]{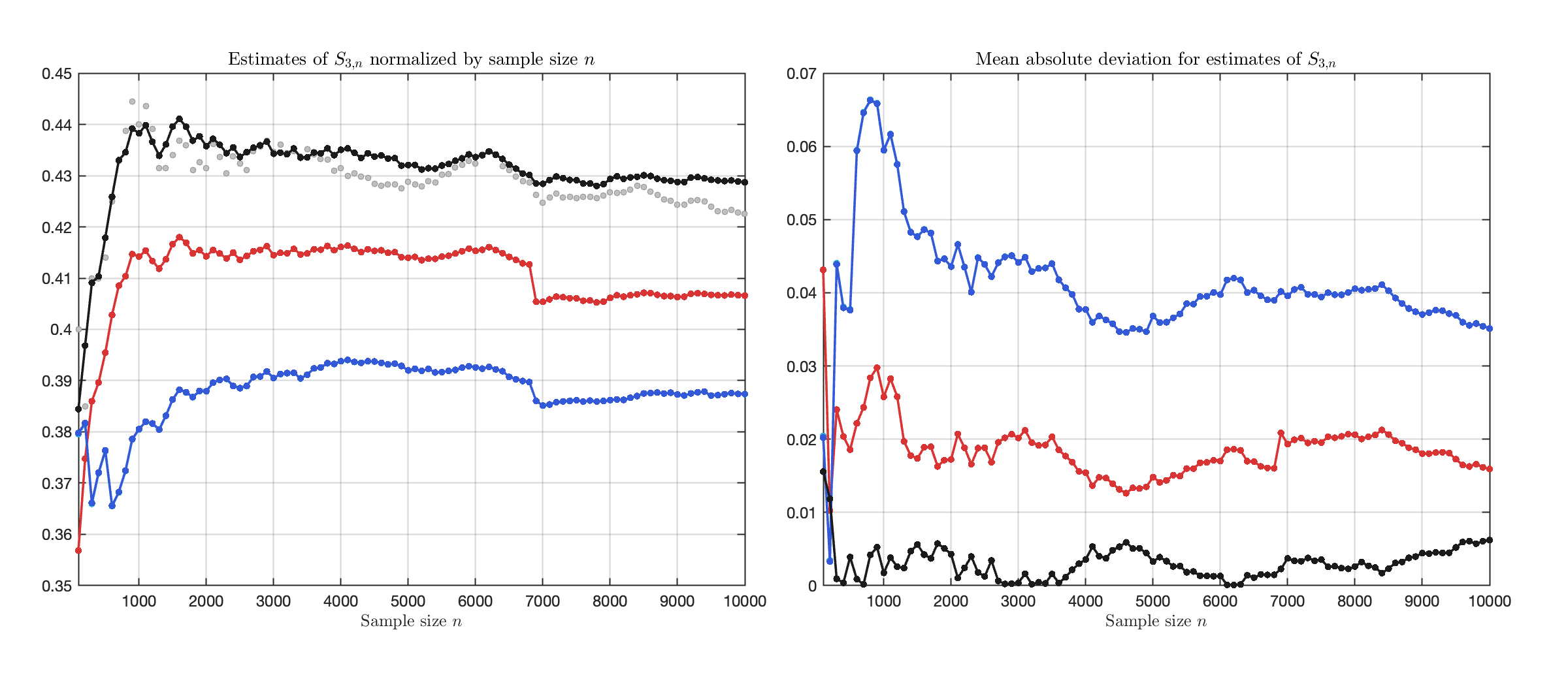}
\end{center}
\caption{\scriptsize{Square-root of half-Cauchy prior, $S_{3,n}$. Left panel: true values of $n^{-1}S_{3,n}$ (Grey o-) and estimates $n^{-1}\hat{S}^{\text{\tiny{[O]}}}_{3,n}$ (Black .-), $n^{-1}\hat{S}^{\text{\tiny{[ML]}}}_{3,n}$ (Blue .-), $n^{-1}\hat{S}^{\text{\tiny{[B]}}}_{3,n}$ (Cyan .-)  and $n^{-1}\hat{S}^{\text{\tiny{[Q-B]}}}_{3,n}$ (Red .-). Right panel: MAD of $\hat{S}^{\text{\tiny{[O]}}}_{3,n}$ (Black .-), $\hat{S}^{\text{\tiny{[ML]}}}_{3,n}$ (Blue .-), $\hat{S}^{\text{\tiny{[B]}}}_{3,n}$ (Cyan .-) and $\hat{S}^{\text{\tiny{[Q-B]}}}_{3,n}$ (Red .-)}}
\label{fig_hafcauchy_s3_supp}
\end{figure}

\begin{table}[ht]
\centering
\caption{Square-root of half-Cauchy prior, $S_{3,n}$: MAD as $n$ varies}
{
\setlength{\tabcolsep}{0pt}
\begin{tabular}{@{}l@{\hspace{1.2cm}}*{4}{>{\centering\arraybackslash}p{2.15cm}}@{}}
\hline
\hline
& $\hat{S}^{\text{\tiny{[O]}}}_{3,n}$ & $\hat{S}^{\text{\tiny{[ML]}}}_{3,n}$ & $\hat{S}^{\text{\tiny{[B]}}}_{3,n}$ & $\hat{S}^{\text{\tiny{[Q-B]}}}_{3,n}$ \\[0.1cm]
\hline
\multicolumn{5}{@{}l}{\underline{$n=1{,}000$}} \\[0.05cm]
MAD & 0.0017 & 0.0595 & 0.0595 & \textbf{0.0258} \\[0.2cm]

\multicolumn{5}{@{}l}{\underline{$n=2{,}000$}} \\[0.05cm]
MAD & 0.0043 & 0.0436 & 0.0436 & \textbf{0.0172} \\[0.2cm]

\multicolumn{5}{@{}l}{\underline{$n=3{,}000$}} \\[0.05cm]
MAD & 0.0004 & 0.0442 & 0.0441 & \textbf{0.0202} \\[0.2cm]

\multicolumn{5}{@{}l}{\underline{$n=4{,}000$}} \\[0.05cm]
MAD & 0.0036 & 0.0377 & 0.0377 & \textbf{0.0154} \\[0.2cm]

\multicolumn{5}{@{}l}{\underline{$n=5{,}000$}} \\[0.05cm]
MAD & 0.0033 & 0.0368 & 0.0368 & \textbf{0.0148} \\[0.2cm]

\multicolumn{5}{@{}l}{\underline{$n=6{,}000$}} \\[0.05cm]
MAD & 0.0013 & 0.0398 & 0.0397 & \textbf{0.0170} \\[0.2cm]

\multicolumn{5}{@{}l}{\underline{$n=7{,}000$}} \\[0.05cm]
MAD & 0.0037 & 0.0396 & 0.0396 & \textbf{0.0087} \\[0.2cm]

\multicolumn{5}{@{}l}{\underline{$n=8{,}000$}} \\[0.05cm]
MAD & 0.0026 & 0.0406 & 0.0406 & \textbf{0.0099} \\[0.2cm]

\multicolumn{5}{@{}l}{\underline{$n=9{,}000$}} \\[0.05cm]
MAD & 0.0044 & 0.0370 & 0.0370 & \textbf{0.0074} \\[0.2cm]

\multicolumn{5}{@{}l}{\underline{$n=10{,}000$}} \\[0.05cm]
MAD & 0.0062 & 0.0351 & 0.0351 & \textbf{0.0053} \\[0.1cm]
\hline
\hline
\end{tabular}
}
\label{tab_hafcauchy_s3_supp}
\end{table}

To conclude, we provide quasi-Bayes credible intervals at level $1-\alpha=0.95$. Specifically, credible intervals are constructed from \eqref{cred_interval} by relying on Newton's algorithm initialized as in Table~\ref{tab_hafcauchy_s1_supp}. Figure~\ref{halfcauchy_fig_interval_s1}-\ref{halfcauchy_fig_interval_s3} display the quasi-Bayes credible intervals for $S_{1,n}$ and $S_{3,n}$, respectively, with the corresponding oracle credible intervals $I^\ast_{S_{1,n}}$ and $I^\ast_{S_{3,n}}$ under the square-root of half-Cauchy prior.

\begin{figure}[h!]
\begin{center}
\includegraphics[width=1\linewidth,height=0.36\textheight,keepaspectratio]{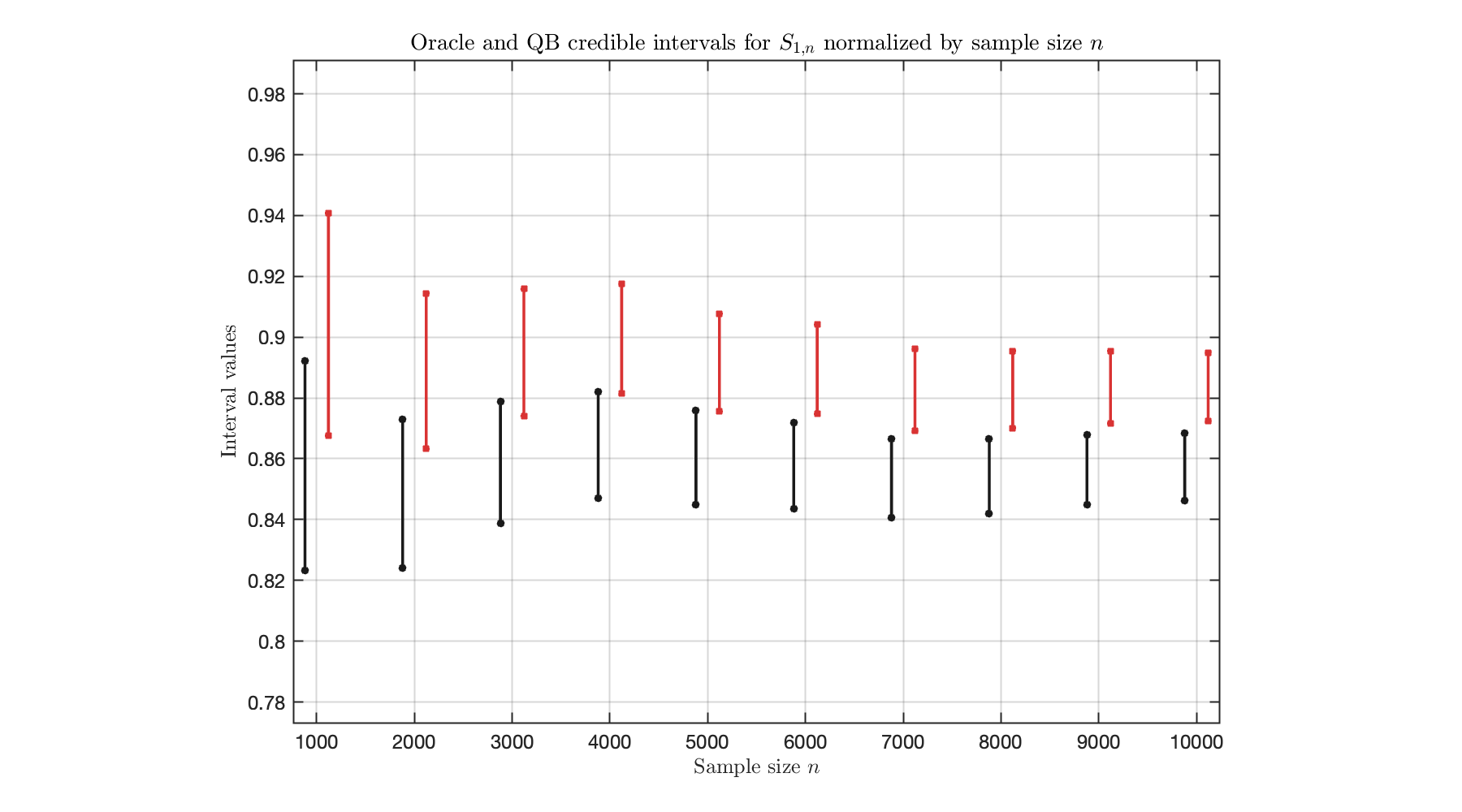}
\end{center}
\caption{\scriptsize{Square-root of half-Cauchy prior, $S_{1,n}$: oracle credible intervals (black) and quasi-Bayes credible intervals (red)}}
\label{halfcauchy_fig_interval_s1}
\end{figure}

\begin{figure}[h!]
\begin{center}
\includegraphics[width=1\linewidth,height=0.36\textheight,keepaspectratio]{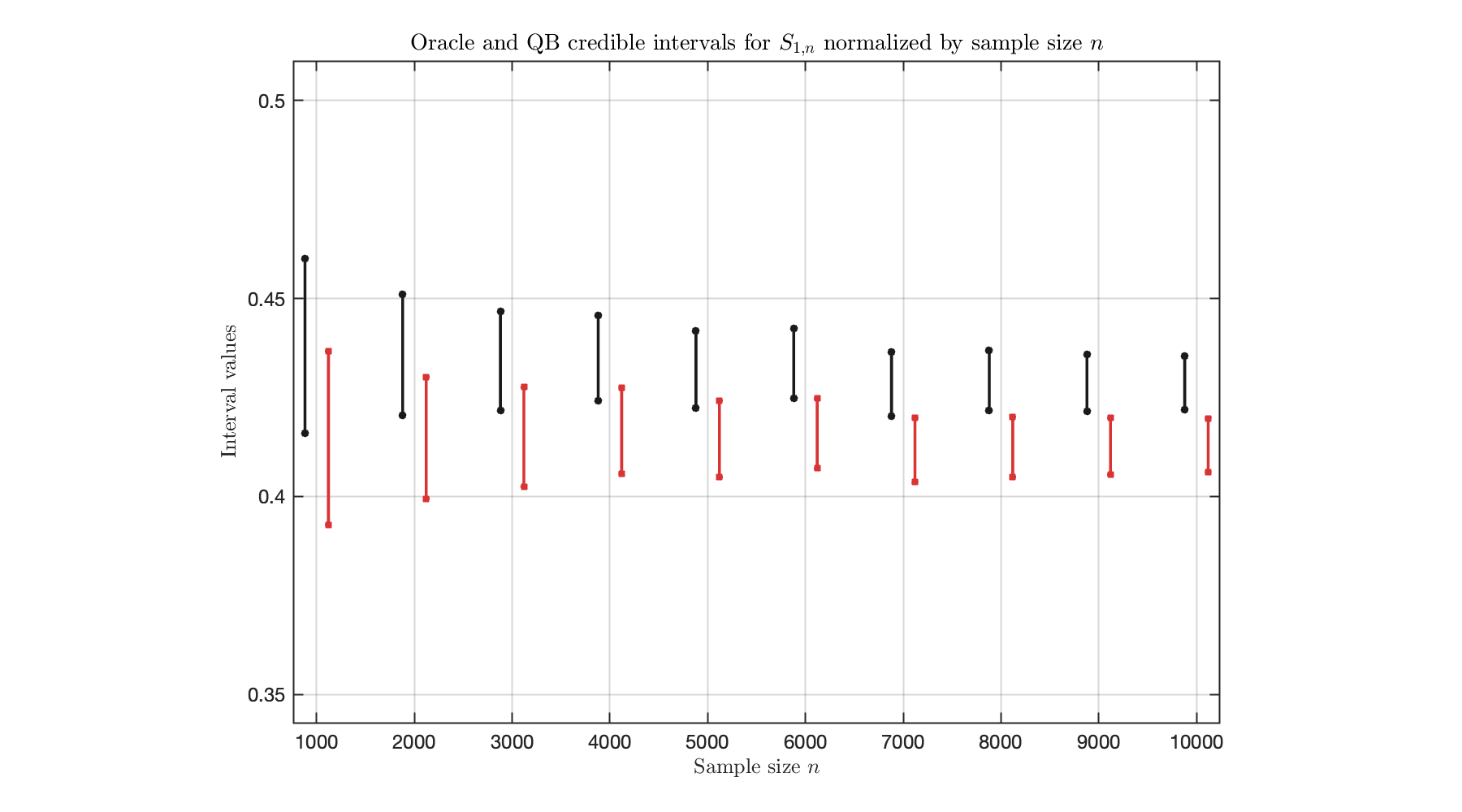}
\end{center}
\caption{\scriptsize{Square-root of half-Cauchy prior, $S_{3,n}$: oracle credible intervals (black) and quasi-Bayes credible intervals (red)}}
\label{halfcauchy_fig_interval_s3}
\end{figure}


\section{Additional real-data experiments}\label{supp_real}

\subsection{European automobile insurance data}

We apply the quasi-Bayes EB approach to the benchmark automobile insurance claims dataset \citep[Table 6.1]{Efr(21)}. The data consist of one year of claims from a European automobile insurance company on $n=9,461$ policyholders. From the first two rows of Table~\ref{efron_tab}, $7840$ of the $n_{0}=9461$ policy holders made no claims ($y=0$) during the year, $n_{1}=1317$ made a single claim ($y=1$), $n_{2}=239$ made two claims $y=2$ each, etc., continuing to the one person who made seven claims. In the EB literature, the European automobile insurance dataset is commonly analyzed under a Poisson mixture model \eqref{model_distr_poi}, where the number of claims $X_i$ for policyholder $i$ is modeled conditionally on an individual risk parameter $\theta_i$, with $\theta_i$ drawn from an unknown mixing distribution $G$.  Table~\ref{efron_tab} reports estimates of $S_{1,n}$ (with $\kappa=2$) and $S_{3,n}$, obtained by applying the same procedures described in the synthetic experiments: i) parametric EB estimates (for $S_{1,n}$ and $S_{3,n}$); ii) the nonparametric ``$u,v$" estimate (only for $S_{1,n}$); the quasi-Bayes EB estimate (for $S_{1,n}$ and $S_{3,n}$).

\begin{table}[ht]
\centering
\caption{Counts $n_{y}$ of number of claims $y$ made in a single year, for $y=0,1,\ldots,7$}
\begin{tabular}[t]{lccccccccccc}
\hline
$y$\hspace{0.5cm} &&& 0\hspace{0.75cm} & 1\hspace{0.75cm} & 2\hspace{0.75cm} & 3\hspace{0.75cm} & 4\hspace{0.75cm} & 5\hspace{0.75cm} & 6\hspace{0.75cm} & 7\hspace{0.75cm}   \\[0.1cm]
$n_{y}$\hspace{0.5cm} &&& 7,840\hspace{0.75cm} & 1,317\hspace{0.75cm} & 239\hspace{0.75cm} & 42\hspace{0.75cm} & 14\hspace{0.75cm} & 4\hspace{0.75cm} & 4\hspace{0.75cm} & 1\\[0.1cm]
\hline 
\hline\\[0.05cm]
\end{tabular}
\label{efron_tab}
\end{table}

\begin{table}[ht]
\centering
\caption{European automobile insurance data: estimates of $S_{1,n}$}
\begin{tabular}[t]{lccccc}
\hline
\hline
$\hat{S}^{\text{\tiny{[ML]}}}_{1,n}$ &\hspace{1cm}$\hat{S}^{\text{\tiny{[B]}}}_{1,n}$&\hspace{1cm}$\hat{S}^{\text{\tiny[``u,v"]}}_{1,n}$  &&\hspace{1cm} $\hat{S}^{\text{\tiny{[Q-B]}}}_{1,n}$ \\[0.2cm]
\hline
976.03&\hspace{1cm}976.03  &\hspace{1cm}  1,721&&\hspace{1cm} \textbf{1,361.44}\\
\hline
\hline
\end{tabular}
\label{tab_efron_s1}
\end{table}%

\begin{table}[ht]
\centering
\caption{European automobile insurance data: estimates of $S_{3,n}$}
\begin{tabular}[t]{lcccc}
\hline
\hline
$\hat{S}^{\text{\tiny{[ML]}}}_{3,n}$ &\hspace{1cm}$\hat{S}^{\text{\tiny{[B]}}}_{3,n}$ &&\hspace{1cm} $\hat{S}^{\text{\tiny{[Q-B]}}}_{3,n}$ \\[0.1cm]
\hline
1,890.32&\hspace{1cm}1,890.32  &&\hspace{1cm} \textbf{1,588.14}\\
\hline
\hline
\end{tabular}
\label{tab_efron_s3}
\end{table}%

\subsection{Twitter data}\label{app5_num2}

We apply the quasi-Bayes EB approach to a dataset of tweets that received at least $50$ retweets on Twitter between October 7 and November 6, 2011, available at \url{https://snap.stanford.edu/seismic/}. The dataset contains $166{,}783$ tweets and $34{,}617{,}704$ retweets, together with metadata such as tweet ID, posting time, retweet time, and follower counts. For each tweet, we focus on the number of early retweets, namely retweets received within the first $30$ seconds after posting, since unusually rapid retweeting may indicate automated or bot-driven activity. From the first two rows of Table~\ref{tweet_tab},  $n_{0}=40{,}259$ tweets received no early retweet $(y=0)$, $n_{1}=28{,}339$ received one early retweet ($y=1$), $n_{2}=21{,}581$ received two early retweets $y=2$, and so on. We analyze these data under the Poisson mixture model \eqref{model_distr_poi}, where the number of early retweets $X_{i}$ for tweet $i$ is modeled conditionally on an individual intensity parameter $\theta_i$, with $\theta_{i}$ drawn from an unknown mixing distribution $G$. Here, $\theta_i$ can be interpreted as a latent measure of instantaneous virality. Identifying tweets with unusually high intensity may help detect automated activity or early virality. Table~\ref{tweet_tab} reports estimates of $S_{1,n}$ (with $\kappa=2$) and $S_{3,n}$, obtained by applying the same procedures described in the synthetic experiments: i) parametric EB estimates (for $S_{1,n}$ and $S_{3,n}$); ii) the nonparametric ``$u,v$" estimate (only for $S_{1,n}$); the quasi-Bayes EB estimate (for $S_{1,n}$ and $S_{3,n}$).

\begin{table}[ht]
\centering
\caption{Counts $n_{y}$ of the number of tweets that have $y$ retweets within $30$ seconds from the tweet post time, for $y=0,1,\ldots, 77$, and corresponding estimates of the number of retweets expected}
\begin{tabular}[t]{lccccccccccc}
\hline
\hline
$y$\hspace{0.5cm} &&& 0 \hspace{0.75cm}& 1\hspace{0.75cm} & 2\hspace{0.75cm} & 3\hspace{0.75cm} & 4\hspace{0.75cm} & 5\hspace{0.75cm} & 6\hspace{0.5cm} & 7   \\[0.1cm]
$n_{y}$\hspace{0.5cm} &&& 40,259\hspace{0.75cm} & 28,339\hspace{0.75cm} & 21,581\hspace{0.75cm} & 16,479\hspace{0.75cm} & 12,130\hspace{0.75cm} & 9,238\hspace{0.75cm} & 7,193\hspace{0.75cm} & 5,464\\[0.1cm]
\hline 
\hline\\[0.05cm]
\end{tabular}
\begin{tabular}[t]{lccccccccccc}
\hline
$y$\hspace{0.5cm} &&& 8\hspace{0.75cm} & 11\hspace{0.75cm} & 14\hspace{0.75cm} & 17\hspace{0.75cm} & 20\hspace{0.75cm} & 23\hspace{0.75cm} & 26\hspace{0.75cm} & 29\hspace{0.75cm}   \\[0.1cm]
$n_{y}$\hspace{0.5cm} &&& 4,319\hspace{0.75cm} & 2,291\hspace{0.75cm} & 1,352\hspace{0.75cm} & 817\hspace{0.75cm} & 475\hspace{0.75cm} & 293\hspace{0.75cm} & 188\hspace{0.75cm} & 129\\[0.1cm]
\hline 
\hline\\[0.05cm]
\end{tabular}
\begin{tabular}[t]{lccccccccccc}
\hline
$y$\hspace{0.5cm} &&& 32\hspace{0.75cm} & 35\hspace{0.75cm} & 38\hspace{0.75cm} & 41\hspace{0.75cm} & 44\hspace{0.75cm} & 47\hspace{0.75cm} & 50\hspace{0.75cm} & 53   \\[0.1cm]
$n_{y}$\hspace{0.5cm} &&& 89\hspace{0.75cm} & 69\hspace{0.75cm} & 60\hspace{0.75cm} & 61\hspace{0.75cm} & 36\hspace{0.75cm} & 36\hspace{0.75cm} & 32\hspace{0.75cm} & 21\\[0.1cm]
\hline 
\hline\\[0.05cm]
\end{tabular}
\begin{tabular}[t]{lccccccccccc}
\hline
$y$\hspace{0.5cm} &&& 56\hspace{0.75cm} & 59\hspace{0.75cm} & 62\hspace{0.75cm} & 65\hspace{0.75cm} & 68\hspace{0.75cm} & 71\hspace{0.75cm} & 74\hspace{0.75cm} & 77   \\[0.1cm]
$n_{y}$\hspace{0.5cm} &&& 24\hspace{0.75cm} & 16\hspace{0.75cm} & 14\hspace{0.75cm} & 13\hspace{0.75cm} & 14\hspace{0.75cm} & 11\hspace{0.75cm} & 2\hspace{0.75cm} & 5\\[0.1cm]
\hline 
\hline
\end{tabular}
\label{tweet_tab}
\end{table}

\begin{table}[ht]
\centering
\caption{Twitter data: estimates of $S_{1,n}$}
\begin{tabular}[t]{lccccc}
\hline
\hline
$\hat{S}^{\text{\tiny{[ML]}}}_{1,n}$ &\hspace{1cm}$\hat{S}^{\text{\tiny{[B]}}}_{1,n}$&\hspace{1cm}$\hat{S}^{\text{\tiny[``u,v"]}}_{1,n}$  &&\hspace{1cm} $\hat{S}^{\text{\tiny{[Q-B]}}}_{1,n}$ \\[0.2cm]
\hline
70,214.73&\hspace{1cm}70,214.73  &\hspace{1cm}  120,938&&\hspace{1cm} \textbf{117,942}\\
\hline
\hline
\end{tabular}
\label{tab_tweet_s1}
\end{table}%

\begin{table}[ht]
\centering
\caption{Twitter data: estimates of $S_{3,n}$}
\begin{tabular}[t]{lcccc}
\hline
\hline
$\hat{S}^{\text{\tiny{[ML]}}}_{3,n}$ &\hspace{1cm}$\hat{S}^{\text{\tiny{[B]}}}_{3,n}$ &&\hspace{1cm} $\hat{S}^{\text{\tiny{[Q-B]}}}_{3,n}$ \\[0.2cm]
\hline
77,426.24&\hspace{1cm}77,426.24  &&\hspace{1cm} \textbf{65,998.14}\\
\hline
\hline
\end{tabular}
\label{tab_tweet_s3}
\end{table}%


\section{Examples: Gaussian kernel}\label{app_gauss}

For the Gaussian mixture model, we present examples of $S_{n}$ drawn from \citet[Section 2.4]{Zha(05)}. Let $n\geq1$ and consider the real-valued random vectors $(X_{1},\theta_{1}),\ldots,(X_{n},\theta_{n})$ distributed as follows:
\begin{align}\label{model_distr_gaussian}
X_i\mid\theta_{i} & \quad\simind\quad \text{Gaussian}(\cdot\mid\theta_{i},\,1)\qquad i=1,\ldots,n\\[-0.2cm]
\notag\theta_{i}& \quad\simiid\quad G,
\end{align}
where $\text{Gaussian}(\cdot\mid\theta,\,1)$ is the Gaussian kernel of mean $\theta$ and variance $1$, and $G$ is the unknown mixing distribution on $\Theta\subset\mathbb{R}$. Given the observed $X_{i}$'s, we address the estimation of $S_{1,n}=\sum_{1\leq i\leq n}\theta_{i}I(X_{i}\leq\kappa)$ and $S_{3,n}=\sum_{1\leq i\leq n}I(X_{i}>\theta_{i})$, where $\kappa\in\mathbb{R}$ is a given threshold.

\subsection{Parametric $g$-modeling and ``$u,v$" method}

Let $G$ in \eqref{model_distr_gaussian} be a Gaussian distribution with mean $\mu$ and variance $\sigma^{2}$. Under this assumption, we apply parametric $g$-modeling to obtain EB and Bayes EB estimates of $S_{r,n}$, $r=1,3$.

\begin{prp}\label{prop:gaussian}
Let $n\geq1$ and let $(X_{1},\theta_{1}),\ldots,(X_{n},\theta_{n})$ be random vectors as in \eqref{model_distr_gaussian}, with $G$ be a Gaussian distribution with mean $\mu$ and variance $\sigma^{2}$, here denote by $G_{\mu,\sigma^{2}}$. Then
\begin{equation}\label{s1_gauss}
\hat{S}_{1,n}(G_{\mu,\sigma^{2}})=\E_{G_{\mu,\sigma^{2}}}[S_{1,n}\mid X_{1:n}]=\frac{1}{n}\sum_{i=1}^{n}\frac{\mu+\sigma^{2}X_{i}}{1+\sigma^{2}}I(X_{i}\leq\kappa)
\end{equation}
and
\begin{equation}\label{s3_gauss}
\hat{S}_{3,n}(G_{\mu,\sigma^{2}})=\E_{G_{\mu,\sigma^{2}}}[S_{3,n}\mid X_{1:n}]=\frac{1}{2}\text{Erfc}\left(\frac{\mu-X_{i}}{\sqrt{2\sigma^{2}(1+\sigma^{2})}}\right),
\end{equation}
where $\text{Erfc}(x)=(2\pi)^{-1/2}\int_{(x,+\infty)}\text{e}^{-y^{2}}\ddr y$. EB estimates of $S_{1,n}$ and $S_{3,n}$ are obtained as plug-in estimates by replacing $\mu$ and $\sigma$ in \eqref{s1_gauss} and \eqref{s3_gauss} with the maximum likelihood estimate
\begin{displaymath}
\hat{\mu}^{\text{\tiny{[ML]}}}_{n}=\frac{1}{n}\sum_{i=1}^{n}X_{i}
\end{displaymath}
and
\begin{displaymath}
\hat{\sigma}^{\text{\tiny{[ML]}}}_{n}=\sqrt{\frac{1}{n}\sum_{i=1}^{n}\left(X_{i}-\frac{1}{n}\sum_{i=1}^{n}X_{i}\right)^{2}-1}.
\end{displaymath}
By placing on $(1+\sigma^{2})^{-1}$ a Gamma prior with (hyper)-parameter $(a_{\sigma^{2}},b_{\sigma^{2}})$ and on $\mu$ a Gaussian prior with (hyper)-parameter $(a_{\mu},((1+\sigma^{2})^{-1}b_{\mu})^{-1/2})$, the Bayes EB estimates of $S_{1,n}$ and $S_{3,n}$ are obtained as plug-in estimates by replacing $\mu$ and $\sigma$ in \eqref{s1_gauss} and \eqref{s3_gauss} with the Bayesian estimate
\begin{displaymath}
\hat{\mu}^{\text{\tiny{[B]}}}_{n}=\frac{a_{\mu}b_{\mu}+n\hat{\mu}^{\text{\tiny{[ML]}}}_{n}}{n+b_{\mu}}
\end{displaymath}
and
\begin{displaymath}
\hat{\sigma}^{\text{\tiny{[B]}}}_{n}=\sqrt{\frac{b_{\sigma^{2}}+2^{-1}\sum_{i=1}^{n}(X_{i}-\hat{\mu}^{\text{\tiny{[ML]}}}_{n})^{2}+\frac{nb_{\mu}}{2(n+b_{\mu})}(\hat{\mu}^{\text{\tiny{[ML]}}}_{n}-a_{\mu})^{2}}{a_{\sigma^{2}}+n/2}-1}.
\end{displaymath}
The resulting EB and Bayes EB estimates of $S_{r,n}$, for $r=1,3$, here denoted by $\hat{S}^{\text{\tiny{[ML]}}}_{r,n}$ and $\hat{S}^{\text{\tiny{[B]}}}_{r,n}$ respectively, are asymptotically efficient \citep[Theorem 2.2]{Zha(05)}.
\end{prp}

Now, we consider $G$ in \eqref{model_distr_gaussian} completely unknown, namely we do not specify any parametric assumption. In this nonparametric setting, the ``$u,v$" method of \citet{Rob(88)} estimates $S_{n}$ in \eqref{eq:sum} with $\hat{S}^{\text{\tiny[``u,v"]}}_{n}=\sum_{1\leq i\leq n}v(X_{i})$ if there exists a function $v$ that solve the ``$u,v$" equation
\begin{equation}\label{int_eq2}
\int_{-\infty}^{+\infty}(v(x)-u(x,\theta))\frac{1}{\sqrt{2\pi}}\text{e}^{-\frac{1}{2}(x-\theta)^{2}}\ddr x=0\qquad\forall\, \theta\in\Theta.
\end{equation}
For the quantity $S_{3,n}$, Equation \eqref{int_eq2} has the unique solution $v(x)=2^{-1}$. Hence, the ``$u,v$" estimate of $S_{3,n}$ is
\begin{equation}\label{uv_gauss}
\hat{S}^{\text{\tiny[``u,v"]}}_{3,n}=\frac{n}{2}.
\end{equation}
The estimate $\hat{S}^{\text{\tiny[``u,v"]}}_{3,n}$ is asymptotically efficient, as $n\rightarrow+\infty$ \citep[Theorem 2.5]{Zha(05)}. With regards to $S_{1,n}$, the ``$u,v$" equation \eqref{int_eq2} does not admit an explicit solution $v$, which makes not possible the derivation of a corresponding ``$u,v$" estimates \citep[Section 2.4]{Zha(05)}.


\subsection{Synthetic-data analysis}

For $i=1,\ldots,100$, let $\mathbf{X}_{i}=X_{1:100 i}$ denote a dataset of size $n=100i$ generated from the Gaussian mixture model \eqref{model_distr_gaussian}, with a Gaussian prior $G$ with mean $2$ and variance $1$. The $\mathbf{X}_{i}$'s are nested, so that the sample size increases progressively: at stage $i$, we have $n=100i$ data, obtained by adding $100$ new data at each step, starting from $n=100$. For each dataset $\mathbf{X}_{i}$ of size $n=100i$, we apply the quasi-Bayes EB approach to estimate $S_{1,n}$ and $S_{2,n}$, with $\kappa=2$, and $S_{3,n}$.

The implementation of Newton's algorithm \eqref{eq:newton} requires the numerical evaluation of an integral, i.e. the marginal likelihood, which we approximate via the trapezoidal rule. To this end, the density function of $G_{n}$ is represented through its values on a fixed grid of $d$ quadrature points over $\Theta$, with $d$ controlling the numerical resolution of the integration. This representation is used solely for numerical evaluation and does not impose any modeling restriction on $\Theta$. For a dataset $X_{1:n}$ of size $n$, letting $L_\Theta=\min\{\min\{X_{1:n}\},\lfloor Q_{n,0.01}-4\rfloor\}$ and $U_\Theta=\max\{\max\{X_{1:n}\},\lceil Q_{n,0.99}+4\rceil\}$, with $Q_{n,\delta}=\text{Quantile}(X_{1:n};\delta)$, we employ a uniform grid with $d=1{,}000$ over $\Theta=(L_{\Theta},U_\Theta)$, set $G_0$ to be Uniform on $\Theta$, and set $\alpha_n=(1+n)^{-0.99}$. Under this setting for Newton's algorithm, we obtain an estimate $G_n$ of the mixing distribution $G$, which, when substituted into \eqref{eq:est}, gives the quasi-Bayes EB estimates $\hat{S}^{\text{\tiny{[Q-B]}}}_{r,n}$ of $S_{r,n}$, for $r=1,2,3$.

The estimates $\hat{S}^{\text{\tiny{[Q-B]}}}_{r,n}$, normalized by the sample size $n$, are reported in Figure~\ref{fig_gauss_s1}-\ref{fig_gauss_s3}, together with their mean absolute deviations (MAD) from the true values of $S_{r,n}$. See also Table~\ref{tab_gauss_s1}-\ref{tab_gauss_s3}. In this analysis, the estimate $\hat{S}^{\text{\tiny{[Q-B]}}}_{r,n}$ is compared with: i) the EB estimate $\hat{S}^{\text{\tiny{[ML]}}}_{r,n}$ and Bayes EB estimate $\hat{S}^{\text{\tiny{[B]}}}_{r,n}$ from Proposition \ref{prop:gaussian}; ii) the ``$u,v$" estimate $\hat{S}^{\text{\tiny[``u,v"]}}_{r,n}$ in \eqref{uv_poiss}, which is available only for $S_{3,n}$. We also report the oracle Bayes estimate of $S_{r,n}$, namely $\hat{S}^{\text{\tiny{[O]}}}_{r,n}=\E_{G}[S_{r,n}\,|\,X_{1:n}]$ with $G$ being the Gaussian distribution with mean $2$ and variance $1$. 

\begin{figure}[h!]
\begin{center}
\includegraphics[width=1\linewidth,height=0.36\textheight,keepaspectratio]{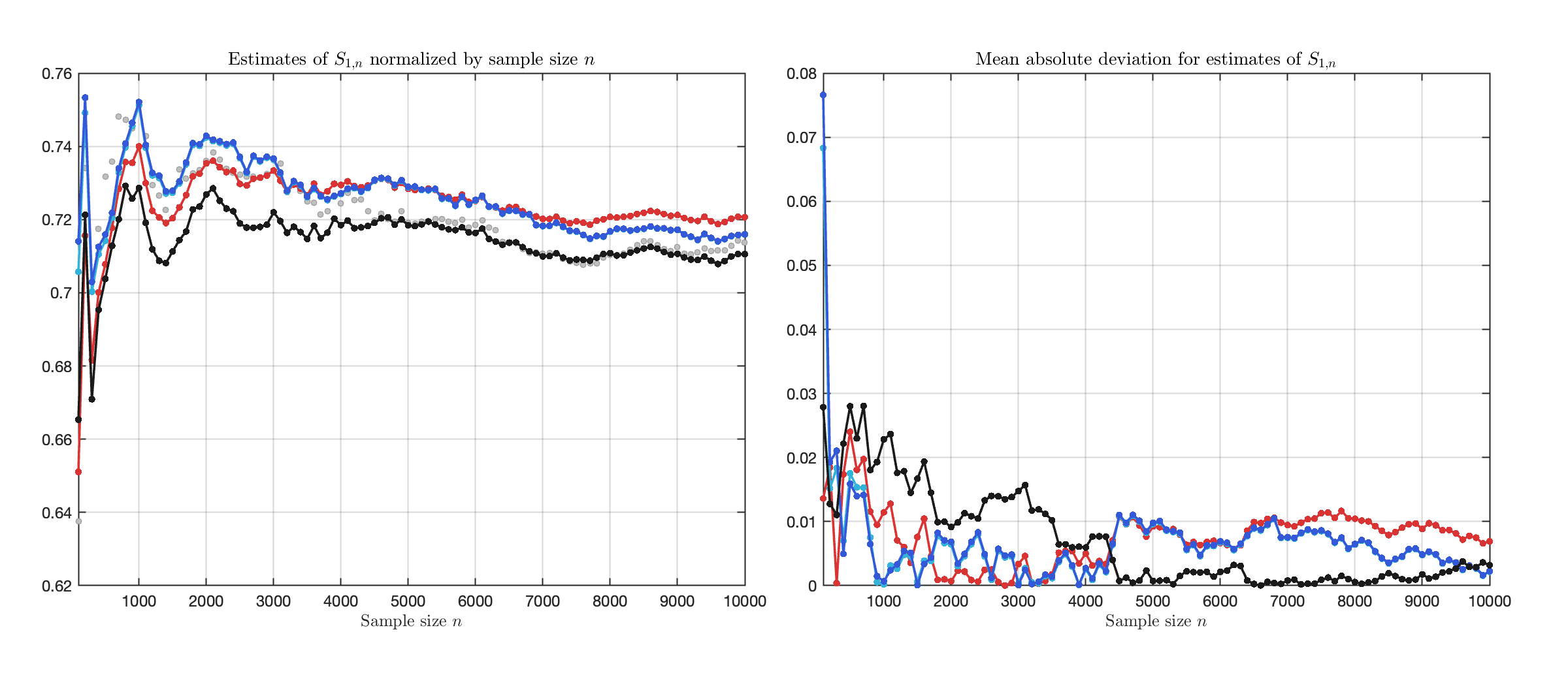}
\end{center}
\caption{\scriptsize{Gaussian prior, $S_{1,n}$. Left panel: true values of $n^{-1}S_{1,n}$ (Grey o-) and estimates $n^{-1}\hat{S}^{\text{\tiny{[O]}}}_{1,n}$ (Black .-), $n^{-1}\hat{S}^{\text{\tiny{[ML]}}}_{1,n}$ (Blue .-), $n^{-1}\hat{S}^{\text{\tiny{[B]}}}_{1,n}$ (Cyan .-) and $n^{-1}\hat{S}^{\text{\tiny{[Q-B]}}}_{1,n}$ (Red .-). Right panel: MAD of $\hat{S}^{\text{\tiny{[O]}}}_{1,n}$ (Black .-), $\hat{S}^{\text{\tiny{[ML]}}}_{1,n}$ (Blue .-), $\hat{S}^{\text{\tiny{[B]}}}_{1,n}$ (Cyan .-) and $\hat{S}^{\text{\tiny{[Q-B]}}}_{1,n}$ (Red .-)}}
\label{fig_gauss_s1}
\end{figure}

\begin{table}[ht]
\centering
\caption{Gaussian prior, $S_{1,n}$: MAD as $n$ varies}
{
\setlength{\tabcolsep}{0pt}
\begin{tabular}{@{}l@{\hspace{1.2cm}}*{4}{>{\centering\arraybackslash}p{2.15cm}}@{}}
\hline
\hline
& $\hat{S}^{\text{\tiny{[O]}}}_{1,n}$ & $\hat{S}^{\text{\tiny{[ML]}}}_{1,n}$ & $\hat{S}^{\text{\tiny{[B]}}}_{1,n}$ & $\hat{S}^{\text{\tiny{[Q-B]}}}_{1,n}$ \\[0.1cm]
\hline
\multicolumn{5}{@{}l}{\underline{$n=1{,}000$}} \\[0.05cm]
MAD & 0.0228 & 0.0002 & 0.0007 & \textbf{0.0114} \\[0.2cm]

\multicolumn{5}{@{}l}{\underline{$n=3{,}000$}} \\[0.05cm]
MAD & 0.0147 & 0.0003 & 0.0001 & \textbf{0.0033} \\[0.2cm]

\multicolumn{5}{@{}l}{\underline{$n=5{,}000$}} \\[0.05cm]
MAD & 0.0007 & 0.0096 & 0.0098 & \textbf{0.0092} \\[0.2cm]

\multicolumn{5}{@{}l}{\underline{$n=7{,}000$}} \\[0.05cm]
MAD & 0.0008 & 0.0074 & 0.0076 & \textbf{0.0094} \\[0.2cm]

\multicolumn{5}{@{}l}{\underline{$n=9{,}000$}} \\[0.05cm]
MAD & 0.0017 & 0.0047 & 0.0048 & \textbf{0.0088} \\[0.1cm]
\hline
\hline
\end{tabular}
}
\label{tab_gauss_s1}
\end{table}

\begin{figure}[h!]
\begin{center}
\includegraphics[width=1\linewidth,height=0.36\textheight,keepaspectratio]{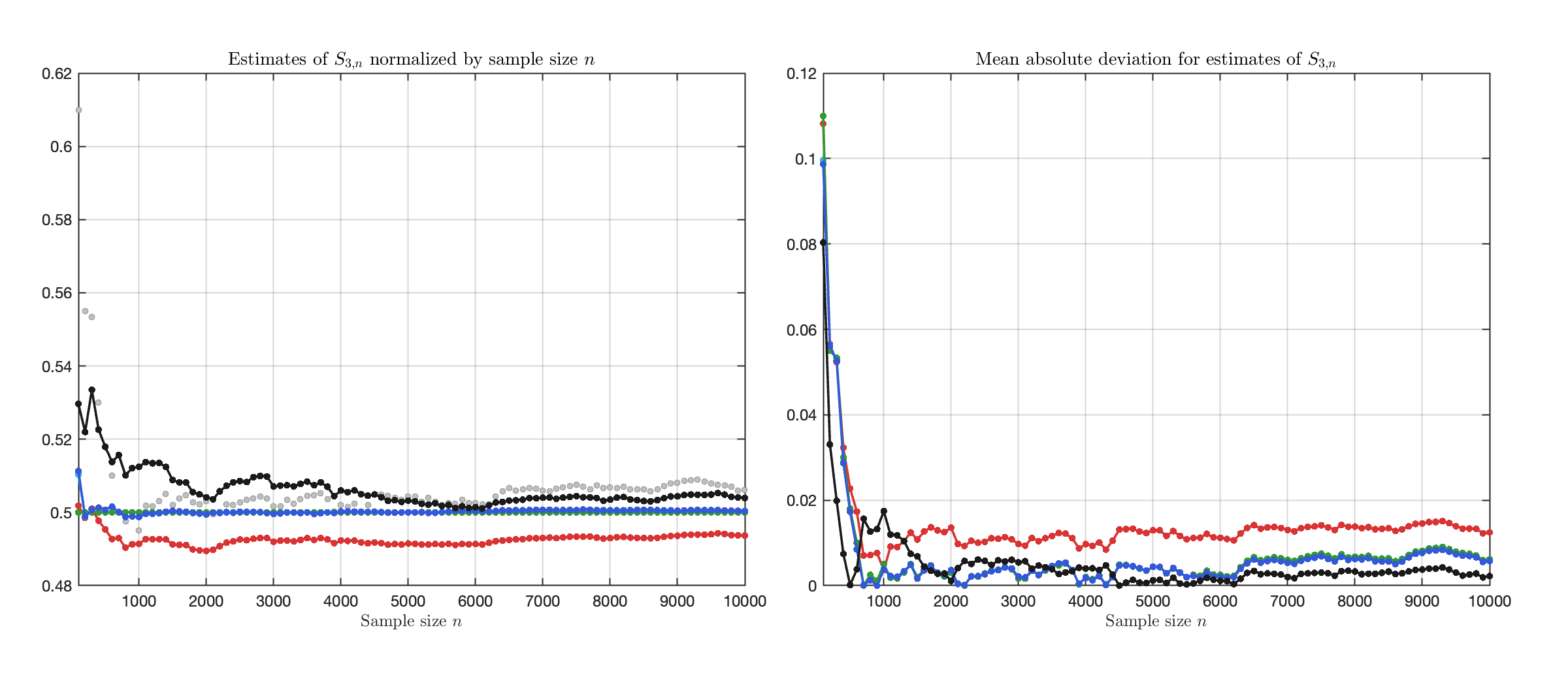}
\end{center}
\caption{\scriptsize{Gaussian prior, $S_{3,n}$. Left panel: true values of $n^{-1}S_{3,n}$ (Grey o-) and estimates $n^{-1}\hat{S}^{\text{\tiny{[O]}}}_{1,n}$ (Black .-), $n^{-1}\hat{S}^{\text{\tiny{[ML]}}}_{1,n}$ (Blue .-), $n^{-1}\hat{S}^{\text{\tiny{[B]}}}_{3,n}$ (Cyan .-), $n^{-1}\hat{S}^{\text{\tiny[``u,v"]}}_{3,n}$ (Green .-) and $n^{-1}\hat{S}^{\text{\tiny{[Q-B]}}}_{3,n}$ (red .-). Right panel: MAD of $\hat{S}^{\text{\tiny{[O]}}}_{1,n}$ (Black .-), $\hat{S}^{\text{\tiny{[ML]}}}_{3,n}$ (Blue .-), $\hat{S}^{\text{\tiny{[B]}}}_{3,n}$ (Cyan .-), $\hat{S}^{\text{\tiny[``u,v"]}}_{3,n}$ (Green .-) and $\hat{S}^{\text{\tiny{[Q-B]}}}_{3,n}$ (red .-)}}
\label{fig_gauss_s3}
\end{figure}

\begin{table}[ht]
\centering
\caption{Gaussian prior, $S_{3,n}$: MAD as $n$ varies}
{
\setlength{\tabcolsep}{0pt}
\begin{tabular}{@{}l@{\hspace{1.2cm}}*{5}{>{\centering\arraybackslash}p{2.15cm}}@{}}
\hline
\hline
& $\hat{S}^{\text{\tiny{[O]}}}_{3,n}$ & $\hat{S}^{\text{\tiny{[ML]}}}_{3,n}$ & $\hat{S}^{\text{\tiny{[B]}}}_{3,n}$ & $\hat{S}^{\text{\tiny{[``u,v'']}}}_{3,n}$ & $\hat{S}^{\text{\tiny{[Q-B]}}}_{3,n}$ \\[0.1cm]
\hline
\multicolumn{6}{@{}l}{\underline{$n=1{,}000$}} \\[0.05cm]
MAD & 0.0174 & 0.0037 & 0.0037 & 0.0050 & \textbf{0.0036} \\[0.2cm]

\multicolumn{6}{@{}l}{\underline{$n=3{,}000$}} \\[0.05cm]
MAD & 0.0054 & 0.0021 & 0.0021 & 0.0017 & \textbf{0.0097} \\[0.2cm]

\multicolumn{6}{@{}l}{\underline{$n=5{,}000$}} \\[0.05cm]
MAD & 0.0012 & 0.0044 & 0.0044 & 0.0044 & \textbf{0.0129} \\[0.2cm]

\multicolumn{6}{@{}l}{\underline{$n=7{,}000$}} \\[0.05cm]
MAD & 0.0020 & 0.0053 & 0.0053 & 0.0060 & \textbf{0.0130} \\[0.2cm]

\multicolumn{6}{@{}l}{\underline{$n=9{,}000$}} \\[0.05cm]
MAD & 0.0037 & 0.0077 & 0.0077 & 0.0081 & \textbf{0.0145} \\[0.1cm]
\hline
\hline
\end{tabular}
}
\label{tab_gauss_s3}
\end{table}

The results in Figure~\ref{fig_gauss_s1}--\ref{fig_gauss_s3} and Table~\ref{tab_gauss_s1}--\ref{tab_gauss_s3} show that the competing estimates exhibit broadly comparable performance. Focusing first on the estimation of $S_{1,n}$, the quasi-Bayes EB estimate performs competitively with the parametric EB and Bayes EB estimates; the ``$u,v$'' estimate is not available in this case. With regards to the estimation of $S_{3,n}$, all the estimates, including the ``$u,v$'' estimate, which is available for $S_{3,n}$ exhibit similar levels of empirical accuracy, with no method clearly dominating the others. Overall, the results indicate that the quasi-Bayes EB estimate remains competitive across  $S_{1,n}$ and $S_{3,n}$. 

A sensitivity analysis of the estimate $\hat{S}^{\text{\tiny{[Q-B]}}}_{1,n}$, with respect to the grid resolution $d\in\{5{,}000,1{,}000,500,100,50\}$, is reported in Table~\ref{tab_gauss_s1_sens_supp}, showing the mean absolute deviations (MAD) and CPU times as both $n$ and $d$ vary. The CPU time refers to the time (in seconds) for processing a new observation on a laptop MacBook Pro (M1 type processor). Results in Table~\ref{tab_gauss_s1_sens_supp} show that the empirical performance of $\hat{S}^{\text{\tiny{[Q-B]}}}_{1,n}$ is robust to the choice of the grid resolution $d$.

\begin{table}[ht]
\centering
\caption{Gaussian prior: MAD and CPU time (in seconds) of $\hat{S}^{\text{\tiny{[Q-B]}}}_{1,n}$ as $n$ and $d$ vary}
{
\setlength{\tabcolsep}{0pt}
\begin{tabular}{@{}l@{\hspace{1.2cm}}*{5}{>{\centering\arraybackslash}p{2.15cm}}@{}}
\hline
\hline
 & $d=5{,}000$ & $d=1{,}000$ & $d=500$ & $d=100$ & $d=50$ \\[0.1cm]
\hline
\multicolumn{6}{@{}l}{\underline{$n=1{,}000$}} \\[0.05cm]
MAD      & 0.0114 & 0.0114 & 0.0114 & 0.0115 & 0.0114 \\
CPU time & 0.0002 & 9.1523e-05 & 5.2018e-05 & 1.9854e-05 & 1.7194e-05 \\[0.2cm]

\multicolumn{6}{@{}l}{\underline{$n=2{,}000$}} \\[0.05cm]
MAD      & 0.0007 & 0.0007 & 0.0007 & 0.0007 & 0.0008 \\
CPU time & 0.0002 & 8.936e-05 & 5.7565e-05 & 1.9696e-05 & 1.7219e-05 \\[0.2cm]

\multicolumn{6}{@{}l}{\underline{$n=3{,}000$}} \\[0.05cm]
MAD      & 0.0033 & 0.0033 & 0.0033 & 0.0034 & 0.0033 \\
CPU time & 0.0003 & 7.8414e-05 & 5.3328e-05 & 1.9429e-05 & 1.7024e-05 \\[0.2cm]

\multicolumn{6}{@{}l}{\underline{$n=4{,}000$}} \\[0.05cm]
MAD      & 0.0050 & 0.0050 & 0.0050 & 0.0050 & 0.0051 \\
CPU time & 0.0002 & 8.7816e-05 & 5.4923e-05 & 1.9846e-05 & 1.7595e-05 \\[0.2cm]

\multicolumn{6}{@{}l}{\underline{$n=5{,}000$}} \\[0.05cm]
MAD      & 0.0092 & 0.0092 & 0.0092 & 0.0092 & 0.0093 \\
CPU time & 0.0002 & 7.5229e-05 & 5.9706e-05 & 2.1562e-05 & 2.2959e-05 \\[0.2cm]

\multicolumn{6}{@{}l}{\underline{$n=6{,}000$}} \\[0.05cm]
MAD      & 0.0066 & 0.0066 & 0.0066 & 0.0066 & 0.0067 \\
CPU time & 0.0002 & 8.5247e-05 & 5.2317e-05 & 2.3341e-05 & 1.7645e-05 \\[0.2cm]

\multicolumn{6}{@{}l}{\underline{$n=7{,}000$}} \\[0.05cm]
MAD      & 0.0094 & 0.0094 & 0.0094 & 0.0095 & 0.0095 \\
CPU time & 0.0002 & 8.0833e-05 & 5.4014e-05 & 1.9908e-05 & 1.8113e-05 \\[0.2cm]

\multicolumn{6}{@{}l}{\underline{$n=8{,}000$}} \\[0.05cm]
MAD      & 0.0104 & 0.0104 & 0.0104 & 0.0104 & 0.0102 \\
CPU time & 0.0002 & 8.3144e-05 & 5.3945e-05 & 1.9919e-05 & 1.7194e-05 \\[0.2cm]

\multicolumn{6}{@{}l}{\underline{$n=9{,}000$}} \\[0.05cm]
MAD      & 0.0088 & 0.0088 & 0.0088 & 0.0088 & 0.0087 \\
CPU time & 0.0002 & 8.8492e-05 & 5.2426e-05 & 1.9771e-05 & 1.7087e-05 \\[0.2cm]

\multicolumn{6}{@{}l}{\underline{$n=10{,}000$}} \\[0.05cm]
MAD      & 0.0069 & 0.0069 & 0.0069 & 0.0069 & 0.0067 \\
CPU time & 0.0002 & 7.9214e-05 & 4.8358e-05 & 2.051e-05 & 1.7363e-05 \\[0.1cm]
\hline
\hline
\end{tabular}
}
\label{tab_gauss_s1_sens_supp}
\end{table}

To conclude, we provide quasi-Bayes credible intervals at level $1-\alpha=0.95$. Specifically, credible intervals are constructed from \eqref{cred_interval} by relying on Newton's algorithm initialized as in Table~\ref{tab_gauss_s1}. Figure~\ref{gaussgauss_fig_interval_s1}-\ref{gaussgauss_fig_interval_s3} display the quasi-Bayes credible intervals for $S_{1,n}$ and $S_{3,n}$, respectively, with the corresponding oracle credible intervals $I^\ast_{S_{1,n}}$ and $I^\ast_{S_{3,n}}$ under the Gaussian prior. Under the Gaussian mixture model \eqref{model_distr_gaussian}, oracle credible intervals for both $S_{1,n}$ and $S_{3,n}$ are obtained by following the same construction as in the Poisson mixture model, with the corresponding Gaussian conditional distribution replacing the Poisson distribution.

 \begin{figure}[h!]
\begin{center}
\includegraphics[width=1\linewidth,height=0.36\textheight,keepaspectratio]{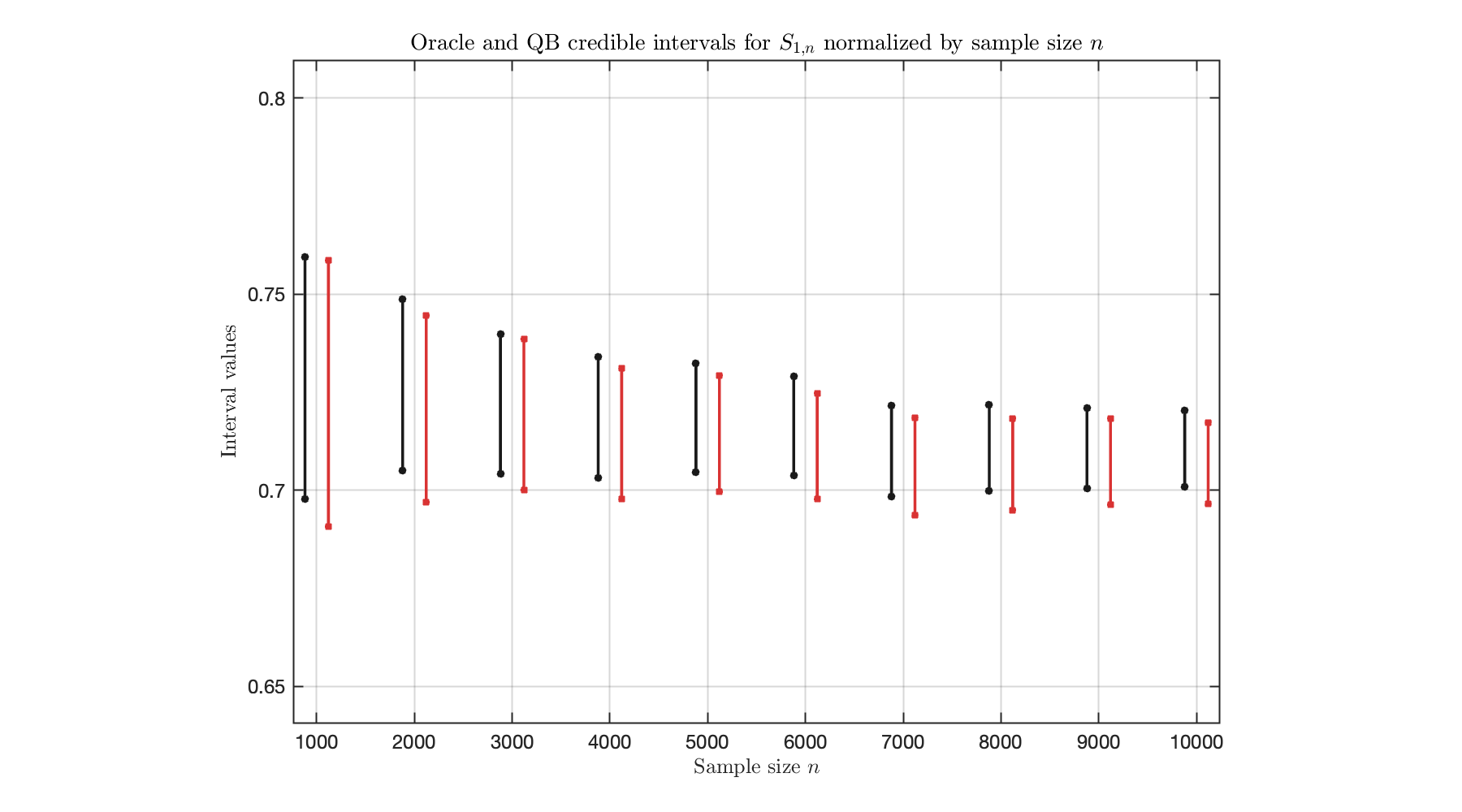}
\end{center}
\caption{\scriptsize{Gaussian prior, $S_{1,n}$: oracle credible intervals (black) and quasi-Bayes credible intervals (red)}}
\label{gaussgauss_fig_interval_s1}
\end{figure}

\begin{figure}[h!]
\begin{center}
\includegraphics[width=1\linewidth,height=0.36\textheight,keepaspectratio]{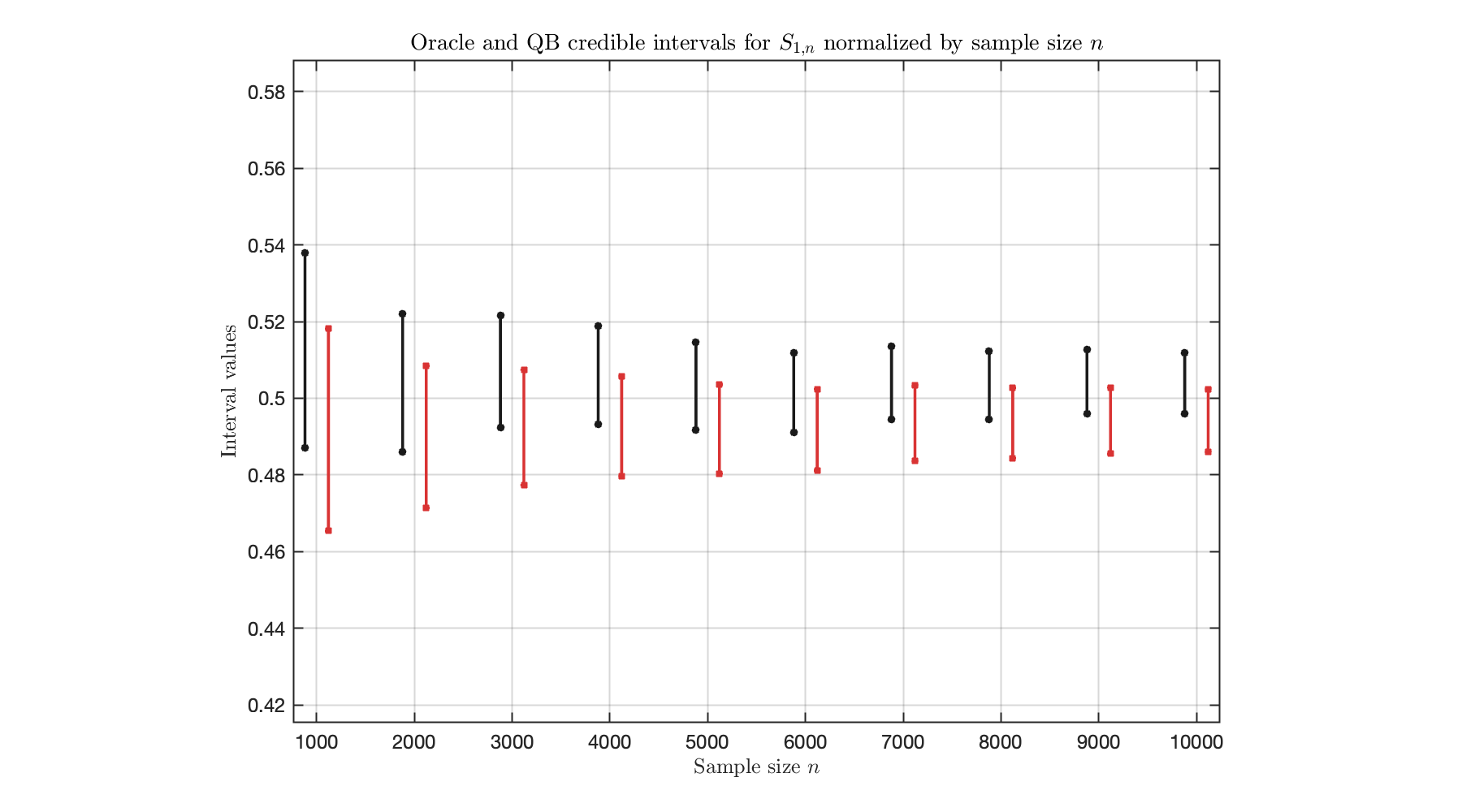}
\end{center}
\caption{\scriptsize{Gaussian prior, $S_{3,n}$: oracle credible intervals (black) and quasi-Bayes credible intervals (red)}}
\label{gaussgauss_fig_interval_s3}
\end{figure}

\section{Multidimensional extension}\label{app_mult}

The proposed quasi-Bayes EB methodology extends naturally, with only minor modifications, to multidimensional settings under a coordinate-wise independence assumption. Let $\Xb_i=(X_{i,1},\dots,X_{i,d})$ be a $d$-dimensional random vector with independent coordinates, $i=1,\ldots,n$, where each $X_{i,j}$ has density function $k(\cdot\mid \theta_{i,j})$. In particular, writing $\thetab_i=(\theta_{i,1},\dots,\theta_{i,d})\in\Theta^d$,
\begin{displaymath}
k(\xb_i\mid \thetab_i)=\prod_{j=1}^k k(x_{i,j}\mid \theta_{i,j}),
\end{displaymath}
and the corresponding marginal density function of the $\Xb_i$'s is obtained by means of mixing over a distribution $G$ on $\Theta^d$. Specifically, if $\thetab_{i}=(\theta_{i,1},\dots,\theta_{i,d})$ has the probability distribution $G$ on ${\Theta}^d$, then the marginal density function of $\Xb_i=(X_{i,1},\dots,X_{i,d})$ with respect to $\lambda^k$ is
\begin{equation}\label{eq:marginalmultiple}
f_G(\xb)=\int_{{\Theta}^d}k(\xb\mid\thetab)G(\ddr\thetab).
\end{equation}
Given a measurable vector-valued utility function $\ub:\mathbb X^d\times\Theta^d\to\mathbb R^s$, we consider the estimation of
\begin{displaymath}
\Sb_n=\sum_{i=1}^n \ub(\Xb_i,\thetab_i).
\end{displaymath}
In this multidimensional framework, Newton’s algorithm extends verbatim, yielding a recursive update of the mixing distribution $G_n$ on $\Theta^d$. Specifically, the Newton's algorithm becomes
\begin{equation}\label{eq:newtonmult}
G_{n+1}(\ddr\thetab)=(1-\alpha_{n+1})G_n(\ddr\thetab)+\alpha_{n+1}
\frac{k(\Xb_{n+1}\mid\thetab) G_n(\ddr\thetab)}{f_{G_n}(\Xb_{n+1})},
\end{equation}
where, for every $G$, $f_{G}(\xb)=\int_{{\Theta}^d} k(\xb\mid\thetab)G(\ddr\thetab)$, we denote by $G_0$ a fixed probability mass function on ${\Theta}^d$ and $(\alpha_n)$ is a sequence of positive numbers (namely the learning rate) satisfying $\sum_{n=1}^\infty \alpha_n=+\infty$ and $\sum_{n=1}^\infty \alpha_n^2<+\infty$.  Note that, while the computational cost increases with the dimension $d$, Newton's algorithm \eqref{eq:newtonmult} preserves constant per-observation complexity.

Based on Newton’s algorithm \eqref{eq:newtonmult}, the quasi-Bayes EB estimate $\hat{\Sb}^{\text{\tiny{[Q-B]}}}_{n}$ of $\Sb_n$ is defined as  
\begin{displaymath}
    \hat{\Sb}^{\text{\tiny{[Q-B]}}}_{n}=\E^{(n)}(\Sb_n)=\sum_{j=1}^n \int_{{\Theta}^d} \ub(\Xb_j,\thetab)\,G_n(\ddr\thetab\mid \Xb_j).
\end{displaymath}
All the properties discussed in Section~\ref{sec2} extend to the multidimensional $\hat{\Sb}^{\text{\tiny{[Q-B]}}}_{n}$, and the corresponding proofs follow essentially the same arguments. In particular,  Theorem~\ref{th:central} becomes
\begin{displaymath}
\Pn\left(B_n^{-1/2}\big(\hat{\Sb}^{\text{\tiny{[Q-B]}}}_{n}- \Sb_n\big)\in\cdot\;\middle|\;\Xb_{1:n}\right)
{\longrightarrow}\mathcal N_s(0,I)\quad \P\mbox{-a.s.},
\end{displaymath}
where $\mathcal N_s(0,I)$ denotes the $s$-dimensional standard Gaussian distribution, and  $B_n=\sum_{j=1}^n \Sigma_{n,j}$, with  
\begin{displaymath}
\Sigma_{n,j}=\int_{{\Theta}^d} \ub(\Xb_j,\thetab)\ub(\Xb_j,\thetab)^{T} G_n(\ddr\thetab\mid \Xb_j)
-\ub_n(\Xb_j)\ub_n(\Xb_j)^{T},
\end{displaymath}
and  
\begin{displaymath}
\ub_n(\Xb_j)=\int_{{\Theta}^d}\ub(\Xb_j,\thetab)\,G_n(\ddr\thetab\mid \Xb_j)\mu(d\thetab).
\end{displaymath}
Very similarly, all the properties established in Section~\ref{sec3} extend to the multidimensional $\hat{\Sb}^{\text{\tiny{[Q-B]}}}_{n}$.


\section{Nonparametric EB}\label{app_neb}

Consider the Poisson mixture model \eqref{model_distr_poi} with $G$ completely unknown, namely without any parametric assumption on $G$, and consider the problem of estimating $S_{1,n}$ and $S_{3,n}$. We implement a nonparametric EB approach by replacing $G$ with its nonparametric maximum likelihood estimate, i.e.
\begin{equation}\label{ml}
\hat{G}_{n}=\argmax_{G}\frac{1}{n}\sum_{i=1}^{n}\log p_{G}(y_{i}),
\end{equation}
where 
\begin{displaymath}
p_{G}(y_{i})=\P(Y_{i}=y_{i})=\int_{\Theta}\frac{\theta^{y_{i}}\text{e}^{-\theta}}{y_{i}!}G(\ddr\theta)
\end{displaymath}
and the maximization \eqref{ml} is taken over all distributions on $\Theta$ \citep{Lin(95)}. In particular, $\hat{G}_{n}$ is computed via the vertex direction algorithm specialized to the Poisson setting \citep[Section 5]{Jan(24)}. Hence, the corresponding nonparametric EB estimates of $S_{1,n}$ and $S_{3,n}$ are
\begin{equation}\label{nml_est_s1}
\hat{S}^{\text{\tiny{[N-ML]}}}_{1,n}=\sum_{i=1}^{n}\E_{\hat{G}_{n}}[\theta_{i}\,|\,X_{i}]I(X_{i}\leq\kappa).
\end{equation}
and
\begin{equation}\label{nml_est_s3}
\hat{S}^{\text{\tiny{[N-ML]}}}_{1,n}=\sum_{i=1}^{n}\E_{\hat{G}_{n}}[I(X_{i}>\theta_{i})\,|\,X_{i}],
\end{equation}
respectively. As an alternative to nonparametric maximum likelihood one may also consider minimum distance estimates, as proposed in \citep{Jan(24)} for the Poisson compound decision problem.

We apply the nonparametric EB estimates \eqref{nml_est_s1} and \eqref{nml_est_s3} to the synthetic data of Section~\ref{sec4}. The vertex direction algorithm is based on an initial discretization of $G$ over a finite grid over $\Theta$. Following \citep[Section 5]{Jan(24)}, we consider a uniform grid of $d=1,000$ points over the interval $(0,\max{X_{1},\ldots,X_{n}})$, with a uniform initial probability assignment. Figure~\ref{fig_weib_s1_new} reports estimates $\hat{S}^{\text{\tiny{[O]}}}_{1,n}$, $\hat{S}^{\text{\tiny{[ML]}}}_{1,n}$, $\hat{S}^{\text{\tiny{[B]}}}_{1,n}$, $\hat{S}^{\text{\tiny[``u,v"]}}_{1,n}$, $\hat{S}^{\text{\tiny{[N-ML]}}}_{1,n}$  and $\hat{S}^{\text{\tiny{[Q-B]}}}_{1,n}$, normalized by the sample size $n$, as well as the mean absolute deviation (MAD) for the estimates of $S_{1,n}$. See also Table~\ref{tab_weib_s1_supp_new}. Table~\ref{tab_weib_s1ml_supp} reports the CPU times of $\hat{S}^{\text{\tiny{[N-ML]}}}_{1,n}$ and $\hat{S}^{\text{\tiny{[Q-B]}}}_{1,n}$ as the sample size $n$ varies. Figure~\ref{fig_weib_s3_new} reports estimates $\hat{S}^{\text{\tiny{[O]}}}_{3,n}$, $\hat{S}^{\text{\tiny{[ML]}}}_{3,n}$, $\hat{S}^{\text{\tiny{[B]}}}_{3,n}$, $\hat{S}^{\text{\tiny{[N-ML]}}}_{3,n}$  and $\hat{S}^{\text{\tiny{[Q-B]}}}_{3,n}$, normalized by the sample size $n$, as well as the mean absolute deviation (MAD) for the estimates of $S_{3,n}$. See also Table~\ref{tab_weib_s3_supp_new}.

\begin{figure}[h!]
\begin{center}
\includegraphics[width=1\linewidth,height=0.36\textheight,keepaspectratio]{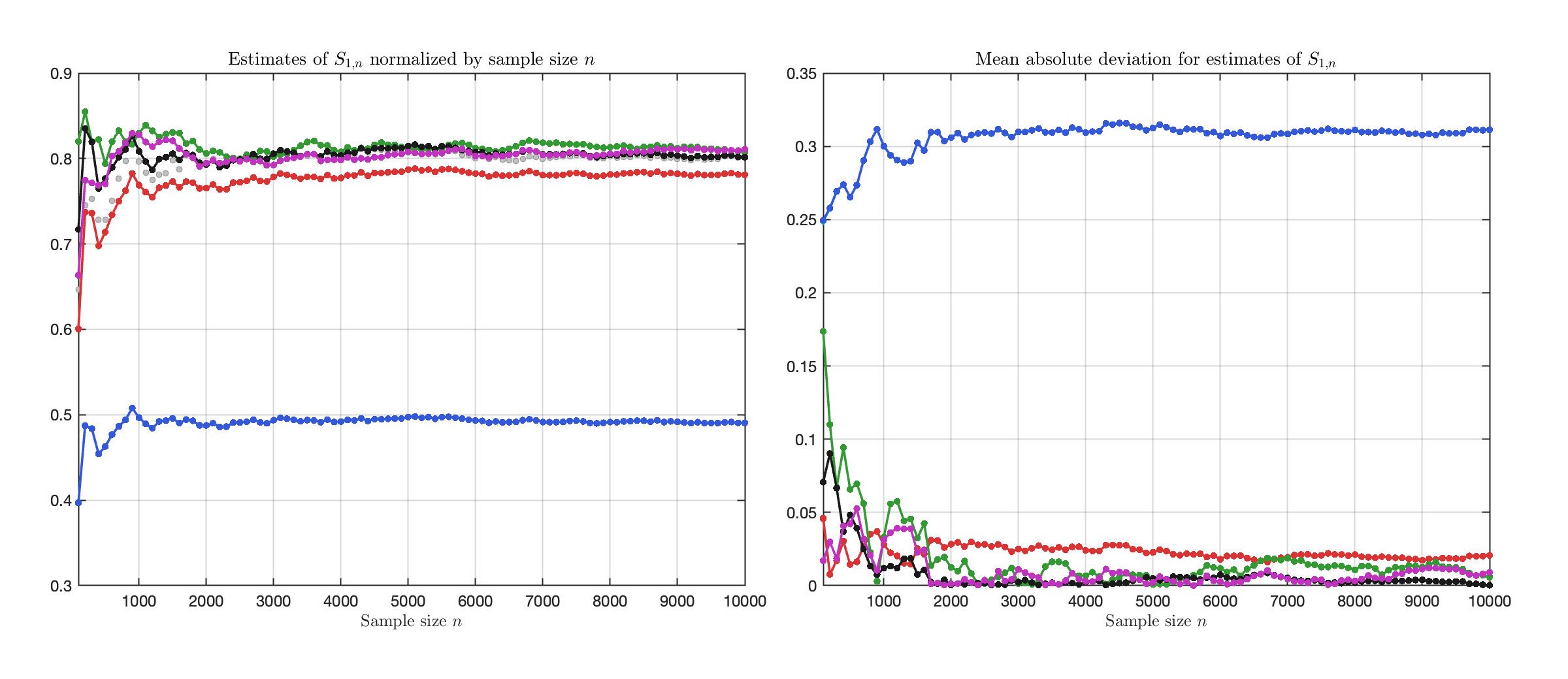}
\end{center}
\caption{\scriptsize{Weibull prior, $S_{1,n}$. Left panel: true values $n^{-1}S_{1,n}$ (Grey o-) and estimates $n^{-1}\hat{S}^{\text{\tiny{[O]}}}_{1,n}$ (Black .-), $n^{-1}\hat{S}^{\text{\tiny{[ML]}}}_{1,n}$ (Blue .-), $n^{-1}\hat{S}^{\text{\tiny{[B]}}}_{1,n}$ (Cyan .-), $\hat{S}^{\text{\tiny{[N-ML]}}}_{1,n}$ (Magenta .-) and $n^{-1}\hat{1}^{\text{\tiny{[Q-B]}}}_{1,n}$ (Red .-). Right panel: MAD of $\hat{S}^{\text{\tiny{[O]}}}_{1,n}$ (Black .-), $\hat{S}^{\text{\tiny{[ML]}}}_{1,n}$ (Blue .-), $\hat{S}^{\text{\tiny{[B]}}}_{1,n}$ (Cyan .-), $\hat{S}^{\text{\tiny{[N-ML]}}}_{1,n}$ (Magenta .-) and $\hat{S}^{\text{\tiny{[Q-B]}}}_{1,n}$ (Red .-)}}
\label{fig_weib_s1_new}
\end{figure}

\begin{table}[ht]
\centering
\caption{Weibull prior, $S_{1,n}$: MAD as $n$ varies}
{
\setlength{\tabcolsep}{0pt}
\begin{tabular}{@{}l@{\hspace{1.2cm}}*{6}{>{\centering\arraybackslash}p{2.15cm}}@{}}
\hline
\hline
& $\hat{S}^{\text{\tiny{[O]}}}_{1,n}$ & $\hat{S}^{\text{\tiny{[ML]}}}_{1,n}$ & $\hat{S}^{\text{\tiny{[B]}}}_{1,n}$ & $\hat{S}^{\text{\tiny{[``u,v'']}}}_{1,n}$ & $\hat{S}^{\text{\tiny{[N-ML]}}}_{1,n}$ & $\hat{S}^{\text{\tiny{[Q-B]}}}_{1,n}$ \\[0.1cm]
\hline
\multicolumn{7}{@{}l}{\underline{$n=1{,}000$}} \\[0.05cm]
MAD & 0.0118 & 0.3002 & 0.3002 & 0.0331 & 0.0312 & \textbf{0.0279} \\[0.2cm]

\multicolumn{7}{@{}l}{\underline{$n=2{,}000$}} \\[0.05cm]
MAD & 0.0004 & 0.3060 & 0.3061 & 0.0124 & 0.0012 & \textbf{0.0282} \\[0.2cm]

\multicolumn{7}{@{}l}{\underline{$n=3{,}000$}} \\[0.05cm]
MAD & 0.0022 & 0.3099 & 0.3099 & 0.0012 & 0.0109 & \textbf{0.0250} \\[0.2cm]

\multicolumn{7}{@{}l}{\underline{$n=4{,}000$}} \\[0.05cm]
MAD & 0.0023 & 0.3094 & 0.3095 & 0.0066 & 0.0029 & \textbf{0.0240} \\[0.2cm]

\multicolumn{7}{@{}l}{\underline{$n=5{,}000$}} \\[0.05cm]
MAD & 0.0048 & 0.3126 & 0.3126 & 0.0007 & 0.0022 & \textbf{0.0227} \\[0.2cm]

\multicolumn{7}{@{}l}{\underline{$n=6{,}000$}} \\[0.05cm]
MAD & 0.0074 & 0.3071 & 0.3071 & 0.0118 & 0.0026 & \textbf{0.0179} \\[0.2cm]

\multicolumn{7}{@{}l}{\underline{$n=7{,}000$}} \\[0.05cm]
MAD & 0.0051 & 0.3085 & 0.3085 & 0.0188 & 0.0045 & \textbf{0.0193} \\[0.2cm]

\multicolumn{7}{@{}l}{\underline{$n=8{,}000$}} \\[0.05cm]
MAD & 0.0019 & 0.3112 & 0.3112 & 0.0107 & 0.0031 & \textbf{0.0211} \\[0.2cm]

\multicolumn{7}{@{}l}{\underline{$n=9{,}000$}} \\[0.05cm]
MAD & 0.0038 & 0.3079 & 0.3079 & 0.0124 & 0.0113 & \textbf{0.0173} \\[0.2cm]

\multicolumn{7}{@{}l}{\underline{$n=10{,}000$}} \\[0.05cm]
MAD & 0.0002 & 0.3114 & 0.3114 & 0.0058 & 0.0090 & \textbf{0.0206} \\[0.1cm]
\hline
\hline
\end{tabular}
}
\label{tab_weib_s1_supp_new}
\end{table}

\begin{table}[ht]
\centering
\caption{Weibull prior: CPU time (in seconds) of $\hat{S}^{\text{\tiny{[N-ML]}}}_{1,n}$ and $\hat{S}^{\text{\tiny{[Q-B]}}}_{1,n}$ for processing $n$ data}
{
\setlength{\tabcolsep}{0pt}
\begin{tabular}{@{}l@{\hspace{1.2cm}}*{2}{>{\centering\arraybackslash}p{2.15cm}}@{}}
\hline
\hline
& $\hat{S}^{\text{\tiny{[N-ML]}}}_{1,n}$ & $\hat{S}^{\text{\tiny{[Q-B]}}}_{1,n}$ \\[0.1cm]
\hline
\multicolumn{3}{@{}l}{\underline{$n=1{,}000$}} \\[0.05cm]
CPU time & 2.2478 & 1.0238 \\[0.2cm]

\multicolumn{3}{@{}l}{\underline{$n=2{,}000$}} \\[0.05cm]
CPU time & 4.2345 & 2.1108 \\[0.2cm]

\multicolumn{3}{@{}l}{\underline{$n=3{,}000$}} \\[0.05cm]
CPU time & 5.9023 & 3.3306 \\[0.2cm]

\multicolumn{3}{@{}l}{\underline{$n=4{,}000$}} \\[0.05cm]
CPU time & 10.879 & 2.1589 \\[0.2cm]

\multicolumn{3}{@{}l}{\underline{$n=5{,}000$}} \\[0.05cm]
CPU time & 8.0596 & 4.3178 \\[0.2cm]

\multicolumn{3}{@{}l}{\underline{$n=6{,}000$}} \\[0.05cm]
CPU time & 13.171 & 3.3566 \\[0.2cm]

\multicolumn{3}{@{}l}{\underline{$n=7{,}000$}} \\[0.05cm]
CPU time & 13.087 & 6.7132 \\[0.2cm]

\multicolumn{3}{@{}l}{\underline{$n=8{,}000$}} \\[0.05cm]
CPU time & 13.246 & 8.5050 \\[0.2cm]

\multicolumn{3}{@{}l}{\underline{$n=9{,}000$}} \\[0.05cm]
CPU time & 15.983 & 9.6956 \\[0.2cm]

\multicolumn{3}{@{}l}{\underline{$n=10{,}000$}} \\[0.05cm]
CPU time & 17.756 & 10.9712 \\[0.1cm]
\hline
\hline
\end{tabular}
}
\label{tab_weib_s1ml_supp_new}
\end{table}

\begin{figure}[h!]
\begin{center}
\includegraphics[width=1\linewidth,height=0.36\textheight,keepaspectratio]{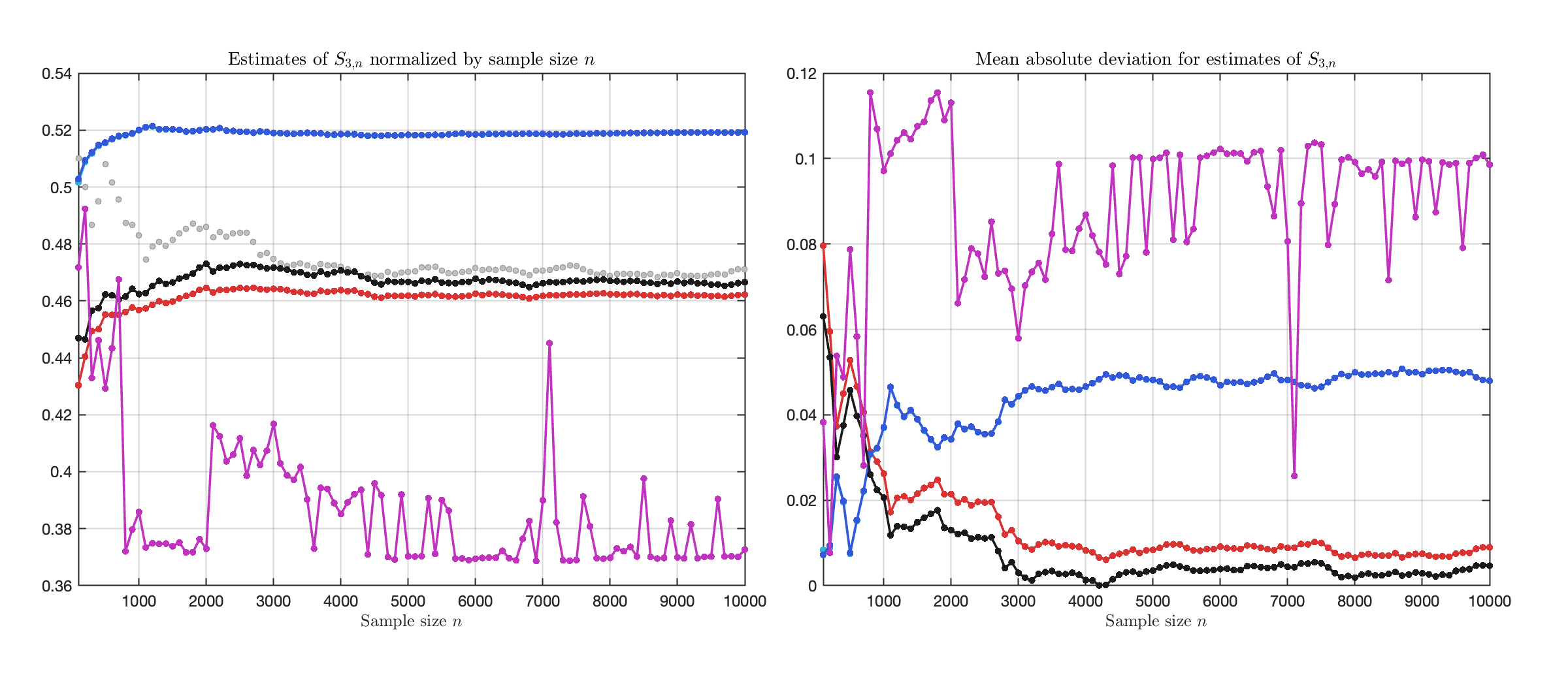}
\end{center}
\caption{\scriptsize{Weibull prior, $S_{3,n}$. Left panel: true values $n^{-1}S_{3,n}$ (Grey o-) and estimates $n^{-1}\hat{S}^{\text{\tiny{[O]}}}_{3,n}$ (Black .-), $n^{-1}\hat{S}^{\text{\tiny{[ML]}}}_{3,n}$ (Blue .-), $n^{-1}\hat{S}^{\text{\tiny{[B]}}}_{3,n}$ (Cyan .-), $\hat{S}^{\text{\tiny{[N-ML]}}}_{3,n}$ (Magenta .-) and $n^{-1}\hat{S}^{\text{\tiny{[Q-B]}}}_{3,n}$ (Red .-). Right panel: MAD of $\hat{S}^{\text{\tiny{[O]}}}_{3,n}$ (Black .-), $\hat{S}^{\text{\tiny{[ML]}}}_{3,n}$ (Blue .-), $\hat{S}^{\text{\tiny{[B]}}}_{3,n}$ (Cyan .-), $\hat{S}^{\text{\tiny{[N-ML]}}}_{3,n}$ (Magenta .-) and $\hat{S}^{\text{\tiny{[Q-B]}}}_{3,n}$ (Red .-)}}
\label{fig_weib_s3_new}
\end{figure}

\begin{table}[ht]
\centering
\caption{Weibull prior, $S_{3,n}$: MAD as $n$ varies}
{
\setlength{\tabcolsep}{0pt}
\begin{tabular}{@{}l@{\hspace{1.2cm}}*{5}{>{\centering\arraybackslash}p{2.15cm}}@{}}
\hline
\hline
& $\hat{S}^{\text{\tiny{[O]}}}_{3,n}$ & $\hat{S}^{\text{\tiny{[ML]}}}_{3,n}$ & $\hat{S}^{\text{\tiny{[B]}}}_{3,n}$ & $\hat{S}^{\text{\tiny{[N-ML]}}}_{3,n}$ & $\hat{S}^{\text{\tiny{[Q-B]}}}_{3,n}$ \\[0.1cm]
\hline
\multicolumn{6}{@{}l}{\underline{$n=1{,}000$}} \\[0.05cm]
MAD & 0.0206 & 0.0369 & 0.0371 & 0.0972 & \textbf{0.0262} \\[0.2cm]

\multicolumn{6}{@{}l}{\underline{$n=2{,}000$}} \\[0.05cm]
MAD & 0.0130 & 0.0342 & 0.0343 & 0.1131 & \textbf{0.0214} \\[0.2cm]

\multicolumn{6}{@{}l}{\underline{$n=3{,}000$}} \\[0.05cm]
MAD & 0.0030 & 0.0443 & 0.0444 & 0.0579 & \textbf{0.0104} \\[0.2cm]

\multicolumn{6}{@{}l}{\underline{$n=4{,}000$}} \\[0.05cm]
MAD & 0.0013 & 0.0466 & 0.0466 & 0.0869 & \textbf{0.0082} \\[0.2cm]

\multicolumn{6}{@{}l}{\underline{$n=5{,}000$}} \\[0.05cm]
MAD & 0.0035 & 0.0482 & 0.0482 & 0.0999 & \textbf{0.0084} \\[0.2cm]

\multicolumn{6}{@{}l}{\underline{$n=6{,}000$}} \\[0.05cm]
MAD & 0.0039 & 0.0469 & 0.0469 & 0.1023 & \textbf{0.0091} \\[0.2cm]

\multicolumn{6}{@{}l}{\underline{$n=7{,}000$}} \\[0.05cm]
MAD & 0.0044 & 0.0481 & 0.0481 & 0.0806 & \textbf{0.0088} \\[0.2cm]

\multicolumn{6}{@{}l}{\underline{$n=8{,}000$}} \\[0.05cm]
MAD & 0.0019 & 0.0499 & 0.0500 & 0.0992 & \textbf{0.0065} \\[0.2cm]

\multicolumn{6}{@{}l}{\underline{$n=9{,}000$}} \\[0.05cm]
MAD & 0.0028 & 0.0495 & 0.0495 & 0.0998 & \textbf{0.0074} \\[0.2cm]

\multicolumn{6}{@{}l}{\underline{$n=10{,}000$}} \\[0.05cm]
MAD & 0.0046 & 0.0480 & 0.0480 & 0.0986 & \textbf{0.0090} \\[0.1cm]
\hline
\hline
\end{tabular}
}
\label{tab_weib_s3_supp_new}
\end{table}


\section{Species-sampling problems}\label{app_ssp}

Consider $N\geq1$ random samples (with replacement) from a population of $d$ species. Given observed species' frequencies in the samples, species sampling problems aim to estimate unobserved characteristics of the population, as well as to predict the species composition of $M\geq1$ additional unobservable samples from the same population. Examples include the number $d$ of species in the population, the probability of discovering a new species, and the number of new species in the additional samples. Letting $X_{i}$ and $Y_{i}$ denote the frequencies of the $i$-th species in the initial and additional samples, respectively, for $i=1,\ldots,d$, we consider the model
\begin{align}\label{model_species}
Y_i\mid X_{i},\,\lambda_{i} & \quad\simind\quad\text{Poisson}(\lambda_{i})\qquad i=1,\ldots,d\\[-0.2cm]
\notag X_i\mid\lambda_{i} & \quad\simind\quad\text{Poisson}(\lambda_{i})\\[-0.2cm]
\notag\lambda_{i}\mid G & \quad\simiid\quad G,
\end{align}
where $G$ is the unknown mixing distribution on the set $\mathbb{R}^{+}$ \citep[Section 3]{Zha(05)}. Now, let
\begin{displaymath}
p_{i}=\frac{\lambda_{i}}{\sum_{j=1}^{d}\lambda_{j}},
\end{displaymath}
be the probability of the $i$-th species in the population, for $i=1,\ldots,d$. In particular, recall that conditionally on the sample size $N=\sum_{1\leq i\leq n}X_{i}\sim\text{Poisson}(\sum_{1\leq i\leq d}\lambda_{i})$, the random variables $X_{i}$'s are distributed according to a Multinomial distribution with parameter $(N,p_{1},\ldots,p_{d})$.

We let $\P_{G}$ the probability measure that specifies the joint distribution of the $(Y_{i},X_{i},\lambda_{i})$'s in the model Poisson mixture model \eqref{model_species}. In particular, under $\P_{g}$, the likelihood function is 
\begin{equation}\label{like_species}
L_{\P_{G}}(G;\,n_{1},\ldots,n_{N})=\prod_{k=1}^{N}\left(\frac{\int_{\Lambda}\frac{\text{e}^{-\lambda}}{k!}\lambda^{k}G(\ddr\lambda)}{1-\int_{\Lambda}\text{e}^{-\lambda}G(\ddr\lambda)}\right)^{n_{k}},
\end{equation}
viewed as a function of the unknown mixing distribution $G$, where $n_{k}=\sum_{1\leq i\leq d}I(X_{i}=k)$ denotes the number of species with frequency $k$ in the sample, for $k=1,\ldots,N$ (frequencies of frequencies). That is, $\sum_{1\leq k\leq N}n_{k}$ is the number of species in the sample and $\sum_{1\leq k\leq N}kn_{k}=N$.

Hereafter, we interpret the quantity $d$ as the number of species represented in the population out of a total of $n$ species in a hypothetical superpopulation, assuming $p_{i}=0$ if the $i$-th species is not present in the population. Therefore, under this assumption, for a given
\begin{displaymath}
p^{\ast}<\P_{G}(X>0)=\E_{G}\left[\sum_{i=1}^{d}p_{i}I(X_{i}>0)\right],
\end{displaymath}
we denote by $\P_{p^{\ast},G}$ the probability measure that specifies the joint distribution of the random variables $(Y_{i},X_{i},\lambda_{i})$'s in the Poisson mixture model \eqref{model_species} when $d$ is replaced by $n$, and such that
\begin{equation}\label{model_species2}
\P_{(p^{\ast},G)}(X_{i}=x)=(1-p^{\ast})I(x=0)+p^{\ast}\frac{\P_{G}(X_{i}=x)}{\P_{G}(X>0)}I(x>0).
\end{equation}
The probability measure $\P_{(p^{\ast},g)}$ takes into account the probability $p^{\ast}$ that a species is represented in the population, merging with $\P_{G}$ as $p^{\ast}\rightarrow1$. It is reasonable to take $p^{\ast}$ close to $1$ if $N$ is sufficiently large. Under $\P_{(p^{\ast},G)}$, the likelihood function is from \eqref{like_species}-\eqref{model_species2}, and denoted by $L_{\P_{(p^{\ast},G})}$.  

The Poisson mixture model \eqref{model_species} with $d$ replaced by $n$, and probability measure $\P_{(p^{\ast},G)}$ in \eqref{model_species2}, allows to write the number $d$ of species in the population as a sum $S_{n}$. That is, we write
\begin{displaymath}
D_{n}=\sum_{i=1}^{n}I(p_{i}>0).
\end{displaymath}
Similarly, other species sampling problems can written as sums $S_{n}$, for a suitable choice of the utility $u$. For instance, the probability of discovering a new species can be written as follows:
\begin{displaymath}
M_{n}=\sum_{i=1}^{n}p_{i}I(p_{i}>0,X_{i}=0),
\end{displaymath}
i.e. the missing mass \citep{Goo(53),Efr(76),Mao(02)}. Furthermore, the number of new (hitherto unseen) species in the additional (unobserved) sample is
\begin{displaymath}
U_{n}=\sum_{i=1}^{n}I(p_{i}>0,X_{i}=0,Y_{i}>0)
\end{displaymath}
i.e. the unseen \citep{Goo(56),Efr(76)}. While many other examples of species sampling problems could be considered, their estimation follows similar lines.

Parametric EB and Bayes EB estimates of $D_{n}$, $M_{n}$ and $U_{n}$ can be derived as in Proposition \ref{prop:poisson}, by assuming the unknown mixing distribution $G$ to be parameterized by $\tau\in\mathcal{T}$, for an Euclidean space $\mathcal{T}$, here denoted by $G_{\tau}$. For instance, to estimate $D_{n}$ one has the following
\begin{equation}\label{species_d}
\hat{D}_{n}(G_{\tau})=\E_{(p^{\ast},G_{\tau})}[D_{n}\mid X_{1:n}]=\sum_{i=1}^{n}\P_{(p^{\ast},G_{\tau})}\left(\frac{\lambda_{i}}{\sum_{i=1}^{n}\lambda_{i}}>0\right),
\end{equation}
from which an EB estimate follows by replacing $\tau$ in \eqref{species_d} with the maximum likelihood estimate
\begin{displaymath}
\hat{\tau}_{n}=\text{Argmax}_{\tau\in\mathcal{T}}L_{\P_{(p^{\ast},G_{\tau}})}(G_{\tau};\,n_{1},\ldots,n_{N}).
\end{displaymath}
Similarly, a parametric Bayes EB estimate of $D_{n}$ follows by replacing $\tau$ in \eqref{species_d} with a Bayesian estimate of $\tau$, under suitable prior assumptions for $\tau$. Parametric EB and Bayes EB estimates of the quantities $M_{n}$ and $U_{n}$ are obtained along the very same lines \citep{Zha(05)}.

No ``$u,v$" estimates can be derived for $D_{n}$, $M_{n}$ and $U_{n}$, as the corresponding ``$u,v$"  (characteristic) integral equations do no admit solutions for $v$; see \citet[Section 3.2]{Zha(05)}.



\begin{thebibliography}{9}
\bibitem[Battiston and Cappello(2025)]{Bat(25)}
\textsc{Battiston, M. and Cappello, L.} (2025). New (and old) predictive schemes with a.c.i.d. sequences. \textit{Preprint arXiv:2507.21874}.

\bibitem[Bethlehem et al.(1990)]{Bet(90)}
\textsc{Bethlehem, J.G., Keller, W.J. and Pannekoek, J.} (1990). Disclosure control of microdata. \textit{J. Am. Statist. Assoc.} \textbf{85} 38--45.

\bibitem[Bissiri et al.(2016)]{Bis(16)}
\textsc{Bissiri, P.G., Holmes, C.C., and Walker, S.G.} (2007). A general framework for updating belief distributions. \textit{J. R. Statist. Soc. B} \textbf{78}, 1103--1130.

\bibitem[Brown et al.(2013)]{Bro(13)}
\textsc{Brown, L.D., Greenshtein, E. and Ritov, Y.} (2013). The Poisson compound decision problem revisited. \textit{J. Am. Statist. Assoc.} \textbf{108}, 741--749.

\bibitem[Bunge and Fitzpatrick(1993)]{Bun(93)}
\textsc{Bunge, J. and Fitzpatrick, M.} (1993) Estimating the number of species: a review. \textit{J. Am. Statist. Assoc.} \textbf{88}, 364-373.

\bibitem[Cannella et al.(2026)]{Can(26)}
\textsc{Cannella, N., Teh, A., Han, Y. and Polyanskiy, Y.} (2026) Universal priors: solving empirical Bayes via Bayesian inference and pretraining. \textit{Preprint arXiv:2602.15136}.

\bibitem[Deely and Lindley(1981)]{Dee(81)}
\textsc{Deely, J.J. and Lindley, D.V.} (1981). Bayes empirical Bayes. \textit{J. Am. Statist. Assoc.} \textbf{76}, 833--841.

\bibitem[Efron(2014)]{Efr(14)}
\textsc{Efron, B.} (2014). Two modeling strategies for empirical Bayes estimation. \textit{Statist. Sci.} \textbf{29}, 285--301.

\bibitem[Efron(2019)]{Efr(19)}
\textsc{Efron, B.} (2019). Bayes, oracle Bayes and empirical Bayes. \textit{Statist. Sci.} \textbf{34}, 177--201.

\bibitem[Efron and Hastie(2021)]{Efr(21)}
\textsc{Efron. B. and Hastie, T.} (2021). \textit{Computer age statistical inference: algorithms, evidence, and data science}. Cambridge University Press.

\bibitem[Efron and Thisted(1976)]{Efr(76)}  
\textsc{Efron, B. and Thisted, R.} (1976). Estimating the number of unseen species: How many words did Shakespeare know? \textit{Biometrika} \textbf{63}, 435--447.

\bibitem[Favaro and Fortini(2024)]{Fav(24)}
\textsc{Favaro, S. and Fortini, S.} (2024). Quasi-Bayes empirical Bayes: a sequential approach to the Poisson compound decision problem. \textit{Preprint arXiv:2411.07651}.

\bibitem[Favaro and Teh(2013)]{Fav(13)}
\textsc{Favaro, S. and Teh, Y.W.} (2013). MCMC for normalized random measure mixture models. \textit{Statist. Sci.} \textbf{28}, 335--359.

\bibitem[Ferguson(1973)]{Fer(73)}
\textsc{Ferguson, T.S.} (1973). A Bayesian analysis of some nonparametric problems. \textit{Ann. Statist.} \textbf{1}, 209--230.

\bibitem[Fong et al.(2023)]{Fon(23)}
\textsc{Fong, E., Holmes, C. and Walker, S. G.} (2023). Martingale posterior distributions. \textit{J. R. Statist. Soc. B} \textbf{85}, 1357--1391.

\bibitem[Fortini and Petrone(2020)]{For(20)}
\textsc{Fortini, S. and Petrone, S.} (2020). Quasi-Bayesian properties of a procedure for sequential learning in mixture models. \textit{J. R. Statist. Soc. B} \textbf{82}, 1087--1114.

\bibitem[Fortini and Petrone(2025)]{For(25)}
\textsc{Fortini, S. and Petrone, S.} (2025). Exchangeability, Prediction and Predictive Modeling in Bayesian Statistics. \textit{Statist. Sci.} \textbf{40}, 40--67.

\bibitem[Good(1953)]{Goo(53)}
\textsc{Good, I.J.}(1953). The population frequencies of species and the estimation of population parameters. \textit{Biometrika} \textbf{40}, 237-264.

\bibitem[Good and Toulmin(1956)]{Goo(56)}  
\textsc{Good, I.J. and Toulmin, G.H.} (1956). The number of new species, and the increase in population coverage, when a sample is increased. \textit{Biometrika} \textbf{43}, 45--63.

\bibitem[Hahn et al.(2018)]{Han(18)}  
\textsc{Hahn, P.R., Martin, R. and Walker, S.G.} (2018). On recursive Bayesian predictive distributions. \textit{J. Am. Statist. Assoc.} \textbf{113}, 1085--1093.

\bibitem[Ignatiadis and Kankanala(2026)]{Ign(26)}
\textsc{Ignatiadis, N. and Kankanala, S.} (2026). Compound decisions and empirical Bayes via Bayesian nonparametrics. \textit{Preprint  	arXiv:2602.20115}.

\bibitem[Jana et al.(2025)]{Jan(24)}
\textsc{Jana, S., Polyanskiy, Y. and Wu, Y.} (2025). Optimal empirical Bayes estimation for the Poisson model via minimum-distance methods. \textit{Inf. Inference} \textbf{14}, 1--42.

\bibitem[Knoblauch et al.(2022)]{Kno(22)}
\textsc{Knoblauch, J., Jewson, J., and Damoulas, T.} (2022). An optimization-centric view on Bayes’ rule: reviewing and generalizing variational inference. \textit{J. Mach. Learn. Res.} \textbf{23}, 1--109.

\bibitem[Lindsay(1995)]{Lin(95)}
\textsc{Lindsay, B.G.} (1995). \textit{Mixture models: theory, geometry and applications}. NSF-CBMS Regional Conference Series in Probability and Statistics.

\bibitem[Lo(1984)]{Lo(84)}   
\textsc{Lo, A.Y.} (1984). On a class of Bayesian nonparametric estimates. I. Density estimates \textit{Ann. Statist.} \textbf{12}, 351--357.

\bibitem[Mao and Lindsay(2002)]{Mao(02)}   
\textsc{Mao, C.X. and Lindsay, B.G.} (2004). A Poisson model for the coverage problem with a genomic application. \textit{Biometrika} \textbf{89}, 669--682.

\bibitem[Martin(2012)]{Mar(12)}
\textsc{Martin, R.} (2012). Convergence rate for predictive recursion estimation of finite mixtures. \textit{Stat. Probab. Lett.} \textbf{82}, 378--384.

\bibitem[Martin and Ghosh(2008)]{Mar(08)}
\textsc{Martin, R. and Ghosh, J.K.} (2008). Stochastic approximation and Newton’s estimate of a mixing distribution. \textit{Statist. Sci.} \textbf{23}, 365--382.

\bibitem[Martin and Tokdar(2009)]{MarTok(09)}
\textsc{Martin, R. and  Tokdar, S.T.} (2009). Asymptotic properties of predictive recursion: robustness and rate of convergence. \textit{Electron. J. Stat.} \textbf{3}, 1455--1472.

\bibitem[Newton et al.(1998)]{New(98)}
\textsc{Newton, M.A., Quintana, F.A. and Zhang, Y.} (1998). Nonparametric Bayes methods using predictive updating. In \textit{Practical Nonparametric and Semiparametric Bayesian Statistics}, Springer.

\bibitem[Robbins(1951)]{Rob(51)}
\textsc{Robbins, H.} (1951). Asymptotically subminimax solutions of compound decision problems. In \textit{Proceedings of the Second Berkeley Symposium} \textbf{2}, 131--148.

\bibitem[Robbins(1956)]{Rob(56)}
\textsc{Robbins, H.} (1956). An empirical Bayes approach to statistics. In \textit{Proc. Third Berkeley Symp. Math. Statist. Probab.} \textbf{3}, 157--164.

\bibitem[Robbins(1977)]{Rob(77)}
\textsc{Robbins, H.} (1977). Prediction and estimation for the compound Poisson distribution. \textit{Proc. Natl. Acad. Sci. U.S.A.} \textbf{74 }, 2670--2671.

\bibitem[Robbins(1988)]{Rob(88)}
\textsc{Robbins, H.} (1988). The $u,\,v$ method of estimation. In \textit{Statistical Decision Theory and Related Topics IV}. Springer, New York.

\bibitem[Robbins and Zhang(1988)]{Robz(88)}
\textsc{Robbins, H. and Zhang, C.-H.} (1988). Estimating a treatment effect under biased sampling. \textit{Proc. Natl. Acad. Sci. U.S.A.} \textbf{85}, 3670--3672.

\bibitem[Robbins and Zhang(1989)]{Rob(89)}
\textsc{Robbins, H. and Zhang, C.-H.} (1989). Estimating the superiority of a drug to a placebo when all and only those patients at risk are treated with the drug. \textit{Proc. Natl. Acad. Sci. U.S.A.} \textbf{86}, 3003--3005.

\bibitem[Robbins and Zhang(1991)]{Rob(91)}
\textsc{Robbins, H. and Zhang, C.-H.} (1991). Estimating a multiplicative treatment effect under biased allocation. \textit{Biometrika} \textbf{78}, 349--354.

\bibitem[Robbins and Zhang(2000)]{Rob(00)}
\textsc{Robbins, H. and Zhang, C.-H.} (2000). Efficiency of the $u,\,v$ method of estimation. \textit{Proc. Natl. Acad. Sci. U.S.A.} \textbf{97}, 12976--12979.

\bibitem[Rocher et al.(2019)]{Roc(19)}
\textsc{Rocher, L., Hendrickx, J. M., and de Montjoye, Y. A.} (2019). Estimating the success of re-identifications in incomplete datasets using generative models. \textit{Nat. Commun.} \textbf{10}, 3069.

\bibitem[Shen and Wu(2024)]{She(24)}
\textsc{Shen, Y. and Wu, Y.} (2024). Empirical Bayes estimation: When does $g$-modeling beat $g$-modeling in theory (and in practice)? \textit{Preprint arXiv:2211.12692}.

\bibitem[Skinner and Elliot(2002)]{Ski(02)}
\textsc{Skinner, and Elliot, M.J.} (2002). A measure of disclosure risk for microdata. \textit{J. R. Statist. Soc. B} \textbf{64}, 855--867.

\bibitem[Smith and Makov(1978)]{Smi(78)}
\textsc{Smith, A.F.M. and Makov, U.E.} (1978). A quasi-Bayes sequential procedure for mixtures. \textit{J. R. Statist. Soc. B} \textbf{40}, 106--112.

\bibitem[Tebaldi and West(1998)]{Teb(98)}
\textsc{Tebaldi, C. and West, M.} (1998). Bayesian inference on network traffic using link count data. \textit{J. Amer. Statist. Assoc.} \textbf{93}, 557--573.

\bibitem[Vardi(1996)]{Var(96)}
\textsc{Vardi, Y.} (1996). Network tomography: Estimating source-destination traffic intensities from link data. \textit{J. Amer. Statist. Assoc.} \textbf{91}, 365--377.

\bibitem[Zhang(2005)]{Zha(05)}
\textsc{Zhang, C.-H.} (2005). Estimation of sums of random variables: examples and information bounds. \textit{Ann. Statist.} \textbf{33}, 2022--2041.


\end{thebibliography}

\begin{thebibliography}{9}

\bibitem[Chen(2002)]{Chen(02)}
\textsc{Chen, H.F} (2002). \textit{Stochastic approximation and its applications.} Springer New York, NY.

\bibitem[Efron and Hastie(2021)]{Efr(21)}
\textsc{Efron. B. and Hastie, T.} (2021). \textit{Computer age statistical inference: algorithms, evidence, and data science}. Cambridge University Press.

\bibitem[Efron and Thisted(1976)]{Efr(76)}  
\textsc{Efron, B. and Thisted, R.} (1976). Estimating the number of unseen species: How many words did Shakespeare know? \textit{Biometrika} \textbf{63}, 435--447.


\bibitem[Good(1953)]{Goo(53)}
\textsc{Good, I.J.}(1953). The population frequencies of species and the estimation of population parameters. \textit{Biometrika} \textbf{40}, 237-264.

\bibitem[Good and Toulmin(1956)]{Goo(56)}  
\textsc{Good, I.J. and Toulmin, G.H.} (1956). The number of new species, and the increase in population coverage, when a sample is increased. \textit{Biometrika} \textbf{43}, 45--63.

\bibitem[Jana et al.(2025)]{Jan(24)}
\textsc{Jana, S., Polyanskiy, Y. and Wu, Y.} (2025). Optimal empirical Bayes estimation for the Poisson model via minimum-distance methods. \textit{Inf. Inference} \textbf{14}, 1--42.

\bibitem[Lindsay(1995)]{Lin(95)}
\textsc{Lindsay, B.G.} (1995). \textit{Mixture models: theory, geometry and applications}. NSF-CBMS Regional Conference Series in Probability and Statistics.

\bibitem[Mao and Lindsay(2002)]{Mao(02)}   
\textsc{Mao, C.X. and Lindsay, B.G.} (2004). A Poisson model for the coverage problem with a genomic application. \textit{Biometrika} \textbf{89}, 669--682.

\bibitem[Martin and Tokdar(2009)]{MarTok(09)}
\textsc{Martin, R. and  Tokdar, S.T.} (2009). Asymptotic properties of predictive recursion: robustness and rate of convergence. \textit{Electron. J. Stat.} \textbf{3}, 1455--1472.

\bibitem[Zhang(2005)]{Zha(05)}
\textsc{Zhang, C.-H.} (2005). Estimation of sums of random variables: examples and information bounds. \textit{Ann. Statist.} \textbf{33}, 2022--2041.


\end{thebibliography}
\end{document}